\documentclass[11pt]{article}

\usepackage{xr-hyper}
\usepackage{float}
\usepackage{amsmath,amssymb,amsthm,amsfonts,amstext}
\usepackage{lscape}
\usepackage{bm}
\usepackage{tabularx}
\usepackage{longtable}
\usepackage{color}
\usepackage{graphicx}
\usepackage{caption}
\usepackage{subcaption}
\usepackage{mathrsfs}
\usepackage{setspace}
\usepackage{rotating}
\usepackage{multirow}
\usepackage{xcolor}
\usepackage{ulem}
\usepackage{thrmappendix}
\usepackage{threeparttable}
\usepackage{enumitem}
\usepackage{adjustbox}
\usepackage{bbm}
\usepackage{float}
\usepackage[hidelinks]{hyperref}

\usepackage[longnamesfirst]{natbib}
\usepackage{ulem}
\def\X{\mathbf{X}}

\def\x{\mathbf{x}}

\def\ubx{\underline{x}}
\def\ubbx{\underline{\bold{x}}}

\newcommand{\prob}{\mathbb{P}}
\def\E{\mathbb{E}}
\def\Prob{\mathbb{P}}
\newcommand{\supp}{\text{supp}}
\newcommand{\deq}{\overset{d}{=}}
\def\avg{\frac{1}{N} \sum_{i=1}^N}

\theoremstyle{definition}
\newtheorem{theorem}{Theorem}[section]

\newtheorem{corollary}{Corollary}[section]
\newtheorem{proposition}{Proposition}[section]
\theoremstyle{definition}
\newtheorem{remark}{Remark}[section]

\newcommand{\indep}{\perp \! \! \! \perp}

\newcommand{\R}{\mathbb{R}}

\newcommand{\1}{\mathbbm{1}}

\newcounter{IDsec}
\newcommand{\aAssump}{A\arabic{IDsec}}
\newcounter{estsec}
\newcommand{\bAssump}{B\arabic{estsec}}

\newtheoremstyle{theoremSuppressedNumber}{}{}{}{}{\bfseries}{.}{ }{\thmname{#1}\thmnote{ (\mdseries #3)}}
\theoremstyle{theoremSuppressedNumber}
\newtheorem{IDsec_assump}{Assumption \aAssump \addtocounter{IDsec}{1}}

\makeatletter

\@addtoreset{equation}{section}
\makeatletter

\usepackage{enumitem}
\setlist[itemize]{itemsep=0pt, topsep=0pt, partopsep=0pt}
\setlist[enumerate]{itemsep=0pt, topsep=0pt, partopsep=0pt}

\usepackage{amsmath}

\begin{document}

\title{Identification and Estimation of Partial Effects in Nonlinear Semiparametric Panel Models\thanks{Liu: Department of Economics, University of Pittsburgh, \texttt{laura.liu@pitt.edu}. Poirier: Department of Economics, Georgetown University, \texttt{alexandre.poirier@georgetown.edu}. Shiu: Institute for Economic and Social Research, Jinan University, \texttt{jishiu.econ@gmail.com}.  We thank the editor, anonymous guest editor, and referee for constructive comments and suggestions. We also thank St{\'e}phane Bonhomme, Iv{\'a}n Fern{\'a}ndez-Val, Jiaying Gu, Chris Hansen, Bo Honor{\'e}, Louise Laage, Sokbae (Simon) Lee, Ying-Ying Lee, Whitney Newey, Joon Park, Francis Vella,  and seminar participants at Simon Fraser, Harvard/MIT, University of Chicago, OSU, Xiamen University, Cambridge, UCL, Oxford, Penn, FRB Philadelphia, Pittsburgh, USC, Rochester, Indiana University, Princeton, Brown, Rutgers, Taiwan Economics Research conference, China Meeting of the Econometric Society, DMV Econometrics Conference, North American Winter Meeting of the Econometric Society, the Conference in Honor of Jim Powell, Workshop on Recent Advances in Econometrics at University of Glasgow, North American Summer Meeting of the Econometric Society, NBER-NSF CEME Conference for Young Econometricians, Women in Econometrics Conference, and Econometric Society European Meeting for helpful comments and discussions. 
The authors are solely responsible for any remaining errors. Authors Liu and Poirier contributed equally to this work.}}

\author{Laura Liu
\and
Alexandre Poirier 
\and
Ji-Liang Shiu 
}

\maketitle

\newsavebox{\tablebox} \newlength{\tableboxwidth}

 \setstretch{1}

\begin{abstract}
\begin{spacing}{1}

Average partial effects (APEs) are often not point identified in panel models with unrestricted unobserved individual heterogeneity, such as a binary response panel model with fixed effects and logistic errors as a special case. This lack of point identification occurs despite the identification of these models' common coefficients. We provide a unified framework to establish the point identification of various partial effects in a wide class of nonlinear semiparametric models under an index sufficiency assumption on the unobserved heterogeneity, even when the error distribution is unspecified and non-stationary. This assumption does not impose parametric restrictions on the unobserved heterogeneity and idiosyncratic errors. We also present partial identification results when the support condition fails. We then propose three-step semiparametric estimators for APEs, average structural functions, and average marginal effects, and show their consistency and asymptotic normality. 
Finally, we illustrate our approach in a study of determinants of married women's labor supply.

\end{spacing}
\end{abstract}

 \noindent


\baselineskip19.5pt

\newpage
\section{Introduction} \label{sec:intro}

Nonlinear panel models with unobserved individual heterogeneity are commonly used in empirical research. This paper is concerned with panels where the outcome is generated from the nonlinear semiparametric model
\begin{align}\label{eq:baseline}
        Y_{it} = g_t(X_{it}'\beta_0,C_i,U_{it}),
\end{align}
for units $i=1 ,\ldots, N$ and time periods $t=1, \ldots, T$. Here, $X_{it}$ are covariates, $C_i$ are possibly multi-dimensional unobserved individual heterogeneity, and $U_{it}$ are possibly multi-dimensional unobserved idiosyncratic errors. The function $g_t$ is potentially unknown and may vary across time. We assume $N$ is large, but that $T$ is small and fixed, as is the case in many microeconomic datasets. This class of models includes fixed effects, random effects, and intermediate levels of structure on the conditional distribution of $C_i|\X_i$, where $\X_i = (X_{i1},\ldots,X_{iT})'$. A leading example of this class of models in a binary outcome panel model generated by
\begin{align}\label{eq:baseline_binary}
        Y_{it} = \1\left(X_{it}'\beta_0 + C_i - U_{it} \geq 0 \right).
\end{align}
See \cite{Wooldridge2010} Chapter 15.8 for an exposition of such models. We will use this binary response model to illustrate some results from the more general model in \eqref{eq:baseline}.

Identification results for the common parameters $\beta_0$ in \eqref{eq:baseline_binary} are well-known and go back to the work of \cite{Rasch1960} in the case where $U_{it}$ is assumed to be logistic and \cite{Manski1987} when its distribution is unspecified. Identification results for $\beta_0$ in many other special cases of \eqref{eq:baseline} have been derived. However, in this paper we focus on features of the distribution of the potential outcome
\begin{align*}
        Y_{it}(\ubx_t) \equiv g_t(\ubx_t'\beta_0,C_i,U_{it})
\end{align*}
where $\ubx_t$ is in the support of $X_t$. They include the \textit{average structural function} (ASF), \textit{average partial effects} (APE), and \textit{average marginal effects} (AME). First, as studied in \cite{BlundellPowell2003}, the ASF at potential value $\ubx_t$ is the unconditional expectation of potential outcome $Y_{it}(\ubx_t)$. In the binary response model, it is the conditional response probability $\prob(Y_{it} = 1|X_{it} = \ubx_t, C_i = c)$ averaged over the marginal distribution of the unobserved heterogeneity, $F_C$. It is used to assess the average impact of interventions in which the value of $X_{it}$ is manipulated.
Second, the APE is a derivative of the ASF with respect to one covariate, hence measuring the partial effect of this covariate on the conditional response probability averaged over the marginal distribution of $C_i$. Both the ASF and APE can be averaged over a distribution for $\ubx_t$ to evaluate their averages. For example, we can average the ASF or APE across all covariates except for one ``treatment'' variable. Third, the AME measures the impact on the average potential outcome of a marginal increase in a covariate across the entire population. These measures are commonly used to evaluate the causal impact of policies. See \cite{AbrevayaHsu2021} for a survey of various partial effects in panels.

\label{loc:nonid-intro} When $C_i|\X_i$ is unrestricted, i.e., under fixed effects, the ASF, APE, and AME are often not point identified. This can be the case even when the error distribution is known and when $\beta_0$ is point identified: see \cite{DaveziesDHaultfoeuilleLaage2021} who show their partial identification in a binary panel logit model with fixed effects, a model in which $\beta_0$ is point identified. 

The main contribution of this paper is to show the point identification of these features under an index sufficiency assumption on the unobserved heterogeneity. This assumption restricts this conditional distribution to depend on covariates only through $v(\X_i)$, a (multiple) index of $\X_i$. It is related to an assumption of \cite{AltonjiMatzkin2005} and \cite{BesterHansen2009} that they use to show the identification of the \textit{local average response} (LAR).  While the APE averages partial effects over the unconditional distribution of the unobservables, the LAR differs by conditioning on the covariates: see the discussion below and in Section \ref{subsec:model_estimands} for a comparison of these estimands. \label{loc:intro-known-v}Here we assume that the index function is known, or known up to a finite-dimensional parameter. Under our assumption, $v(\X_i)$ acts as a control function that does not require the specification of a first stage or the existence of an instrument. As we discuss in Section \ref{subsubsec:indexassn}, whether $v(\X_i)$ is a suitable control variable with enough variation depends crucially on the panel structure of the data. We also allow for estimated indices of the form $v(\X_i)'\gamma_0$, as in multiple index models \citep{IchimuraLee1991}.\footnote{In general, we can relax the linear index structure by replacing $X_{it}'\beta_0$ and $v(\X_i)'\gamma_0$ with $w(X_{it};\beta_0)$ and $v(\X_i;\gamma_0)$, where $w(\cdot;\cdot)$ and $v(\cdot;\cdot)$ are known functions with unknown finite-dimensional parameters $\beta_0$ and $\gamma_0$. The linear index structure is more commonly used in the literature, so we focus on it in this paper.} As in \cite{ImbensNewey2009}, the support of this index variable plays an important role, which we study in detail. 

Note that the identification results in this paper do not rely on parametric assumptions on the conditional distribution of $C_i|\X_i$, nor on the distribution of $U_{it}$. While the ASF, APE, and AME depend directly on the distribution of $C_i$, which is not specified or identified, these partial effects are identified despite this dependence. Our approach can be viewed as a unified framework for identifying various partial effects under this assumption in a broad class of nonlinear panel models.

Though the primary emphasis of this paper is on the point identification of the partial effects, we also derive the partial identification results when the support condition is not satisfied and establish sharp bounds for the ASF in Section \ref{rmk:supportV}.  The identified set is relatively narrow when deviations from the support condition are small. Moreover, these partial identification results also help shed light on the effects of the index sufficiency and support condition on the identified set.

We then propose three-step semiparametric estimators for the ASF, APE, and AME.  For example, we show the ASF is the partial mean of the conditional expectation of $Y_{it}$ given $(X_{it}'\beta_0,v(\X_i))$, integrated over the marginal distribution of $v(\X_i)$. In the first step, we estimate $\beta_0$ using one of the many available estimators in the literature. In the second step, we estimate the above conditional expectation, replacing the unobserved $X_{it}'\beta_0$ by generated regressor $X_{it}'\widehat{\beta}$. We then use local polynomial regression to recover this conditional mean. In the third step, we average this estimated conditional mean over the empirical distribution of $v(\X_i)$. The APE estimator is analogous, replacing the conditional expectation estimate with an estimate of its derivative, which is obtained directly via the local polynomial regression. The AME is similarly estimated. The estimators are easy to implement, and their convergence rates are fast relative to other nonparametric estimators since, after integrating over the distribution of $v(\X_i)$, the ASF and APE are functions of one-dimensional $X_{it}'\beta_0$. We offer a full treatment of their asymptotic properties in Supplemental Appendix \ref{sec:estimation2}.

In our empirical illustration, we study women's labor force participation using our semiparametric estimator and compare it with a random effects (RE) estimator and a correlated random effects (CRE) estimator. The RE and CRE estimators we consider are commonly used parametric estimators that assume $C_i|V_i$ is Gaussian and that $U_{it}$ is logistic (see the definitions in Section \ref{sec:alt-est}). Compared to the parametric RE/CRE, our semiparametric APE estimates are closer to zero for lower husband's incomes and more negative for higher ones, whereas the parametric RE/CRE estimates show less variation across husband's incomes. Additionally, the effects of the husband's income are no longer significant once we allow for flexibility in the distributions of the unobserved heterogeneity and idiosyncratic errors.

In addition, we  discuss an extension of our identification results to cases with sequentially exogenous $U_{it}$ in Supplemental Appendix \ref{sec:extensions}. This includes models with lagged dependent variables, see, e.g., \cite{HonoreKyriazidou2000}. The associated estimators are similar under strict or sequential exogeneity.

Finally, we also conduct Monte Carlo simulation experiments in Supplemental Appendix \ref{sec:MonteCarlo}. Our results show that the semiparametric estimator yields smaller biases but larger standard deviations, and the former channel tends to dominate when the true distributions of the unobserved heterogeneity and the idiosyncratic errors are non-Gaussian and non-logistic, respectively.

\subsection*{Related Literature}

We now discuss the literature that is closely related to the main text. References to related literature on estimation and dynamic panels can be found in their respective sections in the Supplemental Appendix.

 First, while we focus on functionals of $Y_{it}(\ubx_t)$, our work builds on an extensive literature on the identification and estimation of $\beta_0$ in model \eqref{eq:baseline}. This literature can be subdivided based on its distributional assumptions on $C_i|\X_i$ and those on $U_{it}|C_i,\X_i$. 

In the case where $g_t$ is known and both $C_i|\X_i$ and $U_{it}|C_i,\X_i$ are parametrized, the distribution of $Y_i|\X_i$ is fully parametrized and $\beta_0$ can be estimated via integrated maximum likelihood. The binary outcome case is studied in \cite{Chamberlain1980}. This case includes random effects, where $C_i|\X_i \deq C_i$ and $C_i$ follow a parametric distribution. See Chapter 3 in \cite{Chamberlain1984} and Chapter 15.8 in \cite{Wooldridge2010} for a review of this approach. 

Under fixed effects, the literature on the identification and estimation of $\beta_0$ in binary models with logistic errors goes back to the work of Rasch \citep{Rasch1960,Rasch1961}. Also see \cite{Andersen1970} and \cite{Chamberlain1980}. For general error distributions, \cite{Manski1987} shows the identification of $\beta_0$ in binary panels when $X_{it}$ contains a continuous regressor with support equal to $\R$, and when $U_{it}|C_i,\X_i$ is stationary.\label{loc:Manski-cond-intro} \cite{Abrevaya1999} considers the identification of $\beta_0$ in the model $Y_{it} = g(X_{it}'\beta_0 + C_i + U_{it})$ when $U_{it}$ is nonparametric. The identification argument generalizes the one from \cite{Manski1987} for binary panels. Also see \cite{Abrevaya2000} and \cite{BotosaruMurisPendakur2021} for identification results under weaker assumptions on the link function $g$. These papers mainly focused on the identification of $\beta_0$ and $g(\cdot)$, whereas \cite{BotosaruMuris2022} show a range of point and partial identification results for the distribution of $Y_{it}(\ubx_t)$ under the above assumptions. In contrast, we achieve point identification even when $Y_{it}$ is binary and without assuming stationarity, but we impose restrictions on the conditional distribution of the heterogeneity.

We consider an intermediate restriction on $F_{C_i|\X_i}$: we assume it depends on $\X_i$ only through a known potentially multivariate index $v(\X_i)$.\label{loc:knowv-lit} We do not parametrize the distribution of $C_i|\X_i$ nor restrict how it depends on this index. As mentioned above, our primary focus is on aspects of the distribution of $Y_{it}(\ubx_t)$, such as the ASF, APE, and AME, rather than $\beta_0$.

Second, due to our conditional independence assumption between the heterogeneity and covariates conditional on an index, our work is also related to a large literature on control functions. \citet*{NeweyPowellVella1999} show the identification of structural functions in a triangular model, where a control variable $V_i$ is identified from a first stage. \cite{BlundellPowell2004} consider a binary response model with endogeneity and focus on the identification and estimation of the ASF. \cite{ImbensNewey2009} consider a nonseparable triangular model and, like us, focus on identifying functionals of the structural function, such as the ASF. The multiple index structure we obtain for the conditional expectation $\E[Y_{it}|\X_i]$ is related to the work of \cite{IchimuraLee1991} and \cite{EscancianoJacho-ChavezLewbel2016} among others. For results on the ASF in binary panels, \cite{MaurerKleinVella2011} use a semiparametric maximum likelihood approach and a control function assumption to identify and estimate the ASF. Also see \cite{Laage2020} for a panel data model with triangular endogeneity.

Third, although they focus on a different estimand, the work of \cite{AltonjiMatzkin2005} and \cite{BesterHansen2009} is closely related to ours. In \cite{AltonjiMatzkin2005}, they consider an exchangeability assumption, where $F_{C|X_1,\ldots,X_T}$ is invariant to relabeling of the time indices on the regressors. They then assume that $C_i|X_{it},v(\X_i) \deq C_i|v(\X_i)$ where $v(\X_i)$ are known symmetric functions of $(X_{i1},\ldots,X_{iT})$. They consider a nonparametric outcome equation and show the identification of the LAR, which averages changes in the structural function over the \textit{conditional} distribution of the heterogeneity. Note that the LAR differs from the APE since the former integrates over the conditional distribution of $C_i|X_{it}$ rather than its marginal distribution. We also achieve point identification of the LAR: see Theorem \ref{thm:LAR-AME_id}. Because of the single-index structure of outcome equation $g_t(X_{it}'\beta_0,C_i,U_{it})$, this identification is achieved under weaker support conditions on $\X_i$ and the indices than in their model. We discuss in more detail in Section \ref{subsec:model_estimands} the difference in estimands, and related differences in assumptions on the support of the index are illustrated in the discussion after Theorem \ref{thm:LAR-AME_id}. Moreover, our semiparametric structure allows for much faster rates of convergence for our APE when compared to the rates obtained for the LAR in their nonparametric outcome equation. In particular, their rate of convergence for their LAR estimator decreases with the dimension of $X_{it}$ while the rate of convergence of our APE estimator does not, since it depends on the dimension of $X_{it}'\beta_0$, which is fixed. 

In \cite{BesterHansen2009} they also consider an index sufficiency assumption, where each index is a function of each covariate. Specifically, they have $v(\X_i) = (v_1(\X_{i}^{(1)}),\ldots,v_{d_X}(\X_{i}^{(d_X)}))$, where $d_X$ denotes the dimension of $X_{it}$ , and $\X_i^{(k)}$ denotes a $T\times 1$ vector with the $k$th components of $X_{it}$ for $t=1,\ldots,T$. Their indices $\{v_j(\cdot)\}_{j=1}^{d_X}$ are allowed to be unknown under the separability requirement, while our indices are assumed to be known, or known up to finite-dimensional parameters. 

Other identification approaches in these models have also been proposed. \cite{GrahamPowell2012} consider a correlated random coefficients model and estimate averages of these coefficients, which correspond to the APEs. \cite{HoderleinWhite2012} consider the identification of the LAR for a subpopulation of stayers in a nonseparable model. \cite{BonhommeLamadonManresa2017} consider a discretization of the unobserved heterogeneity as an intermediate assumption between fixed and random effects.

Finally, we also consider partial identification when the support condition fails, and our proof builds on \cite{ImbensNewey2009}. In fixed effects binary response models, \cite{DaveziesDHaultfoeuilleLaage2021} derive bounds for the AME and ASF when $U_{it}$ is assumed to be logistic.
\cite{BotosaruMuris2022} develop partial identification results under weaker assumptions on the link function. \cite{ChernozhukovFernandez-ValHahnNewey2013} derive bounds on the ASF with nonparametric distributions of $C_i|\X_i$ and of $U_{it}$. Also see \cite{ChernozhukovFernandez-ValHoderleinHolzmannNewey2015}. 

The remainder of this paper is organized as follows. In Section \ref{sec:identification} we present the baseline model and provide our main identification results. Section \ref{sec:estimation} briefly describes our proposed estimators for the ASF and APE. 
Section \ref{sec:application} applies our ASF and APE estimators to an empirical illustration of female labor force participation. Finally, Section \ref{sec:conclusion} concludes.  
Appendix \ref{sec:betaID} shows the identification of $\beta_0$ under the index condition, and Appendix \ref{app_sec:proofs} contains the proofs for all propositions and theorems. A Supplemental Appendix is also available online with theoretical results on the asymptotic properties of our proposed estimators, an extension to a dynamic panel data model, a discussion of implementation details, Monte Carlo experiments, and additional tables and figures for the empirical illustration.

\section{Model and Identification} \label{sec:identification}

In this section, we describe the panel model of interest and show the identification of the ASF, APE, and AME under an index sufficiency assumption on the conditional distribution of the heterogeneity. 

\subsection{Model and Estimands}\label{subsec:model_estimands}
Recall the baseline model in equation \eqref{eq:baseline}
\begin{align*}
        Y_{it} &= g_t(X_{it}'\beta_0,C_i,U_{it}), 
\end{align*}
where $i=1,\ldots,N$, $t=1,\ldots,T$. Here, $X_{it} \in \mathcal{X}_t \subseteq \R^{d_X}$ are covariates and $\beta_0 \in \mathcal{B} \subseteq \R^{d_X}$ are unknown parameters. Let $\X_i \in \mathcal{X} \subseteq \R^{T\times d_X}$ denote the observed covariate matrix which has $X_{it}'$ as its $t$th row, where $\R^{a \times b}$ denotes the set of matrices with $a$ rows, $b$ columns, and real entries.  Let $C_i \in \mathcal{C} \subseteq \R^{d_C}$ denote the unobserved individual heterogeneity, and $U_{it} \in \mathcal{U}_t \subseteq \R^{d_U}$ are idiosyncratic errors. Let $Y_i = (Y_{i1},\ldots,Y_{iT})$ denote the vector of outcomes for unit $i$. Note that the outcome function $g_t$ can depend on the time period. For example, this allows for $g_t(X_{it}'\beta_0, C_i, U_{it}) = g(X_{it}'\beta_0 + \delta_t, C_i, U_{it})$, where $\delta_t$ is an additive time-effect, or for $g_t(X_{it}'\beta_0, C_i, U_{it}) = g(X_{it}'\beta_0 + \delta_t C_i, U_{it})$, an interactive effect.

The $i$ subscript is suppressed in the remainder of this section and when there is no confusion. We maintain the following assumptions on the baseline model. Let $\supp(\cdot)$ denote the support of a random vector or variable.

\begin{IDsec_assump}[Model assumptions] \label{assn:model} For each $t = 1,\ldots,T$,\begin{itemize}
        \item [(i)] $Y_t$ is generated according to equation \eqref{eq:baseline};
        \item [(ii)] $U_t \indep \X|C$;
        
        \item [(iii)] $\E[|g_t(\ubx_t'\beta_0,C,U_t)|] < \infty$ for all $\ubx_t'\beta_0 \in \supp(X_t'\beta_0)$.\label{loc:supp}
\end{itemize}
\end{IDsec_assump}

Besides assuming model equation \eqref{eq:baseline} holds, A\ref{assn:model}.(ii) also imposes that unobserved variables $U_t$ are independent of covariates given the unobserved individual heterogeneity. This strict exogeneity assumption rules out the presence of lagged dependent variables in $\X$. We relax this assumption and consider models with lagged dependent variables in Supplemental Appendix \ref{sec:extensions}. 
The relationship between $C$ and $\X$ is unrestricted by A\ref{assn:model}, and this assumption allows for serial correlation and nonstationarity in $U_t$. Lastly, A\ref{assn:model}.(iii) ensures that the ASF is well defined.

Let $Y_t(\ubx_t) \equiv g_t(\ubx_t'\beta_0,C,U_t)$ denote the potential outcome at time $t$ evaluated at covariate value $\ubx_t \in \mathcal{X}_t$. We define the ASF at time $t$ evaluated at $\ubx_t$ by
\begin{align*}
        \text{ASF}_t(\ubx_t) &\equiv \E[Y_t(\ubx_t)] = \E[g_t(\ubx_t'\beta_0,C,U_t)].
\end{align*}
It is the average outcome if $X_t$ were set to $\ubx_t$ in an exogenous manner. The ASF generally differs from the identified conditional expectation $\E[Y_t|X_t = \ubx_t]$ due to the dependence between $C$ and $X_t$, unless $C \indep X_t$, a random effects assumption.

In the binary response model, it can alternatively be defined as a function of the \textit{conditional response probability},  $\prob(Y_t = 1 | X_t = \ubx_t, C = c)$. The ASF is then defined as the conditional response probability integrated over the marginal distribution of the unobserved effect $C$:
\begin{align*}
        \text{ASF}_t(\ubx_t) &= \int_{\mathcal{C}} \prob\left(Y_t = 1 |X_t = \ubx_t, C = c \right) \, dF_C(c).
\end{align*}
By Assumption A\ref{assn:model}.(ii), this definition coincides with ours.

Our second object of interest is the APE, which measures the partial effect of changing one covariate, averaged over the marginal distribution of $C$. If this covariate is continuously distributed, the APE is the derivative of the ASF with respect to this covariate, assuming that the derivative exists. Formally, define the APE of the $k$th element of $\ubx_t \in \mathcal{X}_t$, denoted by $\ubx_t^{(k)}$, as follows:
\begin{align*}
        \text{APE}_{k,t}(\ubx_t) \equiv \frac{\partial \E[Y_t(\ubx_t)]}{\partial \ubx_t^{(k)}}  &= \beta_0^{(k)} \cdot \frac{\partial}{\partial a} \left.\E[g_t(a,C,U_t)]\right|_{a = \ubx_t'\beta_0},
\end{align*}
where $\beta_0^{(k)}$ is the $k$th element of $\beta_0$. 

In the case where $X_t^{(k)}$ is discretely distributed, the APE is the difference between the ASF at two values, which can be interpreted as an average treatment effect. We let
\begin{align*}
        \text{APE}_{k,t}(\ubx_t,  \ubx_{t}^*) &\equiv \E[Y_t(\ubx_{t}^*) - Y_t(\ubx_t)] = \text{ASF}_t(\ubx_{t}^*)- \text{ASF}_t(\ubx_t),
\end{align*}
where $\ubx_{t}^*$ is a vector that differs from $\ubx_t$ in its $k$th position.

Finally, assuming that the necessary derivatives exist, define the \textit{local average response} as follows:\footnote{Note that we can also define an alternative LAR that conditions on covariate values at time $t$ only: 
$\partial \E[Y_t(\ubx_t)|X_t = x_t]/\partial \ubx_t^{(k)} |_{x_t = \ubx_t}$. This alternative LAR can be obtained from $\E[\text{LAR}_{k,t}(\X)|X_t = \ubx_t]$.}
\begin{align*}
        \text{LAR}_{k,t}(\ubbx) &\equiv \left.\frac{\partial \E[Y_t(\ubx_t)|\X = \x]}{\partial \ubx_t^{(k)}} \right|_{\x = \ubbx} = \beta_0^{(k)} \cdot \frac{\partial}{\partial a} \left.\E[ g_t(a,C,U_t)|\X = \ubbx]\right|_{a = \ubx_t'\beta_0},.
\end{align*}
We define the \textit{average marginal effect} as its average over the distribution of $\X$:
\begin{align*}
        \text{AME}_{k,t} &\equiv \E[\text{LAR}_{k,t}(\X)] =  \E\left[ \frac{\partial}{\partial X_t^{(k)}} \E[Y_t|X_t,C]\right].
\end{align*}

In contrast to the LAR, the APE averages this response over the entire population. Thus, the APE is analogous to average treatment effects (ATE) in the causal inference literature, which averages the difference between two potential outcomes over its unconditional distribution; meanwhile, the LAR is analogous to a local treatment effect, where the averaging occurs over the conditional distribution of the heterogeneity given $\X = \ubbx$.\footnote{Specifically, the local average response can be viewed as an average causal response on the treated (ACRT) since it conditions on the subpopulation with covariate values $\ubbx$, see \cite{CallawayGoodman-BaconSantAnna2021}.} Neither estimand is more general since knowledge of the LAR for all $\ubbx \in \supp(\X)$ does not imply knowledge of APEs, and vice-versa.  We later show that the LAR is identified under weaker index support assumptions than the APE. 

The AME averages this local response over the distribution of covariates and essentially measures the impact of a small change in covariate $X_t^{(k)}$ for all units on average outcomes. 

\begin{remark}[Integrated estimands]
It may also be of interest to consider averages of the ASF or APE over certain covariate values. For example, one can consider the APE's average over the marginal distribution of $X_t$, or over the distribution of all covariates except for one:
\begin{align*}
        \widetilde{\text{APE}}_{k,t} &\equiv \int_{\supp(X_t)} \text{APE}_{k,t}(\ubx_t) \; dF_{X_t}(\ubx_t)\\
        \widetilde{\text{APE}}_{k,t}(\ubx_t^{(k)}) &\equiv \int_{\supp(X_t^{(-k)})} \text{APE}_{k,t}(\ubx_t)\;  dF_{X_t^{(-k)}}(\ubx_t^{(-k)})
\end{align*}
where the $(-k)$ superscript denotes removal of the $k$th entry.
\end{remark}
\subsection{Identifying Assumptions} \label{subsec:ASF-APE-Iden}

Without further assumptions, it is generally impossible to point identify these partial effects, even under parametric assumptions on $U_t$. Under fixed effects, the ASF, APE, and AME are generally partially identified.\footnote{ Point identification can be obtained if more structure is assumed, such as random effects: $C \indep \X$. Another example is the set of assumptions in \cite{BotosaruMuris2022}, which include $g_t$ being invertible in $X_t'\beta_0 + C - U_t$, $Y_t$ being continuous, and $(\beta_0,g_t)$ being identified.}

\subsubsection{(Non-)Identification under Fixed Effects}\label{subsubsec:nonidentification}

To fix ideas, we illustrate this identification failure in the binary response model of equation \eqref{eq:baseline_binary} with scalar individual effects, a special case of our general model.

Under fixed effects, $C_i|\X_i$ is unrestricted. Denote by $G_t(\ubx_t'\beta_0,\x)$ the counterfactual conditional probability
\begin{align*}
        G_t(\ubx_t'\beta_0,\x) &\equiv \prob(Y_t(\ubx_t) = 1|\X = \x) =  \int_{\supp(C|\X = \x)} F_{U_t|C}(\ubx_t'\beta_0 + c|c) \, dF_{C|\X}(c|\x).
\end{align*}
Note that we observe the conditional probabilities $\prob(Y_t = 1 | \X = \x) \equiv G_t(x_t'\beta_0,\x)$ for all $\x \in \supp(\X)$ and $t \in \{1,\ldots,T\}$. By the law of total probability, the ASF for covariate value $\ubx_t$ is
\begin{align}
        \text{ASF}_t(\ubx_t) &= \int_{\mathcal{C}} F_{U_t|C}(\ubx_t'\beta_0 + c|c) \, dF_C(c) \notag\\
        &= \int_{\mathcal{X}} G_t(\ubx_t'\beta_0,\x) \, dF_{\X}(\x) \label{eq:ASF_partialID_1}\\
        &= \int_{\mathcal{X}} \underbrace{G_t(\ubx_t'\beta_0,\x) \1(x_t'\beta_0 = \ubx_t'\beta_0)}_{ = \prob(Y_t = 1|\X = \x)\1(x_t'\beta_0 = \ubx_t'\beta_0)} \, dF_{\X}(\x) + \int_{\mathcal{X}} \underbrace{G_t(\ubx_t'\beta_0,\x) \1(x_t'\beta_0 \neq \ubx_t'\beta_0 )}_{\text{ not point identified}} \, dF_{\X}(\x). \label{eq:ASF_partialID_2}
\end{align}
We can see from equation \eqref{eq:ASF_partialID_1} that the ASF is an average over the distribution of $\X$ of conditional probability $G_t(\ubx_t'\beta_0,\X)$. For the ASF to be point identified, we need $G_t(\ubx_t'\beta_0,\x)$ to be identified for all $\x \in \supp(\X)$, but this generally fails since, given $X_t'\beta_0 = \ubx_t'\beta_0$, the support of $\X$ does not equal its marginal support. In equation \eqref{eq:ASF_partialID_2}, the $G_t(\ubx_t'\beta_0,\x)$ where $x_t'\beta_0 \neq \ubx_t'\beta_0$ are counterfactual probabilities that are not point identified from the data since they do not correspond to any conditional probability of $Y_t$ given $\X = \x$. 

Unless restrictions are imposed on the distribution of $C|\X$ or other aspects of the model, this causes the ASF, and therefore the APE too, to be partially identified. In the logit case, \citet*{DaveziesDHaultfoeuilleLaage2021} provide partial identification results for the ASF and AME. If one further relaxes their assumption that the error distribution is known and logistic while retaining their other assumptions, the resulting bounds would become weakly wider,\label{loc:DDL} so point identification cannot be achieved in other binary response models either. In the nonparametric case, \cite{ChernozhukovFernandez-ValHahnNewey2013} obtain bounds on the ASF. In the context of the general nonlinear semiparametric model in equation \eqref{eq:baseline}, we establish sharp bounds on the ASF in Corollary \ref{thm:ASF-bounds-fe} below. 

\subsubsection{Identification of Common Parameters}\label{subsubsec:betaidassn}

While $\beta_0$ is not the object of interest, its identification facilitates the identification of partial effects. In what follows, we also take the identification of $\beta_0$ as given.
\begin{IDsec_assump}[Identification of coefficients] \label{assn:betaid}
        $\beta_0$ is point identified or point identified up to scale.
\end{IDsec_assump}

This assumption can be justified by the fact that $\beta_0$'s identification can be established for many special cases of models \eqref{eq:baseline} satisfying Assumption A\ref{assn:model}. For example, when $Y_t$ is binary, $g_t(X_t'\beta_0,C,U_t) = \1(X_t'\beta_0 + C - U_t \geq 0)$, and $U_t$ follows a standard logistic distribution, \cite{Rasch1960} showed that $\beta_0$ is point identified under minimal assumptions requiring variation in $X_t$ over time, allowing for all regressors to be discrete. Still in the binary outcome model, \cite{Manski1987} showed that $\beta_0$ is identified up to scale when $U_t|C,\X$ is stationary and when $X_t$ includes a continuous regressor with support equal to $\R$.\label{loc:largesupport_comment} Unlike the previous result, this does not require knowledge that $U_t$ follows a logistic distribution. \cite{Zhu2022} recently showed that this identification holds under weaker support assumptions on the regressors. 

Manski's result is generalized to non-binary outcomes in \cite{Abrevaya1999} where he considers $g_t(X_t'\beta_0, C, U_t) = h(C + X_t'\beta_0 - U_t)$, where $h$ is weakly increasing. He also assumes the existence of a regressor with large support and shows $\beta_0$ is point identified up to scale. \cite{Abrevaya2000} and \cite{BotosaruMurisPendakur2021} generalize these results to nonseparable and time-varying models. 

A variety of other identification approaches can also be used. These could include special regressors, as in \cite{HonoreLewbel2002}, but see \citet*{ChenKhanTang2019}\label{loc:ChenKhanTang_citation} who point out that this may not be the case for certain dynamic binary choice panel data models. \cite{Lee1999} and \citet*{ChenSiZhangZhou2017} provide alternative assumptions that yield the identification of $\beta_0$. In Appendix \ref{sec:betaID}, we provide a new approach to point identify $\beta_0$ under the index sufficiency condition (see Assumption A\ref{assn:index} below), where we consider both the baseline model with a known index function as well as the case where the index function is known up to finite-dimensional parameters.

\label{loc:nonid-beta} As mentioned above, the point identification of $\beta_0$ does not imply the point identification of any partial effect as was shown in, for example, \cite{DaveziesDHaultfoeuilleLaage2021} for the binary panel logit model, and in Corollary \ref{thm:ASF-bounds-fe} below  for our nonlinear semiparametric model.

\subsubsection{An Index Assumption}\label{subsubsec:indexassn}

To achieve point identification of partial effects, we consider an index sufficiency restriction that imposes additional structure on the conditional distribution of the heterogeneity. In the example of equation \eqref{eq:ASF_partialID_1}, we assume that $G_t(\ubx_t'\beta_0,\x)$ depends on $\x$ only through known index functions.   
\begin{IDsec_assump}[Index sufficiency] \label{assn:index}
Given $V \equiv v(\X)$, where $v:\R^{T \times d_X} \rightarrow \R^{d_V}$ is known, let $C \mid \X \deq C \mid V$.
\end{IDsec_assump}

This assumption is a correlated random effects assumption that restricts the conditional distribution of $C|\X$ to depend solely on $v(\X)$, which are indices of $\X$. The conditional distribution of $C|v(\X)$ remains nonparametric though. This assumption is similar to Assumption 2.1 in \cite{AltonjiMatzkin2005}. On its own, Assumption \ref{assn:index} holds if $v$ is the identity function, i.e., under fixed effects. However, to identify various partial effects we will impose restrictions on the support of $V$ that rule out fixed effects. These support restrictions vary with the partial effects being considered, so we introduce and discuss them separately: see Theorems \ref{thm:ASF-APE_id} and \ref{thm:LAR-AME_id} in Section \ref{subsec:partial_effects_ID} below.

\subsubsection*{Motivation for Assumption A\ref{assn:index}}

Such an index assumption is considered in \cite{AltonjiMatzkin2005} and \cite{BesterHansen2009}, and can be motivated from several perspectives.

In \cite{AltonjiMatzkin2005}, the exchangeability of $f_{C|\X}(c|x_1,\ldots,x_T)$ in $(x_1,\ldots,x_T)$ is assumed. They consider symmetric polynomials as candidates for the index function, e.g., $v(\X) = \left(\sum_{t=1}^T X_t, \sum_{1\leq t_1 < t_2 \leq T} X_{t_1}X_{t_2} \right)$ when the indices are the first two elementary symmetric functions and $X_t$ is scalar. 
Unlike us, \cite{BesterHansen2009} do not assume $v(\cdot)$ is known, but they do not allow for the indices to be arbitrary functions of $\X$: each component on the index may only depend on one component of $\X_t$. This requires that $T \geq 3$, that covariates are continuously distributed, and that the index $v$ satisfies their separability requirement. Note that if their assumptions are met, we could build on their results to  relax A\ref{assn:index} and assume unknown index functions. The focus of these two papers is also different from ours: they identify the LAR rather than the ASF or APE, and their identification of the LAR is shown for continuous covariates. 

Covariate assignment models in panel data can also be used to find candidate indices. This is explored in \cite{ArkhangelskyImbens2023} where they assume the distribution of $\X|C$ is from an exponential family with a known sufficient statistic. For example, if $(X_1,\ldots,X_T)|C$ are assumed iid Gaussian, then $v(\X) = (\sum_{t=1}^T X_t, \sum_{t=1}^T X_t^2)$ forms a sufficient statistic for $\X$ by the Fisher-Neyman factorization theorem. Therefore, we have that $\X|C,v(\X) \deq \X|v(\X)$, which implies A\ref{assn:index} holds.

In a special case where the indices are time-averages, i.e., $v(\X) = \frac{1}{T}\sum_{t=1}^T X_t$, the index assumption is consistent with $C = \zeta\left(\left(\frac{1}{T}\sum_{t=1}^T X_t\right)'\gamma_0,\eta\right)$ where $\eta \indep \X$ and $\zeta(\cdot,\cdot)$ is any function. This is a relaxation of the specification of the conditional distribution of $C$ given $\X$ in \cite{Mundlak1978} since we do not specify the distribution of $\eta$ nor restrict the functional form of $\zeta(\cdot,\cdot)$. In this specification for $C$, the one-dimensional index $v(\X) = \sum_{t=1}^T X_t'\gamma_0$ also satisfies A\ref{assn:index}, but is unknown due to its dependence on unknown $\gamma_0$. Proposition \ref{prop:betaID_IchimuraLee} in the Appendix shows how the unknown index parameter can be identified building on the work of \cite{IchimuraLee1991} on the identification and estimation of multiple index models.

\subsection{Identification of Partial Effects}\label{subsec:partial_effects_ID}
We can now state our main identification results.
\begin{theorem}\label{thm:ASF-APE_id}
        Let $t \in \{1,\ldots,T\}$, $\ubx_t \in \supp(X_t)$, and let Assumptions A\ref{assn:model}--A\ref{assn:index} hold. Then,
        \begin{enumerate}
                \item $\text{ASF}_t(\ubx_t)$ is point identified from the distribution of $(Y,\X)$ when $\supp(V|X_t'\beta_0 = \ubx_t'\beta_0) = \supp(V)$;
                \item Let the partial derivative of $\text{ASF}_t(\ubx_t)$ with respect to $\ubx_t^{(k)}$ exist. Then, $\text{APE}_{k,t}(\ubx_t)$ is point identified from the distribution of $(Y,\X)$ when $\supp(V|X_t'\beta_0 = u) = \supp(V)$ for all $u$ in a neighborhood of $\ubx_t'\beta_0$.
        \end{enumerate}
\end{theorem}
The APE part of this theorem assumes that $X_t^{(k)}$ is continuously distributed. For discretely distributed $X_t^{(k)}$, the APE is a difference between two ASFs, and its identification is achieved when the two corresponding ASFs are point identified. We omit this case for brevity.

For example, the point identification of the ASF occurs because we can write  it as follows:
\begin{align}
        \text{ASF}_t(\ubx_t) &\equiv \E[Y_t(\ubx_t)]= \int_{\supp(V)} \E[g_t(\ubx_t'\beta_0,C,U_t)|V = v] \, dF_V(v)\notag\\
        &= \int_{\supp(V|X_t'\beta_0 = \ubx_t'\beta_0)} \E[g_t(\ubx_t'\beta_0,C,U_t)|V = v] \, dF_V(v)\notag\\
        &= \int_{\supp(V|X_t'\beta_0 = \ubx_t'\beta_0)} \E[g_t(\ubx_t'\beta_0,C,U_t)|X_t'\beta_0 = \ubx_t'\beta_0,V = v] \, dF_V(v)\notag\\
        &= \int_{\supp(V|X_t'\beta_0 = \ubx_t'\beta_0)} \E[Y_t|X_t'\beta_0 = \ubx_t'\beta_0,V =v] \, dF_V(v)\notag\\
        &= \E[\E[Y_t|X_t'\beta_0 = \ubx_t'\beta_0,V]].\label{eq:ASF_expression}
\end{align}
The second equality follows from iterated expectations and the third from the support assumption in the theorem's statement. The fourth follows from $(C,U_t) \indep X_t'\beta_0|V$, which is implied by Assumptions A\ref{assn:model}.(ii) and A\ref{assn:index}.\label{loc:ASF-intuition} Equation \eqref{eq:ASF_expression} depends only on $\left\{\E[Y_t|X_t'\beta_0 = \ubx_t'\beta_0, V = v] : v \in \supp(V) \right\}$ and on the marginal distribution of $V$, which are both identified from the data. All identification results in this subsection still hold if $V$ is replaced by $V'\gamma_0$ and Assumption A\ref{assn:index_unknown} in the Appendix holds.

\begin{remark}
Note that Theorem \ref{thm:ASF-APE_id} can be used to point identify $\E[m(Y_t(\ubx_t))]$ for any known function $m$ without having to modify any assumptions. In particular, one can identify $F_{Y_t(\ubx_t)}(y) = \E[\1(Y_t(\ubx_t) \leq y)]$, the potential outcomes cdf, by setting $m(a) = \1(a \leq y)$ under the assumptions of Theorem \ref{thm:ASF-APE_id}. Therefore, one can also identify the Quantile Structural Function (QSF),\footnote{See \cite{ImbensNewey2009} or \citet{ChernozhukovFernandez-ValHoderleinHolzmannNewey2015} for example.} $\text{QSF}_t(\tau;\ubx_t) \equiv F_{Y_t(\ubx_t)}^{-1}(\tau)$, whenever the ASF is identified.
\end{remark}

To identify the APE of a continuous regressor, we note that the support assumption implies the ASF is point identified for values of $X_t$ near $\ubx_t$. Since the APE is a derivative of the ASF, we can identify the APE as a limit of finite differences between identified ASFs. Formally, we can write
\begin{align}
        \text{APE}_{k,t}(\ubx_t) = \E\left[\frac{\partial}{\partial \ubx_t^{(k)}}\E[Y_t|X_t'\beta_0 = \ubx_t'\beta_0,V]\right].\label{eq:APE_expression}
\end{align}
All quantities in equation \eqref{eq:APE_expression} are identified, hence the APE is identified. Note that the identification of the ASF and APE bypasses the need to identify $F_C$, the distribution of the heterogeneity. As a side note, in the previous version of this paper \citep{liu2021identification}, we show that $F_C$ is identified under stronger support assumptions on $(X_t'\beta_0,V)$ when outcomes are binary and $C$ is scalar.

We now contrast these with identification results for the LAR and AME.

\begin{theorem}\label{thm:LAR-AME_id}
        Let $t\in\{1,\ldots,T\}$, $\ubbx \in \supp(\X)$, and let Assumptions A\ref{assn:model}--A\ref{assn:index} hold. Then,
        \begin{enumerate}
                \item $\text{LAR}_{k,t}(\ubbx)$ is point identified from the distribution of $(Y,\X)$ when $\E[Y_t(\ubx_t)|\X = \x]$ is differentiable in $\ubx_t^{(k)}$ for $\x = \ubbx$ and when $v(\ubbx) \in \supp(V|X_t'\beta_0 = u)$ for all $u$ in a neighborhood of $\ubx_t'\beta_0$;
                \item $\text{AME}_{k,t}$ is point identified from the distribution of $(Y,\X)$ if the above condition holds for all $\ubbx \in \supp(\X)$ up to a $\prob_\X$-measure zero set.
        \end{enumerate}
\end{theorem}

The following equations help explain how the LAR and AME are identified:
\begin{align}
        \text{LAR}_{k,t}(\ubbx) &= \frac{\partial}{\partial \ubx_t^{(k)}} \E[Y_t|X_t'\beta_0 = \ubx_t'\beta_0, V = v] |_{v = v(\ubbx)}\label{eq:LAR_expression}\\
        \text{AME}_{k,t} &=  \E[\text{LAR}_{k,t}(\X)] = \E\left[\frac{\partial}{\partial X_t^{(k)}} \E[Y_t|X_t'\beta_0, V]\right]. \label{eq:AME_expression}
\end{align}
Note that the condition for the identification of the LAR is weaker than that for the APE. The APE requires that $\supp(V) = \supp(V|X_t'\beta_0 = u)$ for $u$ in a neighborhood of $\ubx_t'\beta_0$, while the LAR requires that $v(\ubbx) \in \supp(V|X_t'\beta_0 = u)$ for $u$ in a neighborhood of $\ubx_t'\beta_0$. This is weaker because $v(\ubbx) \in \supp(V)$ by construction.

These support conditions indirectly but critically rely on the panel structure of the data and on the index structure of $X_t'\beta_0$. We discuss these support conditions below.

\subsubsection*{Discussion of the Support Conditions in Theorems \ref{thm:ASF-APE_id}--\ref{thm:LAR-AME_id}}

The support condition for the APE in Theorem \ref{thm:ASF-APE_id} requires $X_t'\beta_0$ to be continuously distributed in a neighborhood of $\ubx_t'\beta_0$. This allows for some components of $X_t$ to be discretely distributed. 

We do not require $X_t'\beta_0$ to be supported on the entire real line, but for the APE, we do require that the support of the sufficient statistic is independent of the value of $X_t'\beta_0$ in a neighborhood of $\ubx_t'\beta_0$.  
This support assumption is related to the common support assumption of \cite{ImbensNewey2009}, although we only restrict the support of $V|X_t'\beta_0$ rather than the support of $V|X_t$. This is a key benefit of the semiparametric setup where the outcome equation depends on index $X_t'\beta_0$, as the support of $V|X_t'\beta_0$ is by construction a superset of the support of $V|X_t$. As opposed to \cite{ImbensNewey2009}, we do not posit the existence of a first stage or exogenous excluded variables since our indices are functions of $\X$ only. Also see Remark \ref{rmk:unknown V} below.

\cite{AltonjiMatzkin2005} do not consider the identification of the ASF/APE and instead focus on the LAR. To understand the difference in identifying assumptions, in their nonparametric setting, identification of the ASF or APE would require $\supp(v(\X)|X_t  = \ubx_t) = \supp(v(\X))$. This is significantly stronger than our condition whenever more than one covariate is present. To see this, assume $d_X = T = 2$, $X_t^{(1)}$ is continuously distributed on $\R$, and that $X_{t}^{(2)} \in \{0,1\}$ is binary. Let $v(\X) = \sum_{t=1}^2 X_t = (\sum_{t=1}^2 X_t^{(1)}, \sum_{t=1}^2 X_t^{(2)})$. Then, under minimal assumptions, $\supp(v(\X)) = \R \times \{0,1,2\}$ but $\supp(v(\X)|X_1 = \ubx_1) = \R \times \{\ubx_1^{(2)}, \ubx_1^{(2)} + 1\} \neq \supp(v(\X))$. On the other hand, the conditional support of $v(\X)$ given $\{X_t'\beta_0 = x_t'\beta_0\}$ equals $\supp(v(\X))$ when $\beta_0^{(1)} \neq 0$. Therefore, in this example, the ASF/APE will be identified under our assumptions in the semiparametric model, but not in their nonparametric model.

This important condition also has implications on the dimension of $v(\X)$. For example, this condition is violated when $v(\X) = \X$, i.e., no index restrictions are imposed and, equivalently, we have fixed effects. This is because the support of $\X$ does not equal its conditional support given $X_t'\beta_0$: $\supp(\X|X_t'\beta_0 = \ubx_t'\beta_0) \neq \supp(\X)$. On the other hand, if $v(\X) = \sum_{t=1}^T X_t \in \R^{d_X}$, this condition is written as $\supp(\sum_{t=1}^T X_t | X_t'\beta_0 = u) = \supp(\sum_{t=1}^T X_t)$. For example, we can see that this holds in the simple case where $(X_1,\ldots,X_T)$ are jointly normally distributed.

Although not studied here, we note that the support conditions are potentially testable since they only depend on the observed variables $\X$ and identified parameter $\beta_0$. 

Finally, in Section \ref{rmk:supportV} below, we show that while the support condition may not always be warranted, the ASF and APE are partially identified when it fails.

\begin{remark}[Excluded control variable]\label{rmk:unknown V} 
A more general version of A\ref{assn:index} is that we can \textit{identify} a variable $V$ such that $C|\X,V \deq C|V$. This is a control variable assumption, where $V$ may be an unobservable that is not functionally related to $\X$. For example, it could be the residual in a first-stage equation relating $\X$ to some excluded instruments, whose existence we do not assume in this paper. See, for example, \cite{ImbensNewey2009} in the nonseparable cross-sectional case or \cite{Laage2020} for a panel model with triangular endogeneity and control functions. If this control variable satisfies the support conditions in either Theorem \ref{thm:ASF-APE_id} or \ref{thm:LAR-AME_id}, the corresponding identification result still applies provided that $U_t \indep \X|(C,V)$, a modification of A\ref{assn:model}.(ii). As for estimation, if $V$ is identified from a first-stage equation, we should substitute $\widehat{V}_i$ for $V_i$, where $\widehat{V}_i$ is a suitable estimator for the control variable. This additional generated regressor's impact on the limiting distribution would then have to be taken into account. In this paper, we focus on the case where no such $V$ is observed or identified from a first-stage model, and instead where $V$ is an index of $\X$.
\end{remark}

\subsection{Relaxing Support Assumptions and Partial Identification} \label{rmk:supportV}
The validity of the support assumptions in Theorems \ref{thm:ASF-APE_id} and \ref{thm:LAR-AME_id} depends intricately on the support of $\X$, and on the considered value $\ubx_t$. For a given $\ubx_t$, these assumptions may fail. When they do, we can show the ASF, APE, LAR, and AME are partially identified instead. These partial identification results could help us better understand how the index sufficiency and support condition affect the identified set. In this section we focus on the ASF and the related APE.

Let $\mathcal{V}_t(u) = \supp(V|X_t'\beta_0 = u)$, and assume that $\mathcal{V}_t(\ubx_t'\beta_0) \subsetneq \mathcal{V} \equiv \supp(V)$, where $\subsetneq$ is used to denote a proper subset. Then, the conditional expectation $\E[Y_t|X_t'\beta_0 = \ubx_t'\beta_0, V = v]$ is identified for all $v \in \mathcal{V}_t(\ubx_t'\beta_0)$. Therefore, building on Theorem 4 in \cite{ImbensNewey2009}, we obtain the following sharp bounds on the ASF.
\begin{theorem}\label{thm:ASF-bounds}
        Let $t \in \{1,\ldots,T\}$, $\ubx_t \in \supp(X_t)$, and let Assumptions A\ref{assn:model}--A\ref{assn:index} hold. Also, let $g_t(\ubx_t'\beta_0,C,U_t) \in [\underline{g},\overline{g}]$ with probability 1.\footnote{The point identification of $\beta_0$  could require additional assumptions, which in turn may further sharpen the bounds. We do not incorporate these potential additional assumptions as they are model-specific and beyond the current general setup of the paper. Also, we allow for unrestricted $g_t$ with only boundedness required. Imposing additional conditions on $g_t$, such as monotonicity and parametric distribution, could lead to narrower bounds.} Then, the identified set for $\text{ASF}_t(\ubx_t)$ is
\begin{align}
        \mathcal{A}_t(\ubx_t) &=  \left[\int_{\mathcal{V}_t(\ubx_t'\beta_0)} \E[Y_t|X_t'\beta_0 = \ubx_t'\beta_0, V = v] \, dF_{V}(v)  + \underline{g} \cdot \prob(V \in \mathcal{V} \setminus \mathcal{V}_t(\ubx_t'\beta_0)),\right.\notag\\
        &\left. \qquad \qquad \int_{\mathcal{V}_t(\ubx_t'\beta_0)} \E[Y_t|X_t'\beta_0 = \ubx_t'\beta_0, V = v] \, dF_{V}(v) + \overline{g}\cdot\prob(V \in \mathcal{V} \setminus \mathcal{V}_t(\ubx_t'\beta_0))\right].\label{eq:ASF-bounds}
\end{align}
\end{theorem}

In the case where $Y_t$ is binary, $[\underline{g},\overline{g}] = [0,1]$ and the bounds take on a simpler form. The width of these bounds depends only on $\overline{g} - \underline{g}$ and the probability that $V$ falls outside of $\mathcal{V}_t(\ubx_t'\beta_0)$, which is small if $V$ is continuously distributed and the measure of $\mathcal{V} \setminus \mathcal{V}_t(\ubx_t'\beta_0)$ is close to zero. Hence, small violations of the support condition yield a narrow identified set.

For fixed effects as in Section \ref{subsubsec:nonidentification}, $C|\X$ is unrestricted, and equivalently, $V=\X$. The support condition fails because $\mathcal{V}=\supp(\X)\neq\supp(\X|X_t'\beta_0 = \ubx_t'\beta_0)=\mathcal{V}_t(\ubx_t'\beta_0),$ and thus the ASF cannot be point identified in general. The following corollary characterizes the sharp bounds on the ASF under fixed effects.
 
\begin{corollary}\label{thm:ASF-bounds-fe}
Let the conditions in Theorem \ref{thm:ASF-bounds} hold. \textit{Under fixed effects}, i.e., \(V=\X \), the identified set for $\text{ASF}_t(\ubx_t)$ is
\begin{align}
        &  \left[\int_{\supp(\X|X_t'\beta_0 = \ubx_t'\beta_0)} \E[Y_t|\X = \x] \, dF_{\X}(\x)  + \underline{g} \cdot \prob(X_t'\beta_0 \neq \ubx_t'\beta_0),\right.\notag\\
        &\left. \qquad \qquad \int_{\supp(\X|X_t'\beta_0 = \ubx_t'\beta_0)} \E[Y_t|\X = \x] \, dF_{\X}(\x) + \overline{g}\cdot\prob(X_t'\beta_0 \neq \ubx_t'\beta_0)\right].\label{eq:ASF-bounds-fe}
\end{align}
\end{corollary}
Comparing the bounds in Theorem \ref{thm:ASF-bounds} and Corollary \ref{thm:ASF-bounds-fe}, we see that with index sufficiency but no support conditions, the bounds are as specified in equation \eqref{eq:ASF-bounds}; when we further drop the index sufficiency condition, these bounds become even wider as in equation \eqref{eq:ASF-bounds-fe}. For example, the identified set in \eqref{eq:ASF-bounds-fe} equals the trivial set $[\underline{g},\overline{g}]$  whenever $X_t'\beta_0$ is continuously distributed: these bounds are informative only when $\prob(X_t'\beta_0 = \ubx_t'\beta_0) > 0$. This implies that ASF bounds under fixed effects can be uninformative even if $\beta_0$ is point identified. In the end, the difference between the point identification and the identified set in \eqref{eq:ASF-bounds} shows the importance of the support condition, and the difference between the identified sets in \eqref{eq:ASF-bounds} and \eqref{eq:ASF-bounds-fe} highlights the role of the index sufficiency condition.

The ASF bounds can be used to construct sharp bounds on the APE with discrete covariates. Let $\underline{\text{ASF}}_t(\ubx_t)$ and $\overline{\text{ASF}}_t(\ubx_t)$ denote the lower and upper bounds of $\mathcal{A}_t(\ubx_t)$, we have
\begin{align*}
        \text{APE}_{k,t}(\ubx_t,\ubx_{t}^*) &\in \left[\underline{\text{ASF}}_t(\ubx_{t}^*) - \overline{\text{ASF}}_t(\ubx_t), \quad\overline{\text{ASF}}_t(\ubx_{t}^*) - \underline{\text{ASF}}_t(\ubx_t)\right].
\end{align*}
To obtain bounds on the APE for a continuous covariate, bounds on $\frac{\partial}{\partial a}\E[g_t(a,C,U_t)|C=c]$ are needed. In the case where $\frac{\partial}{\partial a}\E[g_t(a,C,U_t)|C=c] \in [\underline{g}',\overline{g}']$ and $\beta_0^{(k)}>0$, the APE bounds are given by\footnote{We leave a formal proof of the APE bounds' sharpness for future work.}
\begin{align*}
        \text{APE}_{k,t}(\ubx_t)&\in \left[\underline{\text{APE}}_{k,t}(\ubx_t),\overline{\text{APE}}_{k,t}(\ubx_t)\right]\\
        &\equiv \left[\int_{\mathcal{V}_t(\ubx_t'\beta_0)} \frac{\partial}{\partial \ubx_t^{(k)}} \E[Y_t|X_t'\beta_0 = \ubx_t'\beta_0, V = v] \, dF_{V}(v)  + \underline{g}'\beta_0^{(k)} \; \prob(V \in \mathcal{V} \setminus \mathcal{V}_t(\ubx_t'\beta_0)),\right.\\
        &\left. \qquad \qquad \int_{\mathcal{V}_t(\ubx_t'\beta_0)} \frac{\partial}{\partial \ubx_t^{(k)}} \E[Y_t|X_t'\beta_0 = \ubx_t'\beta_0, V = v] \, dF_{V}(v) + \overline{g}'\beta_0^{(k)} \; \prob(V \in \mathcal{V} \setminus \mathcal{V}_t(\ubx_t'\beta_0))\right].
\end{align*}
The bounds are reversed if  $\beta_0^{(k)}<0$. With binary outcomes and under the assumption that $U_t \indep C$, the partial derivative $\frac{\partial}{\partial a}\E[g_t(a,C,U_t)|C=c]$ is the density $f_{U_t}$. Its lower bound is trivially 0, but it is harder to postulate an upper bound when $f_{U_t}$ is unrestricted. In the specific case of logit models, however, this density attains a maximal value of 1/4 at the origin. Therefore, substituting $[\underline{g}',\overline{g}'] = [0,1/4]$ yields bounds on the APE in binary panel logit when common support fails.

Bounds for the LAR and AME can be similarly obtained. The estimation methods we provide below for the point identified ASF, APE, or AME can be adapted to estimate these bounds under support condition violations.

\section{Estimation} \label{sec:estimation}
In this section we briefly propose estimators for the ASF, APE, LAR, and AME. Detailed asymptotic results can be found in Supplemental Appendix \ref{sec:estimation2}. The estimators we construct are sample analogs of \eqref{eq:ASF_expression} for the ASF, and of \eqref{eq:APE_expression} for the APE. We assume throughout that we have access to a random sample $\{(Y_i,\X_i)\}_{i=1}^N$.

All the partial effects estimators are obtained in three steps. We first obtain a consistent estimator $\widehat{\beta}$ of the common parameters using one of the many approaches proposed in the related literature. In the second step, we nonparametrically estimate the conditional expectation $h(u,v) \equiv \E[Y_t|X_t'\beta_0 = u, V = v]$ using a local polynomial regression of $Y_t$ on generated regressor $X_t'\widehat{\beta}$ and $V$. Local polynomial regression naturally yields estimates of the derivatives of $h$, which are used in estimating the APE, LAR, and AME. In the final step, we average features of the local polynomial regression estimates over the empirical distribution of $V_i$, i.e., a partial mean structure.

For example, the ASF estimator is constructed analogously to equation \eqref{eq:ASF_expression} above. We average this conditional mean over the empirical marginal distribution of $V_i$ to obtain the ASF estimator:
\begin{align*}
        \widehat{\text{ASF}}_t(\ubx_t) &= \avg \widehat{h}(\ubx_t'\widehat{\beta},V_i ;\widehat{\beta})\hat\pi_{it},
\end{align*}
where $\hat\pi_{it}$ is a trimming function that regularizes the behavior of the estimator. See Supplemental Appendix \ref{sec:estimation2} for more details on this function. The APE is obtained similarly, replacing the estimated $h$ function by the estimate of its partial derivative with respect to its first component, denoted as $h_u$:
\begin{align*}
        \widehat{\text{APE}}_{k,t}(\ubx_t) &= \widehat{\beta}^{(k)} \cdot \avg \widehat{h}_u(\ubx_t'\widehat{\beta},V_i ;\widehat{\beta})\hat\pi_{it}.
\end{align*}
Finally, analogously to equations \eqref{eq:LAR_expression} and \eqref{eq:AME_expression}, we can define estimators for the LAR and AME:
\begin{align*}
        \widehat{\text{LAR}}_{k,t}(\ubbx) &= \widehat{\beta}^{(k)} \cdot \widehat{h}_u(\ubx_t'\widehat{\beta},v(\ubbx);\widehat{\beta})\\
        \widehat{\text{AME}}_{k,t} &= \avg \widehat{\text{LAR}}_{k,t}(\X_i)\hat\pi_{it} = \widehat{\beta}^{(k)} \cdot \avg \widehat{h}_u(X_{it}'\widehat{\beta},v(\X_i);\widehat{\beta})\hat\pi_{it}.
\end{align*}

In Supplemental Appendix \ref{sec:estimation2}, we provide conditions on the convergence rate of $\widehat{\beta}$ and the bandwidth, as well as on the order of the local polynomial regression, that allow us to establish the consistency and asymptotic normality of the ASF and APE estimators. We leave a complete asymptotic analysis of the AME estimator for future work.
It is worth noting that these estimators exhibit fast convergence rates compared to other nonparametric estimators, which can be attributed to the fact that upon integrating over the distribution of $v(\X_i)$, the ASF and APE are essentially functions of one-dimensional $X_{it}'\beta_0$. Their convergence rates do not depend on the dimension of $\X_i$, which is $T \times d_X$.   For example, when the index is one-dimensional, the ASF's and APE's rates of convergence are similar to the standard rates of convergence of univariate nonparametric kernel regression estimators, which are fast within the class of nonparametric estimators. In particular, we show that the ASF can converge at a rate faster than $N^{2/5}$ when the index is one-dimensional.

\section{Empirical Illustration} \label{sec:application}

In this section we compare the performance of our proposed estimators with that of commonly used alternative estimators.

\subsection{Alternative Estimators}\label{sec:alt-est}
The alternative estimators we consider are a parametric random effects (RE) and a correlated random effects (CRE) estimator. See, for example, \cite{Wooldridge2010}. Both assume a standard logistic distribution for the error term $U_t$, but differ in their specifications for the distribution of individual effects $C$. For the RE, 
\[
       C \sim \mathcal{N}(\mu_c,\sigma_c^2)
\] 
and is independent of $V$. For the CRE, 
\[
       C|V\sim \mathcal{N}(\mu_{c0}+\mu_{c1}'V,\sigma_c^2).
\] 
Then, the CRE is equivalent to an augmented RE with $V$ being additional regressors. As is standard in the literature, we use the Maximum Likelihood Estimator (MLE) to jointly estimate $\beta_0$ and the distribution parameters $(\mu_c,\sigma_c^2)$ or $(\mu_{c0},\mu_{c1},\sigma_c^2)$. 

In the same spirit as the semiparametric estimator in Section \ref{sec:estimation}, we allow the marginal distribution of $V$ to be unrestricted. The conditional expectation of the binary outcome and its derivative are calculated based on the MLE estimates, and the ASF and APE are obtained by averaging out $V$.

\subsection{Background and Specification}\label{sec:app-back-spec}
In this empirical illustration, we examine women's participation in the labor market using our semiparametric approach. See the handbook chapter by \cite{killingsworth1986female} for an extensive review of the literature on female labor supply. For illustrative purposes, our analysis is based on the static setup of \cite{Fernandez-Val2009}, where covariates $X_{t}$ include numbers of children in three age categories, log husband's income, a quadratic function of age, as well as time dummies.\footnote{\cite{Fernandez-Val2009} proposes bias-corrected estimators of marginal effects when $T$ is large and when $U_{t}$ follows a normal distribution. \cite{CharlierMelenbergSoest1995} and \cite{ChenSiZhangZhou2017}, among others, also considered female labor force participation in their empirical applications. They used similar model specifications, but most of these papers focused on the estimation of common parameters $\beta_0$ instead of the ASF or APE.} 

The sample consists of $N=1461$ married women observed for $T=9$ years from the PSID between 1980--1988. We use the dataset kindly made available on Iv{\'a}n Fern{\'a}ndez-Val's website, originally sourced from Jes{\'u}s Carro. In the Supplemental Appendix, we plot the distributions of the covariates in Figure \ref{fig:app-dist-obs}, and summarize the corresponding descriptive statistics in Table \ref{tbl:app-descriptive-stats}. Roughly 45\% of the women in the sample always participated in the labor market, less than 10\% never participated, and around 45\% changed their status during the sample period. Movers tended to be younger and have more children in all children's age categories. Never participants were relatively uniformly distributed between ages 30 and 50, whereas the women in other subgroups were generally younger. All subgroups exhibited heavy tails in log husband's income. 

The unobserved individual effects $C$ could be interpreted as an individual's willingness to work. In the benchmark specification, we construct indices $V$ based on the initial values of the covariates \(X_{i1}\). Women's ages and numbers of children are discrete variables, and we consider a cell-by-cell analysis.\footnote{For a more comprehensive empirical analysis, one could handle discrete index variables using a discrete kernel as suggested in \cite{LiRacine2004}, which would be outside the scope of the current empirical illustration.} These covariates generate over 1000 cells in this sample, and some cells do not contain sufficient observations to use a semiparametric estimator within them. Therefore, we collapse the discrete index variables as follows. First, we sum the number of children under 18 in the initial period, then categorize this number into a low, medium, or high group based on the 33rd and 67th percentiles. Similarly, we collapse the initial age into a binary indicator based on the median. This coarsening scheme results in 6 cells, with a range of 156 to 314 observations per cell. Thus, we have three index variables: a trinary fertility variable, a binary age variable, and a continuously distributed average log husband's income. The number of continuous index variables is $d_V=1$. 

Various robustness checks regarding alternative choices of \(V_i\) (e.g., constructed from \(X_{i1}\) or \(\overline X_i=\frac 1 T \sum_t X_{it}\)), alternative estimators (e.g., multiple indices and local logit), and alternative coarsening schemes are explored in Supplemental Appendix \ref{sec:appendix-app} as well as the previous version of this paper \citep{liu2021identification}. The semiparametric estimator is generally robust with respect to these variations. 

\subsection{Results}

Table \ref{tbl:app-param} reports the estimated common coefficients on key covariates. Figure \ref{fig:app-dtime} in the Supplemental Appendix also plots the estimated coefficients on time dummies, which capture the time-variation in aggregate participation rates. We see that women are more inclined to withdraw from the labor force when they have more children, especially younger ones, and when their husbands earn a higher income.
Compared to the RE and CRE, the flexible smoothed maximum score estimator provides slightly larger (in magnitude) estimates with larger standard errors.  

\begin{table}[t]
        \caption{Estimated $\beta_0$ - Female Labor Force Participation}
        \label{tbl:app-param}
        \begin{center}
    \begin{tabular}{lrr@{\hskip .2in}rr@{\hskip .2in}rr} \\  \hline \hline
   & \multicolumn{2}{c}{Smoothed Max.\ Score}  & \multicolumn{2}{c}{RE}  & \multicolumn{2}{c}{CRE}\\
 &  \textcolor[rgb]{1,1,1}{text}$\widehat\beta\hspace{.22in}$   &      \textcolor[rgb]{1,1,1}{text}SD    &       \textcolor[rgb]{1,1,1}{texttext}$\widehat\beta\hspace{.22in}$   &      \textcolor[rgb]{1,1,1}{text}SD   &       \textcolor[rgb]{1,1,1}{texttext}$\widehat\beta\hspace{.22in}$   &      \textcolor[rgb]{1,1,1}{text}SD  \\ \hline                                 \textit{Children 0--2}       &      \textit{ -1\scriptsize\textcolor[rgb]{1,1,1}{***}}      &       \textit{0}       &      \textit{-1\scriptsize\textcolor[rgb]{1,1,1}{***}}      &      \textit{ 0}    &      \textit{ -1\scriptsize\textcolor[rgb]{1,1,1}{***}}      &      \textit{ 0}    \\                                                      
Children 3--5        &       -0.83\scriptsize***   &       0.18    &       {-0.60}\scriptsize***   &       0.08    &       {-0.60}\scriptsize***   &       0.08    \\                                                      
Children 6--17       &       -0.19\scriptsize\textcolor[rgb]{1,1,1}{***}   &       0.17    &       {-0.19}\scriptsize***   &       0.06    &       {-0.17}\scriptsize***   &       0.06    \\                                                      
Log Husband's Income      &       {-0.54}\scriptsize**\textcolor[rgb]{1,1,1}{*}   &       0.25    &       {-0.38}\scriptsize***   &       0.08    &       {-0.34}\scriptsize***   &       0.10    \\                                                      
Age/10     &       3.45\scriptsize*\textcolor[rgb]{1,1,1}{**}    &       1.88    &       {2.34}\scriptsize***    &       0.64    &       {2.63}\scriptsize***    &       0.68    \\                                                      
$\text{(Age/10)}^2$  &       {-0.51}\scriptsize***   &       0.18    &       {-0.35}\scriptsize***   &       0.08    &       {-0.37}\scriptsize***   &       0.08    \\

\hline
    \end{tabular}       
\end{center}
        {\footnotesize {\em Notes:} Standard deviations are calculated via the bootstrap. Significance levels are indicated by *: 10\%, **: 5\%, and ***: 1\%.  The first row follows from scale normalization $|\widehat\beta^{(1)}|=1$, and we rescale the RE and CRE estimates to ensure comparability across estimators. $\widehat\beta^{(1)}$ is negative in all bootstrap samples for all three estimators so, after rescaling, their bootstrap standard deviations all equal to 0. Since the support of $\widehat\beta^{(1)}$ is $\pm 1$, we do not put asterisks in the first row.}
\setlength{\baselineskip}{4mm}
\end{table}

In our empirical example, we focus on the effects of the husband's income, which may affect the wife's reservation wage. We select evaluation points $\ubx_t$ such that the log husband's income ranges from its 20th to 80th quantiles, and other variables are equal to their medians. These choices correspond to a hypothetical woman who is 35 years old, has 0 children between 0 and 2, 0 children between 3 and 5, 1 child between 6 and 17, and whose husband's income ranges from \$21K to \$55K. All time dummies are set to zero in this counterfactual.

\begin{figure}[t]
        \caption{Estimated ASF and APE - Female Labor Force Participation}
        \label{fig:app-asf-ape}
        \begin{center}
        \includegraphics[width=.8\textwidth]{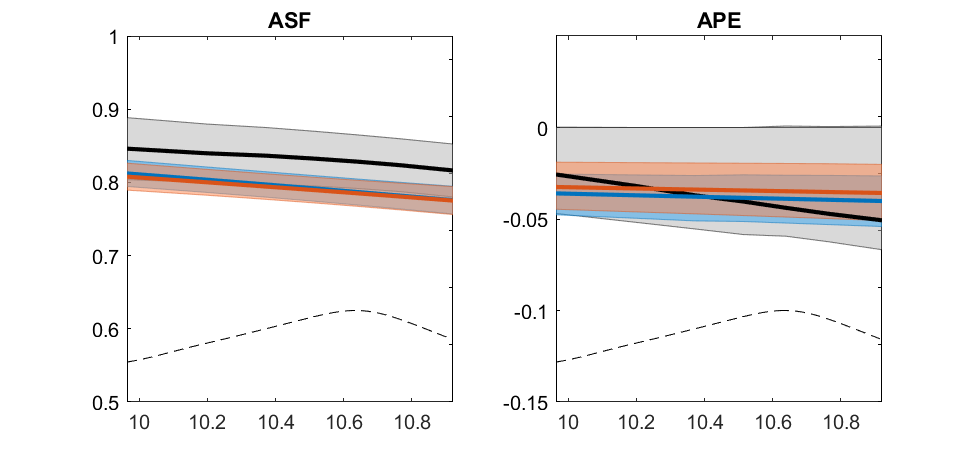}
        \end{center}
        {\footnotesize {\em Notes:} The x-axes are potential values of log husband's income. The black/blue/orange solid lines represent point estimates of the ASF and APE using the semiparametric/RE/CRE estimators. The bands with corresponding colors indicate the 90\% bootstrap confidence intervals. The thin dashed lines at the bottom of both panels show the distribution of log husband's income.}\setlength{\baselineskip}{4mm}                                     
\end{figure}

Figure \ref{fig:app-asf-ape} shows estimates of the ASF and APE across $\ubx_t$ together with the 90\% bootstrap confidence intervals based on 500 bootstrap samples.\footnote{For the ASF, all bootstrap estimates are between 0 and 1, and so is the symmetric percentile-$t$ confidence band based on bootstrap standard deviations. For the APE, the smoothed maximum score in the first step requires monotonicity, i.e., $\frac{d}{d u}\prob(Y_t = 1|X_t'\beta_0 = u)|_{u = \ubx_t'\beta_0} \ge 0$. In the bootstrap, this constraint occasionally binds, so we censor it at zero and employ the percentile bootstrap to account for the possible non-standard distribution due to censoring. Note that in principle, the bootstrap band for the APE could still contain positive values since the estimated coefficient for the log husband's income could be positive in some bootstrap samples. However, this incidence is very rare in our empirical example.} For the ASF, all point estimates are downward sloping with respect to the husband's income. The semiparametric estimator yields slightly higher participation probabilities compared to the RE and CRE. 

For the APE, the semiparametric estimates are closer to zero for lower husband's incomes and more negative for higher ones, while their RE and CRE counterparts are rather flat. Note that for continuous $\ubx_t^{(k)}$,
\[\text{APE}_{k,t}(\ubx_t)=\beta_0^{(k)}\cdot f_{U_t-C}\left(\ubx_t'\beta_0\right)=\beta_0^{(k)} \cdot \int_{\mathcal{C}}  f_{U_t}(\ubx_t'\beta_0 + c) \, dF_C(c),\]
where $f_{U_t-C}$ denotes the pdf of $U_{t}-C$, i.e., a convolution of $-C$ and $U_{t}$. Thus, the slope of the APE with respect to $\ubx_t^{(k)}$ reflects the shapes of $f_C$ and $f_{U_t}$ as well as the magnitude of $\beta_0^{(k)}$. In this sense, the flatter APE profile with respect to the husband's incomes in the RE and CRE could be due to the following three sources: (i) The RE and CRE feature a Gaussian $f_{C|V}$ with estimated mean and variance. The estimated Gaussian variance could be fairly large to accommodate potential non-Gaussian heterogeneity in $C|V$, and the resulting $\widehat f_{U_t-C}$ could be flatter (around the peak) than the true distribution.
(ii) The RE and CRE assume a logistic $f_{U_t}$, which may deviate from the true data generating process.
(iii) The smaller magnitudes of $\widehat\beta_0^{(k)}$ for RE and CRE could be due to misspecification of the distributions of $U_t$ and $C$ and, in turn, further lead to a milder slope of the APE profile, though the discrepancy in the coefficients alone does not fully account for the differences in the slopes of the APE profiles.
In contrast, the semiparametric estimator does not require the parametrization of $f_{C|V}$ or $f_{U_t}$, thus reducing potential biases due to misspecification.

Moreover, our flexible semiparametric estimator, which does not impose parametric assumptions on the distributions of $C$ or $U_t$, yields statistically insignificant APEs with respect to husband's incomes. This finding raises concerns about the potential bias in the highly significant APEs by RE and CRE estimators due to their parametric restrictions. The insignificant APEs are consistent with the empirical observation that married women's labor supply choices became less sensitive to their husbands' income around 1980 when baby boomers started constituting a larger portion of the labor force, and both partners contribute to housework and earnings more equally. Hence fewer married women were at the margin of labor force participation that could be nudged by temporary fluctuations in husbands' income.

\section{Conclusion} \label{sec:conclusion}
The distributions of the unobserved individual heterogeneity and the idiosyncratic errors play a crucial role in identifying partial effects in nonlinear panel models. In this paper, we first show the identification of the ASF, APE, and AME in a nonlinear semiparametric panel model with potentially unspecified distributions of the unobserved heterogeneity and idiosyncratic errors. To achieve point identification, we assume that units with the same value of the index $V$ have correspondingly similar distributions of their unobserved heterogeneity $C$.  We also establish partial identification results when the support condition fails. We then develop three-step semiparametric estimators for the ASF and APE, and show their consistency and asymptotic normality in the Supplemental Appendix. Finally, we illustrate our semiparametric estimator in a study of determinants of women's labor supply. 

In Supplemental Appendix \ref{sec:extensions} we provided an identification result that applies to dynamic panel models. Extending our results to a broader class of dynamic models would be an interesting direction for further research. 
\label{loc:end-main-text}

\clearpage
\begin{singlespace}
\bibliographystyle{econometrica}
\bibliography{bibfile111}
\end{singlespace}

\newpage
\appendix

\section*{Appendix}
This appendix is organized as follows. Appendix \ref{sec:betaID} provides additional identification results for $\beta_0$ under the index condition.   Appendix \ref{app_sec:proofs} contains the proofs for the propositions and theorems in Sections \ref{sec:identification} and Appendix \ref{sec:betaID}.

\section{Identification of $\beta_0$ under the Index Condition}\label{sec:betaID}

The works cited in Section \ref{subsubsec:betaidassn} do not use the index sufficiency in Assumption A\ref{assn:index} to obtain the identification of $\beta_0$. However, it is possible to use A\ref{assn:index} to show $\beta_0$'s identification under minimal additional assumptions on the model. Therefore, Assumption A\ref{assn:betaid} may hold as a consequence of A\ref{assn:index}. In what follows, we first examine the baseline model with a known index function, and then explore the case where the index function is known up to finite-dimensional parameters.
\subsection{Index = $v(\X)$}
Let us first consider the baseline scenario with a known index function $v(\X)$. 
\begin{proposition}\label{prop:betaindexID}
Let Assumptions A\ref{assn:model} and A\ref{assn:index} hold. Suppose there exists $s,t \in \{1,\ldots,T\}$ such that $s\neq t$, $\Psi_s(u,v) \equiv \E[Y_s|X_s'\beta_0 = u, V=v]$ is differentiable with $\frac{\partial}{\partial u} \Psi_s(u,v) \neq 0$, and that $W_t(\x) \equiv \frac{\partial}{\partial x_t} v(\x)'$ has rank $d_V$. Normalize $\beta_0 = (1,\tilde\beta_0)$, where $\widetilde\beta_0 \in \R^{d_X-1}$. Then, $\widetilde\beta_0$ is identified.
\end{proposition}
This result shows that $\beta_0$ can be identified (up to scale) with as few as two time periods, assuming that derivatives of $v(\x)'$ have rank $d_V$. This rank condition requires $d_V \leq d_X$, thus that the indices are of lower dimension than $\X$, the dimension of which is $T \times d_X$.

 Note that we obtain $\beta_0$'s identification up to scale under the index sufficiency condition without assuming that $F_{U_t|C,\X}$ is stationary, that $g_t(X_t'\beta_0,C,U_t)$ has the functional form $\1(X_t'\beta_0 + C + U_t \geq 0)$, or that $X_t'\beta_0$ has large support. These three conditions  are all required by \cite{Manski1987} to obtain the identification of $\beta_0$. Therefore, one can relax a number of assumptions that lead to the identification of $\beta_0$ by using the index assumption $C\indep \X|V$.

\subsection{Index = $v(\mathbf X)'\gamma_0$}\label{sec:betaid-vgamma}
Now we investigate the scenario where the index $v(\X)'\gamma_0$ is known up to finite-dimensional parameters $\gamma_0$.

More specifically, suppose that there are $d_{V^*}$ indices, and each index is a linear combination of known functions of $\X$ with unknown coefficients. That is, $V^*= \left(v^*_1(\X)'\gamma^*_1, \ldots, v^*_{d_{V^*}}(\X)'\gamma^*_{d_{V^*}}\right)\in\R^{1\times d_{V^*}}$. Note that $V^*$ can be multivariate with $d_{V^*}>1$. For notational  simplicity, let $v(\X)$ denote a collection of unique elements in $v^*_1(\X),\ldots,v^*_{d_{V^*}}(\X)$ with dimension $d_V\times 1$, and $\gamma_0$ be the corresponding coefficient matrix with dimension $d_V \times d_{V^*}$ such that 
\begin{align}V^*= \left(v^*_1(\X)'\gamma^*_1, \ldots, v^*_{d_{V^*}}(\X)'\gamma^*_{d_{V^*}}\right)=v(\X)'\gamma_0 .\label{eq:vgamma}\end{align}
The index structure in equation \eqref{eq:vgamma} determines the locations of zero entries in the coefficient matrix $\gamma_0$. Then, $\gamma^*_1,\cdots,\gamma^*_{d_{V^*}}$ are the \textit{non-zero elements} in each column of $\gamma_0$, and $v^*_1(\X),\ldots,v^*_{d_{V^*}}(\X)$ are the corresponding elements in $v(\X).$ Note that we only need to estimate these non-zero coefficients, which can be concatenated into a vector $\gamma^*=\left(\gamma^{*\prime}_1,\cdots,\gamma^{*\prime}_{d_{V^*}}\right)'\in\R^{d_{\gamma^*}}$. 

Also note that the dimension of $V^*=V'\gamma_0$  is usually lower than that of $V$. This dimension reduction further facilitates the estimation of the conditional expectation of the outcomes given both $X_t'\beta_0$ and the sufficient statistic, which is obtained via a nonparametric regression as detailed in Section \ref{sec:estimation} and Supplemental Appendix \ref{sec:estimation2}. 

In this case, the index sufficiency condition can be stated as follows.
\begin{IDsec_assump}[Unknown index sufficiency] \label{assn:index_unknown}
Given $V \equiv v(\X)$, where $v:\R^{T \times d_X} \rightarrow \R^{d_{V} }$ is known, let $C \mid \X \deq C \mid V'\gamma_0$ for a $\gamma_0 \in \R^{d_V \times d_{V^*}}$.
\end{IDsec_assump} 
  
\remark\label{rm:v-vgamma} For a given $V$, $C \indep \X|V'\gamma_0$ implies $C  \indep \X|V$, that is, A\ref{assn:index_unknown} implies A\ref{assn:index}. See Appendix \ref{app_sec:proofs} for a formal proof. Intuitively, A\ref{assn:index_unknown} introduces more structure on how $C$ depends on $\X$ through $V$, and thus helps further reduce the dimensionality in estimation.
 
Under Assumption A\ref{assn:index_unknown} we can show that
\begin{align*}
        \E[Y_t|\X = \ubbx] &= \E[g_t(\ubx_t'\beta_0,C,U_t)|V'\gamma_0 = v(\ubbx)'\gamma_0]\\
        &\equiv \Psi_t(\ubx_t'\beta_0,v(\ubbx)'\gamma_0),
\end{align*}
an unknown function containing $d_{V^*} + 1$ linear indices of $\ubbx$. \cite{IchimuraLee1991} show that under conditions on the indices $(\ubx_t'\beta_0,v(\ubbx)'\gamma_0)=(\ubx_t'\beta_0,v^*_1(\ubbx)'\gamma^*_1, \ldots, v^*_{d_{V^*}}(\ubbx)'\gamma^*_{d_{V^*}})$, the parameters $\left(\beta_0',\gamma^{*\prime}_1,\cdots,\gamma^{*\prime}_{d_{V^*}}\right)'$ are point identified up to scale and are estimable at a $\sqrt{N}$-rate. These conditions require that each index contains a continuous component, that none of the vectors $(X_t,v^*_1(\X),\ldots,v^*_{d_{V^*}}(\X))$ are contained in one another, and the linear independence of the $1 + d_{V^*}$ partial derivatives of $\Psi_t$ with respect to the indices. Considering the panel structure of the data, the indices $v(\X)$ and $X_t$ generally contain different elements: $v(\X)$ depends on the entire time series of covariates $\X=(X_1,\ldots,X_T)$ but remains constant over time, whereas $X_t$ is time-varying and depends only on the covariates at time $t$. We now state a version of their Lemma 3 when regressors are continuous.

\begin{proposition}[Lemma 3 in \cite{IchimuraLee1991}]\label{prop:betaID_IchimuraLee}
Assume that
\begin{enumerate}
        \item Assumptions A\ref{assn:model} and A\ref{assn:index_unknown} hold;
        \item Each of $(X_t,v^*_1(\X),\ldots,v^*_{d_{V^*}}(\X))$ contains a continuous component with nonzero coefficient which is not contained in any of the other variables;
        \item For some $t \in \{1,\ldots,T\}$, the function $\Psi_t(u,v^*) \equiv \E[Y_t|X_t'\beta_0 = u, V'\gamma_0 = v^*]$ is differentiable;
        \item The partial derivatives $\left.\left(\frac{\partial \Psi_t(u,v^*)}{\partial u}, \frac{\partial \Psi_t(u,v^*)}{\partial v^*_1}, \cdots, \frac{\partial \Psi_t(u,v^*)}{\partial v^*_{d_{V^*}}}\right)\right|_{(u,v^*) = (X_t'\beta_0,v(\X)'\gamma_0)}$ are not linearly dependent with probability 1.
\end{enumerate}
Then, $\left(\beta_0',\gamma^{*\prime}_1,\cdots,\gamma^{*\prime}_{d_{V^*}}\right)'$ are identified up to scale.
\end{proposition}

To summarize, here and in Section \ref{subsubsec:betaidassn}, we have highlighted a few identification approaches that result in the point identification (up to scale) of $\beta_0$. This list is not exhaustive and contains a new approach that uses the index sufficiency assumption.

\section{Proofs}\label{app_sec:proofs}
\subsection{Proofs for Section \ref{sec:identification}}\label{app_sec:IDproofs}

\begin{proof}[Proof of Theorem \ref{thm:ASF-APE_id}]

We break this proof into three steps.

\noindent \textbf{Step 1}: We first show that $(C,U_t) \indep h(\X)|V$ for any function $h$.

To see this, let $\mathcal{B} \subseteq \supp(C,U_t)$ be a set and write
\begin{align*}
        \Prob((C,U_t) \in \mathcal{B}|h(\X),V) &= \E[\E[\1((C,U_t) \in \mathcal{B})|C,h(\X),V]|h(\X),V]\\
        &= \E[\E[\1((C,U_t) \in \mathcal{B})|C,V]|h(\X),V]\\
        &= \E[\E[\1((C,U_t) \in \mathcal{B})|C,V]|V]\\
        &= \Prob((C,U_t) \in \mathcal{B}|V).
\end{align*}
The first and fourth equality follows from iterated expectations. The second equality follows from $U_t |(h(\X),V,C) \deq U_t|C \deq U_t|(V,C)$, which is implied by $U_t \indep \X|C$ (Assumption A\ref{assn:model}.(ii)) and by $(h(\X),V)$ being functions of $\X$. To show the third equality holds, note that $C \indep \X|V$ by A\ref{assn:index}. Since $h(\X)$ is a function of $\X$, we also have that $C \indep h(\X)|V$. The set $\mathcal{B}$ was arbitrary thus $(C,U_t) \indep h(\X)|V$.

\noindent \textbf{Step 2}: To show the ASF is identified, note that
\begin{align*}
        \text{ASF}_t(\ubx_t) &= \E[g_t(\ubx_t'\beta_0,C,U_t)]\\
        &= \int_{\supp(V)} \E[g_t(\ubx_t'\beta_0,C,U_t)|V = v] \, dF_V(v)\\
        &= \int_{\supp(V|X_t'\beta_0 = \ubx_t'\beta_0)} \E[g_t(\ubx_t'\beta_0,C,U_t)|V = v] \, dF_V(v)\\
        &= \int_{\supp(V|X_t'\beta_0 = \ubx_t'\beta_0)} \E[g_t(\ubx_t'\beta_0,C,U_t)|X_t'\beta_0 = \ubx_t'\beta_0,V = v] \, dF_V(v)\\
        &= \int_{\supp(V|X_t'\beta_0 = \ubx_t'\beta_0)} \E[Y_t|X_t'\beta_0 = \ubx_t'\beta_0,V =v] \, dF_V(v).
\end{align*}
The second equality follows from iterated expectations, and the third from the support assumption in the theorem's statement. The fourth follows from $(C,U_t) \indep X_t'\beta_0|V$, which is shown in step 1 by setting $h(\X) = X_t'\beta_0$. The conditional expectations are well defined since $v \in \supp(V|X_t'\beta_0 = \ubx_t'\beta_0)$ and $\ubx_t'\beta_0 \in \supp(X_t'\beta_0)$. 

The expression in the final equality depends on $\{\E[Y_t|X_t'\beta_0 = \ubx_t'\beta_0, V = v]:v \in \supp(V|X_t'\beta_0 = \ubx_t'\beta_0)\}$, and $F_V$, the marginal distribution of $V$. By A\ref{assn:betaid}, $\beta_0$ is identified up to scale. Thus, the conditioning set $\{X_t'\beta_0 = \ubx_t'\beta_0\}$ is identified since it is invariant to the scale of $\beta_0$. By definition, $\E[Y_t|X_t'\beta_0 = \ubx_t'\beta_0, V = v]$ is identified for $v \in \supp(V|X_t'\beta_0 = \ubx_t'\beta_0)$. The marginal distribution of $V$ is also identified since $V = v(\X)$ is a known function of observed $\X$. Hence, $\text{ASF}_t(\ubx_t)$ is identified from the distribution of $(Y,\X)$. 

\noindent \textbf{Step 3}: For the APE, by the differentiability assumption we can write
\begin{align}
        \text{APE}_{k,t}(\ubx_t) &= \frac{\partial}{\partial \ubx_t^{(k)}} \text{ASF}_t(\ubx_t) = \lim_{\eta \to 0} \frac{\text{ASF}_t(\ubx_t + \eta e_k) - \text{ASF}_t(\ubx_t)}{\eta },\label{eq:APE_proof}
\end{align}
where $e_k$ is a vector of zeros with component $k$ equal to 1. Thus, $\text{APE}_{k,t}(\ubx_t)$ is identified if $\text{ASF}_t(\ubx_t + \eta e_k)$ is identified for all $\eta $ in a neighborhood of 0. By the above result, this is the case if $\supp(V|X_t'\beta_0 = (\ubx_t + \eta  e_k)'\beta_0) = \supp(V)$ for all $\eta $ in a neighborhood of 0. Note that $\supp(V|X_t'\beta_0 = (\ubx_t + \eta  e_k)'\beta_0) = \supp(V|X_t'\beta_0 = \ubx_t'\beta_0 + \eta  \beta_0^{(k)})$, and that $\ubx_t'\beta_0 + \eta  \beta_0^{(k)}$ lies in an arbitrary neighborhood of $\ubx_t'\beta_0$ as $\eta  \to 0$. Since we assumed that $\supp(V|X_t'\beta_0 = u) = \supp(V)$ for all $u$ in a neighborhood of $\ubx_t'\beta_0$, $\text{ASF}_t(\ubx_t + \eta e_k)$ is identified for $\eta $ sufficiently close to 0.  Therefore $\text{APE}_{k,t}(\ubx_t)$ is identified following the identification of the ASFs in equation \eqref{eq:APE_proof}.
\end{proof}

\begin{proof}[Proof of Theorem \ref{thm:LAR-AME_id}]
We can write the LAR as
\begin{align*}
        \text{LAR}_{k,t}(\ubbx) &= \left.\frac{\partial \E[g_t(\ubx_t'\beta_0,C,U_t)|\X = \x]}{\partial \ubx_t^{(k)}} \right|_{\x = \ubbx}\\
        &= \left.\frac{\partial \E[g_t(\ubx_t'\beta_0,C,U_t)|v(\X) = v(\x)]}{\partial \ubx_t^{(k)}} \right|_{\x = \ubbx}\\
        &= \lim_{\eta  \to 0} \frac{\E[g_t(\ubx_t'\beta_0 + \eta \beta_0^{(k)},C,U_t)|V = v(\ubbx)] - E[g_t(\ubx_t'\beta_0,C,U_t)|V = v(\ubbx)]}{\eta }\\
        &= \lim_{\eta  \to 0} \frac{1}{\eta } \left(\E[g_t(\ubx_t'\beta_0 + \eta \beta_0^{(k)},C,U_t)|X_t'\beta_0 = \ubx_t'\beta_0 + \eta \beta_0^{(k)},V = v(\ubbx)]\right.\\
        &\left. \qquad\qquad - \E[g_t(\ubx_t'\beta_0,C,U_t)|X_t'\beta_0 = \ubx_t'\beta_0, V = v(\ubbx)]\right)\\
        &= \lim_{\eta  \to 0} \frac{\E[Y_t|X_t'\beta_0 = \ubx_t'\beta_0 + \eta \beta_0^{(k)}, V = v(\ubbx)] - \E[Y_t|X_t'\beta_0 = \ubx_t'\beta_0,V = v(\ubbx)]}{\eta }\\
        &= \left.\frac{\partial \E[Y_t|X_t'\beta_0 = \ubx_t'\beta_0, V = v]}{\partial \ubx_t^{(k)}} \right|_{v = v(\ubbx)}.
\end{align*}
The first equality follows by the differentiability of $\E[Y_t(\ubx_t)|\X = \x]$ in $\ubx_t^{(k)}$ and the definition of the LAR. The second follows $(C,U_t) \indep \X|v(\X)$, which is shown in step 1 of the proof of Theorem \ref{thm:ASF-APE_id}. The third is by definition of partial derivatives, and the fourth follows from $v(\ubbx)$ being in $\supp(V|X_t'\beta_0 = \ubx_t'\beta_0 + \eta  \beta_0^{(k)})$ for all sufficiently small $\eta $, which is assumed in the theorem statement. The fifth and sixth equalities follow immediately.

By A\ref{assn:betaid}, $\beta_0$ is identified up to scale, thus the conditioning set $\{X_t'\beta_0 = \ubx_t'\beta_0\}$ is identified. Therefore, the conditional expectation $\E[Y_t|X_t'\beta_0 = \ubx_t'\beta_0 + \eta \beta_0^{(k)}, V = v(\ubbx)]$ is identified for sufficiently small $\eta $ if $v(\ubbx) \in \supp(V|X_t'\beta_0 = \ubx_t'\beta_0 + \eta \beta_0^{(k)})$, which is implied by $v(\ubbx) \in \supp(V|X_t'\beta_0 = u)$ for all $u$ in a neighborhood of $\ubx_t'\beta_0$. This implies that $\text{LAR}_{k,t}(\ubbx)$ is identified. 

The proof of the identification of the AME follows from the identification of the LAR for all $\ubbx$ in $\supp(\X)$ (up to a $\Prob_\X$-measure zero set) and from 
\begin{align*}
        \text{AME}_{k,t} &= \int_{\supp(\X)} \text{LAR}_{k,t}(\ubbx) \, dF_{\X}(\ubbx).
\end{align*}
\end{proof}

\begin{proof}[Proof of Theorem \ref{thm:ASF-bounds}]\label{pf:ASF-bounds}
We first show that $\text{ASF}_t(\ubx_t) \in \mathcal{A}_t(\ubx_t)$ and then show the sharpness of this set.

\noindent\textbf{Part 1: Bounds}

By iterated expectations,
\begin{align}
        \text{ASF}_t(\ubx_t) &= \E[g_t(\ubx_t'\beta_0,C,U_t)]\notag\\
        &= \int_{\mathcal{V}} \E[g_t(\ubx_t'\beta_0,C,U_t)|V = v]\, dF_V(v)\notag\\
        &= \int_{\mathcal{V}_t(\ubx_t'\beta_0)} \E[g_t(\ubx_t'\beta_0,C,U_t)|V = v]\, dF_V(v) + \int_{\mathcal{V} \setminus \mathcal{V}_t(\ubx_t'\beta_0)} \E[g_t(\ubx_t'\beta_0,C,U_t)|V = v]\, dF_V(v).\label{eq:ASFboundsproof1}
\end{align}
The first term in \eqref{eq:ASFboundsproof1} is
\begin{align*}
        \int_{\mathcal{V}_t(\ubx_t'\beta_0)} \E[g_t(\ubx_t'\beta_0,C,U_t)|V = v]\, dF_V(v) &= \int_{\mathcal{V}_t(\ubx_t'\beta_0)} \E[g_t(\ubx_t'\beta_0,C,U_t)|X_t'\beta_0 = \ubx_t'\beta_0, V = v] \,dF_V(v)\\
        &= \int_{\mathcal{V}_t(\ubx_t'\beta_0)} \E[Y_t|X_t'\beta_0 = \ubx_t'\beta_0, V = v]\, dF_V(v),
\end{align*}
where the first equality is obtained by step 1 of the proof of Theorem \ref{thm:ASF-APE_id} and by $v \in \mathcal{V}_t(\ubx_t'\beta_0)$. 

The second term in \eqref{eq:ASFboundsproof1} is bounded below by
\begin{align*}
        \int_{\mathcal{V} \setminus \mathcal{V}_t(\ubx_t'\beta_0)} \E[g_t(\ubx_t'\beta_0,C,U_t)|V = v]\, dF_V(v) &\geq \int_{\mathcal{V} \setminus \mathcal{V}_t(\ubx_t'\beta_0)} \E[\underline{g}|V = v]\, dF_V(v) = \underline{g} \cdot \Prob(V \in \mathcal{V}\setminus \mathcal{V}_t(\ubx_t'\beta_0)).
\end{align*}
By a similar argument, it is bounded above by $\overline{g} \cdot \Prob(V \in \mathcal{V}\setminus \mathcal{V}_t(\ubx_t'\beta_0))$. We conclude that $\text{ASF}_t(\ubx_t) \in \mathcal{A}_t(\ubx_t)$.

\noindent\textbf{Part 2: Sharpness}

We consider the identified set for period $t=1$ without loss of generality. Let $a \in \mathcal{A}_1(\ubx_1)$. We want to show that $a$ is in the identified set for $\text{ASF}_1(\ubx_1)$. To achieve this, we will find structural functions $(\tilde{g}_1,\ldots,\tilde{g}_T)$ and a conditional distribution of $(\tilde{C},\tilde{U}_1,\ldots,\tilde{U}_T)|\X$ such that: \begin{enumerate}[label=(\alph*)]
\item Assumptions A\ref{assn:model} and A\ref{assn:index} hold almost surely for $\left(\tilde C,\{\tilde U_t\},\{\tilde g_t\}\right)$,
\item $\tilde{g}_1(\ubx_1'\beta_0,\tilde{C},\tilde{U}_1) \in [\underline{g}, \overline{g}]$ almost surely, 
\item the distribution of $(\tilde{g}_1(X_1'\beta_0,\tilde{C},\tilde{U}_1),\ldots,\tilde{g}_T(X_T'\beta_0,\tilde{C},\tilde{U}_T))|\X$ equals the distribution of $Y|\X$,
\item $\E[\tilde{g}_1(\ubx_1'\beta_0,\tilde{C},\tilde{U}_t)] = a$.\end{enumerate} Below we first construct $\left(\tilde C,\{\tilde U_t\},\{\tilde g_t\}\right)$, and then verify that they satisfy conditions (a)--(d).

First, we define $\tilde{U}_t$ based on the Skorokhod Representation Theorem. For $t=1,\ldots,T$ we can write
\begin{align*}
        Y_t &= Q_{Y_t|\X}(\tilde{U}_t|\X)
\end{align*}
where $Q_{Y_t|\X}$ is a conditional quantile function and  $\tilde{U}_t|\X \sim \text{Unif}(0,1)$. We then let $\tilde{C} = V$ and
\begin{align*}
        \tilde{g}_1(\ubx_1'\beta_0,\tilde c ,\tilde u _1) &= Q_{Y_1|X_1'\beta_0,V}(\tilde u _1|\ubx_1'\beta_0,\tilde c )\cdot \1(\tilde c  \in \mathcal{V}_1(\ubx_1'\beta_0)) + \tilde{a} \cdot\1(\tilde c  \in \mathcal{V}\setminus \mathcal{V}_1(\ubx_1'\beta_0))\\
        \tilde{g}_2(\ubx_2'\beta_0,\tilde c ,\tilde u _2) &= Q_{Y_2|X_2'\beta_0,V}(\tilde u _2|\ubx_2'\beta_0,\tilde c )\cdot\1(\tilde c  \in \mathcal{V}_2(\ubx_2'\beta_0)) + \underline{g} \cdot\1(\tilde c  \in \mathcal{V}\setminus \mathcal{V}_2(\ubx_2'\beta_0))\\
        &\;\;\vdots\\
        \tilde{g}_T(\ubx_T'\beta_0,\tilde c ,\tilde u _T) &= Q_{Y_T|X_T'\beta_0,V}(\tilde u _T|\ubx_T'\beta_0,\tilde c )\cdot\1(\tilde c  \in \mathcal{V}_T(\ubx_T'\beta_0)) + \underline{g} \cdot\1(\tilde c  \in \mathcal{V}\setminus \mathcal{V}_T(\ubx_T'\beta_0)),
\end{align*}
where
\begin{align}
        \tilde{a} &\equiv \frac{a - \int_{\mathcal{V}_1(\ubx_1'\beta_0)} \E[Y_1|X_1'\beta_0 = \ubx_1'\beta_0, V = v] \, dF_V(v)}{\Prob(V \in \mathcal{V}\setminus \mathcal{V}_1(\ubx_1'\beta_0))}.\label{eq:ASFboundsproof2}
\end{align}
If $\Prob(V \in \mathcal{V}\setminus \mathcal{V}_1(\ubx_1'\beta_0)) = 0$, $\text{ASF}_1(\ubx_1)$ is point identified by Theorem \ref{thm:ASF-APE_id} and this proof is not needed. Thus, we assume $\Prob(V \in \mathcal{V}\setminus \mathcal{V}_1(\ubx_1'\beta_0)) > 0$ for the remainder of this proof. Finally, the conditional quantiles above are well defined when $\tilde c  \in \mathcal{V}_t(\ubx_t'\beta_0)$, so these expressions are well defined due to the indicator functions.

In what follows, we verify that the $\left(\tilde C,\{\tilde U_t\},\{\tilde g_t\}\right)$ constructed above satisfy conditions (a)--(d). First, we show that Assumptions A\ref{assn:model} and A\ref{assn:index} hold almost surely for $\left(\tilde C,\{\tilde U_t\},\{\tilde g_t\}\right)$. We note that $\Prob(V \in \mathcal{V}_t(X_t'\beta_0)) = \E[\Prob(V \in \mathcal{V}_t(X_t'\beta_0)|X_t'\beta_0)] = \E[1] = 1$, therefore $\1(V \in \mathcal{V}_t(X_t'\beta_0)) = 1$ almost surely for $t = 1,\ldots,T$. Also, since $Y_t = g_t(X_t'\beta_0,C,U_t)$, we have that $Y_t \indep \X|(X_t'\beta_0,V)$ by Assumptions A\ref{assn:model}.(ii) and A\ref{assn:index} for the original $\left( C,\{ U_t\},\{ g_t\}\right)$. Therefore, $Q_{Y_t|\X}(\cdot|\X) = Q_{Y_t|X_t'\beta_0,V}(\cdot|X_t'\beta_0,V)$. From these, we obtain that
\begin{align*}
        Y_1 &= Q_{Y_1|\X}(\tilde{U}_1|\X)\\
        &\overset{a.s.}{=} Q_{Y_1|X_1'\beta_0,V}(\tilde{U}_1|X_1'\beta_0,V)\cdot \1(V \in \mathcal{V}_1(X_1'\beta_0)) + \tilde{a} \cdot\1(V \in \mathcal{V}\setminus \mathcal{V}_1(X_1'\beta_0))\\
        &= \tilde{g}_1(X_1'\beta_0,\tilde{C},\tilde{U}_1)
\end{align*}
where we used that $\1(V \in \mathcal{V}_1(X_1'\beta_0)) = 1$ almost surely, and $\tilde{C} = V$. Likewise, for $t=2,\ldots,T$ we have that
\begin{align*}
        Y_t &= Q_{Y_t|\X}(\tilde{U}_t|\X)\\
        &\overset{a.s.}{=} Q_{Y_t|X_t'\beta_0,V}(\tilde{U}_t|X_t'\beta_0,V)\cdot \1(V \in \mathcal{V}_t(X_t'\beta_0)) + \underline{g} \cdot\1(V \in \mathcal{V}\setminus \mathcal{V}_t(X_t'\beta_0))\\
        &= \tilde{g}_t(X_t'\beta_0,\tilde{C},\tilde{U}_t).
\end{align*}
Therefore, Assumption A\ref{assn:model}.(i) holds almost surely for $\left(\tilde C,\{\tilde U_t\},\{\tilde g_t\}\right)$. Assumption A\ref{assn:model}.(ii) holds for $\left(\tilde C,\{\tilde U_t\},\{\tilde g_t\}\right)$, since $\tilde{U}_t \indep \X$ by construction and $\tilde{C} = V = v(\X)$, a function of $\X$. Assumption A\ref{assn:model}.(iii) holds for $\left(\tilde C,\{\tilde U_t\},\{\tilde g_t\}\right)$ because $[\underline{g},\overline{g}]$ is a bounded interval. Assumption A\ref{assn:index} holds trivially  for $\left(\tilde C,\{\tilde U_t\},\{\tilde g_t\}\right)$ since $\tilde{C} \indep \X|V$ when $\tilde{C} = V$, as is the case here. Therefore all model assumptions hold almost surely for $\left(\tilde C,\{\tilde U_t\},\{\tilde g_t\}\right)$.

Second, note that $\tilde{a} \in [\underline{g},\overline{g}]$ by $a \in \mathcal{A}_1(\ubx_1)$, so $\tilde{g}_1(\ubx_1,\tilde{C},\tilde{U}_1) \in [\underline{g}, \overline{g}]$ almost surely.  

Third, the distributions of $(\tilde{g}_1(X_t'\beta_0,\tilde{C},\tilde{U}_1),\ldots,\tilde{g}_T(X_T'\beta_0,\tilde{C},\tilde{U}_T))|\X$ and $Y|\X$ coincide because $(\tilde{g}_1(X_1'\beta_0,\tilde{C},\tilde{U}_1),\ldots,\tilde{g}_T(X_T'\beta_0,\tilde{C},\tilde{U}_T)) = Y$ almost surely, as shown above.

Finally, we verify that the implied ASF, $\E[\tilde{g}_1(\ubx_1'\beta_0,\tilde{C},\tilde{U}_1)]$, equals $a$:
\begin{align}
        &\E[\tilde{g}_1(\ubx_1'\beta_0,\tilde{C},\tilde{U}_1)]\notag\\
        &= \int_{\mathcal{V}} \E[\tilde{g}_1(\ubx_1'\beta_0,\tilde{C},\tilde{U}_1)|V = v] \, dF_V(v)\notag\\
        &= \int_{\mathcal{V}_1(\ubx_1'\beta_0)} \E[Q_{Y_1|X_1'\beta_0,V}(\tilde{U}_1|\ubx_1'\beta_0,V)|V = v] \, dF_V(v)  + \int_{\mathcal{V} \setminus \mathcal{V}_1(\ubx_1'\beta_0)} \E[\tilde{a}|V = v] \, dF_V(v)\notag \\
        &= \int_{\mathcal{V}_1(\ubx_1'\beta_0)} \E[Q_{Y_1|X_1'\beta_0,V}(\tilde{U}_1|X_1'\beta_0,V)|X_1'\beta_0  = \ubx_1'\beta_0, V = v] \, dF_V(v)  + \tilde{a} \cdot \Prob(V \in \mathcal{V}\setminus \mathcal{V}_1(\ubx_1'\beta_0)).\label{eq:ASFboundsproof3}
\end{align}
The first equality follows by iterated expectations, the second by substituting the expression for $\tilde{g}_1$, and the third from $\tilde{U}_1 \indep X_1'\beta_0 |V$, which is implied by $\tilde{U}_1 \indep \X$. By $\tilde{U}_1 \indep (X_1'\beta_0,V)$ and by properties of conditional quantiles, we have that
\begin{align}
        \E[Q_{Y_1|X_1'\beta_0,V}(\tilde{U}_1|X_1'\beta_0,V)|X_1'\beta_0  = \ubx_1'\beta_0, V = v] &= \int_0^1  Q_{Y_1|X_1'\beta_0,V}(\tilde{u}_1|\ubx_1'\beta_0,v)\, d\tilde{u}_1\notag\\
        &= \E[Y_1|X_1'\beta_0  = \ubx_1'\beta_0, V = v].\label{eq:ASFboundsproof4}
\end{align}
Substituting the expression for $\tilde{a}$ in \eqref{eq:ASFboundsproof2} and the expression for the conditional mean in  \eqref{eq:ASFboundsproof4} into equation \eqref{eq:ASFboundsproof3}, we obtain 
\begin{align*}
        \E[\tilde{g}_1(\ubx_1'\beta_0,\tilde{C},\tilde{U}_1)] &= \int_{\mathcal{V}_1(\ubx_1'\beta_0)} \E[Y_1|X_1'\beta_0  = \ubx_1'\beta_0, V = v] \, dF_V(v)\\
        &\qquad  + \frac{a - \int_{\mathcal{V}_1(\ubx_1'\beta_0)} \E[Y_1|X_1'\beta_0 = \ubx_1'\beta_0, V = v] \, dF_V(v)}{\Prob(V \in \mathcal{V}\setminus \mathcal{V}_1(\ubx_1'\beta_0))} \cdot \Prob(V \in \mathcal{V}\setminus \mathcal{V}_1(\ubx_1'\beta_0))\\
        &= a,
\end{align*}
which concludes the proof.
\end{proof}

\begin{proof}[Proof of Corollary \ref{thm:ASF-bounds-fe}]
Note that $\E[Y_t|X_t'\beta_0 = \ubx_t'\beta_0,\X = \x] = \E[Y_t|\X = \x]$ for all $\x \in \supp(\X|X_t'\beta_0 = \ubx_t'\beta_0)$. Also note that $\Prob(\X \in \supp(\X) \setminus \supp(\X|X_t'\beta_0 = \ubx_t'\beta_0)) = \Prob(X_t'\beta_0 \neq \ubx_t'\beta_0)$. We also have that $C \indep \X|\X$ trivially.
This corollary follows from a combination of the above results with an application of Theorem \ref{thm:ASF-bounds} for $V = \X$.
\end{proof}

\subsection{Proofs for Appendix Section \ref{sec:betaID}}\label{app_sec:betaID-proofs}
\begin{proof}[Proof of Proposition \ref{prop:betaindexID}]
Without loss of generality, let $s = 1$ and $t=2$. We begin by computing
\begin{align*}
        \frac{\partial}{\partial x_1} \E[Y_1|\X = \x] &=  \frac{\partial}{\partial x_1} \Psi_1(x_1'\beta_0,v(\x)) = \left.\frac{\partial}{\partial u} \Psi_1(u,v(\x))\right|_{u = x_1'\beta_0} \begin{pmatrix}
                1\\
                \widetilde\beta_0
        \end{pmatrix} + W_1(\x) \frac{\partial}{\partial v} \Psi_1(x_1'\beta_0,v)|_{v = v(\x)}\\
        \frac{\partial}{\partial x_2} \E[Y_1|\X = \x] &=  \frac{\partial}{\partial x_2} \Psi_1(x_1'\beta_0,v(\x)) = W_2(\x) \frac{\partial}{\partial v} \Psi_1(x_1'\beta_0,v)|_{v = v(\x)}.
\end{align*}
where
\begin{align*}
        W_t(\x) &= \frac{\partial}{\partial x_t} v(\x)' \in \R^{d_X \times d_V}.
\end{align*}
From the second equation above, we can recover 
$$ \frac{\partial}{\partial v} \Psi_1(x_1'\beta_0,v)|_{v = v(\x)}  = \left(W_2(\x)'W_2(\x)\right)^{-1}W_2(\x)'\frac{\partial}{\partial x_2} \E[Y_1|\X = \x]$$
since $W_2(\x)$ has rank $d_V$. 

Therefore, we can identify
\begin{align}
        \left.\frac{\partial}{\partial u} \Psi_1(u,v(\x))\right|_{u = x_1'\beta_0} \begin{pmatrix}
                1\\
                \widetilde\beta_0
        \end{pmatrix} = \frac{\partial}{\partial x_1} \E[Y_1|\X = \x] - W_1(\x)\left(W_2(\x)'W_2(\x)\right)^{-1}W_2(\x)'\frac{\partial}{\partial x_2} \E[Y_1|\X = \x].\label{eq:proof_betaindexID}
\end{align}
By the assumption that $\frac{\partial \Psi_1(u,v(\x))}{\partial u} \neq 0$, $\widetilde\beta_0$ is point identified as the ratio of the elements in equation \eqref{eq:proof_betaindexID}.
\end{proof}

\begin{proof}[Proof of Proposition \ref{prop:betaID_IchimuraLee}]
By Assumption A\ref{assn:model}, we have that
\begin{align*}
        \E[Y_t|\X] &= \int_{\supp(C,U_t|\X)} g_t(X_t'\beta_0,c,u_t) \, dF_{C,U_t|\X}(c,u_t|\X)\\
        &= \int_{\supp(C,U_t|\X)} g_t(X_t'\beta_0,c,u_t) \, dF_{C,U_t|V'\gamma_0}(c,u_t|v(\X)'\gamma_0)\\
        &= \Psi_t(X_t'\beta_0,v(\X)'\gamma_0)
\end{align*}
where the second equality follows from $(C,U_t) \indep \X|v(\X)'\gamma_0$, which follows from Assumption A\ref{assn:index_unknown} and step 1 in the proof of Theorem \ref{thm:ASF-APE_id}. The result then follows immediately from Lemma 3 in \cite{IchimuraLee1991} with no additive index in the outcome equation and when applied to continuous regressors only.
\end{proof}

\begin{proof}[Proof of Remark \ref{rm:v-vgamma}]\label{pf:v-vgamma} 
First, we have that $C \indep h(\X)|V'\gamma_0$ for any function $h$, following from Assumption A\ref{assn:index_unknown} and step 1 in the proof of Theorem \ref{thm:ASF-APE_id}.

Then, we have that 
\begin{align*}
        C|V &\deq C|V,V'\gamma_0 \deq C|V'\gamma_0\deq C|\X.
\end{align*}
The first distributional equality is from $V'\gamma_0$ being known when $V$ is known. The second follows from $C \indep h(\X)|V'\gamma_0$ selecting $h(\X) = v(\X)$. The third is given by Assumption A\ref{assn:index_unknown}. Therefore, we obtain the condition in Assumption A\ref{assn:index}: $C|\X \deq C|V$.
\end{proof}
\label{loc:end-main-all}

\end{document}

\title{Online Appendix:\\Identification and Estimation of Partial Effects in Nonlinear Semiparametric Panel Models}

\author{Laura Liu
\and
Alexandre Poirier 
\and
Ji-Liang Shiu 
}

\maketitle

\newsavebox{\tablebox} \newlength{\tableboxwidth}

\appendix
\setcounter{section}{2}

\numberwithin{equation}{section}
\numberwithin{figure}{section}
\numberwithin{table}{section}

This online appendix is organized as follows.  In Appendix \ref{sec:estimation2} we establish the asymptotic properties of the proposed ASF and APE estimators. In Appendix \ref{app_sec:implementation} we discuss a number of implementation details for the estimators. In Appendix \ref{sec:extensions} we extend our main identification results to models with sequential exogeneity, allowing for lagged outcomes as regressors. In Appendix \ref{sec:MonteCarlo} we conduct Monte Carlo experiments to study the finite-sample properties of our estimators. In Appendix \ref{sec:appendix-app} we present additional figures and tables that supplement the empirical results in the main text.
Appendix \ref{supp-app-sec:proofs} contains proofs for results in appendices \ref{sec:estimation2} and \ref{sec:extensions}.

\section{Estimation and Inference} \label{sec:estimation2}
We propose three-step semiparametric estimators for the ASF, APE, and AME. The first step estimates the common parameters $\beta_0$, the second step conducts a nonparametric regression with a generated regressor, and the third step marginalizes over a subset of the regressors. Such estimators are called partial means.\footnote{See \cite{Newey1994} for seminal work on the estimation of partial means without generated regressors. The estimation of partial means with generated regressors is studied in \citet*{MammenRotheSchienle2012, MammenRotheSchienle2016}, and \cite{Lee2018}.}

We show that the rate of convergence of the ASF estimator is similar to that of a kernel regression estimator with one continuous regressor. The APE estimator converges at a similar rate as a derivative of a kernel regression estimator with one continuous regressor. In particular, we show the ASF converges at the rate $\sqrt{Nb_N}$ and the APE at the rate $\sqrt{Nb_N^3}$ where $b_N$ is a scalar bandwidth used in the estimation of the conditional expectation of $Y_t$. We describe below in Assumption B\ref{assn:ASFrates} what assumptions $b_N$ must satisfy. These rates of convergence are obtained from our estimators being partial means, where we average over all components of the conditional expectation of $\E[Y_t|X_t'\beta_0,V]$, except for one. These convergence rates do not depend on $T$ or $d_X$, the dimensions of $\X$. 

Throughout this section, we assume we observe a random sample of $(Y_i,\X_i)$ of size $N$.
\begin{estsec_assump} [Random sampling] \label{assn:iid}
       $\{(Y_i,\X_i)\}_{i=1}^N$ are iid.
\end{estsec_assump}

We start by considering the estimation of $\beta_0$, the first step of our semiparametric estimator.

\subsection[]{Estimation of $\beta_0$} \label{subsec:betaestimation}
In Section \ref{subsubsec:betaidassn} we discussed several previous identification approaches for the common parameters $\beta_0$. Due to the breadth of these approaches, we consider the following high-level assumption on the rate of convergence of a first-step estimator of $\beta_0$.
\begin{estsec_assump}[First-stage estimator] \label{assn:betahat}
The estimator $\widehat{\beta}$ satisfies $a_N \|\widehat{\beta} - \beta_0\| = O_p(1)$ where $a_N = O(N^\epsilon)$ for some $\epsilon > 0$.
\end{estsec_assump}
The rate of convergence of this preliminary estimator plays a role in Assumption B\ref{assn:ASFrates} below. The convergence of $\widehat{\beta}$ needs to be relatively fast to establish the limiting distributions of the ASF and APE estimators. In particular, convergence rates equal to or slower than $N^{1/3}$ are incompatible with our rate assumption B\ref{assn:ASFrates} below. This rules out the maximum score estimator of \cite{Manski1987} for binary panels, but not the smoothed maximum score estimator of \cite{CharlierMelenbergSoest1995} and \cite{Kyriazidou1995}. The smoothed maximum score estimator converges at the rate $N^{\nu/(2\nu + 1)}$, where $\nu$ is the order of the kernel used to estimate $\widehat{\beta}$.

In binary panels, the rate of convergence is usually slower than $\sqrt{N}$. One exception is the $\sqrt{N}$-consistent conditional maximum likelihood estimator (\cite{Rasch1960}, \cite{Andersen1970}) when $U_t$ follows a logistic distribution. While $\sqrt{N}$-estimation of $\beta_0$ is generally not possible in binary panels without specifying $U_t$'s distribution (\cite{Magnac2004}, \cite{Chamberlain2010}), some alternative assumptions and estimators allow for it. In particular, \cite{Lee1999} considers an ``index increment sufficiency'' assumption: $(X_t'\beta_0,C)|(X_t - X_s) \deq (X_t'\beta_0,C)|(X_t - X_s)'\beta_0$. \cite{HonoreLewbel2002} assume the presence of a special regressor among $X_t$. \cite{ChenSiZhangZhou2017} assume that $C = v(\X) + \zeta$, where $\zeta$ satisfies $(U_1,\ldots,U_T,\zeta) \indep \X$. In all three papers, $\sqrt{N}$-consistent estimators for $\beta_0$ are proposed.

With continuous outcomes and the index function taking the form $v(\X)'\gamma_0$, \cite{IchimuraLee1991}'s approach can be used to estimate $\beta_0$ (and non-zero entries in the coefficient matrix $\gamma_0$) at a $\sqrt{N}$-rate: see Appendix \ref{sec:betaid-vgamma} for definitions of the notations. \cite{Abrevaya1999} proposes a $\sqrt{N}$-consistent \textit{leapfrog} estimator when $Y_{it} = g(X_{it}'\beta_0 + C_i + U_{it})$ and $Y_{it}$ is continuous. Also see \cite{Abrevaya2000} and  \cite{BotosaruMuris2017} for other $\sqrt{N}$-consistent estimators of $\beta_0$ in related models.

\subsection{A Semiparametric Estimator of the ASF}\label{sec:asf-est}
We now present the ASF estimator and show its consistency and asymptotic normality under our assumptions. As mentioned earlier, this estimator is a three-step estimator. Appendix \ref{subsec:betaestimation} discussed the first step, which consists of estimating $\beta_0$ using an existing method. We now describe the second and third steps, which estimate the ASF using a sample analog of equation \eqref{eq:ASF_expression}. In the second step, we nonparametrically estimate the conditional expectation $\E[Y_t|X_t'\beta_0 = \ubx_t'\beta_0, V = v]$ using a local polynomial regression of $Y_t$ on generated regressor $X_t'\widehat{\beta}$ and $V$. In the third and final step, we evaluate the estimated conditional expectation at $(\ubx_t'\widehat{\beta},V_i)$ for $i = 1,\ldots,N$, and then average over the empirical marginal distribution of $V_i$. To define this estimator, let $Z_t(\beta) = (X_t'\beta,V) \in \R^{1+d_V}$ and denote $Z_t = Z_t(\beta_0)$. Throughout the paper, we use $z$ to denote $z = (u,v) \in \R^{1 + d_V}$ where $u \in \R$ and $v \in \R^{d_V}$.

In the rest of this section, we assume that $V$'s components are all continuously distributed. We omit the discrete case for notational simplicity. In our analysis, the number of discrete components of $V$ does not affect the convergence rate. When the number of support points for the discrete components is sufficiently small, we can handle these discrete components by performing a cell-by-cell analysis. Alternatively, they can be accommodated through a discrete kernel, for example, as in \cite{LiRacine2004} equation (2.3). 

We consider a local polynomial regression of order $\ell \geq 0$.  The following notation is similar to that in \cite{Masry1996}. For $s \in \{0,1,\ldots,\ell\}$, let $N_s = \binom{s + d_V}{d_V}$ be the number of distinct $(1+d_V)$-tuples $r \in \mathbb{N}^{1+d_V}$ such that $|r| \equiv \sum_{k=1}^{1 + d_V} |r_k| = s$. We arrange these $(1+d_V)$-tuples in a lexicographical order with the highest priority given to the last position so that $(0,\ldots,0,s)$ is the first element and $(s,0,\ldots,0)$ is the last element in this sequence. We let $\tau_s$ denote this one-to-one mapping. This mapping satisfies $\tau_s(1) = (0,\ldots,0,s), \ldots, \tau_s(N_s) = (s,0,\ldots,0)$. For each $s\in\{0,1,\ldots,\ell\}$, define $N_s \times 1$ vector $\xi_s(a)$ by its $k$th element $a^{\tau_s(k)}$, where $k \in \{1,\ldots,N_s\}$ and $a \in \R^{1+d_V}$. Here we used the notation $a^{b} = a_1^{b_1} \times \cdots \times a_{d_V}^{b_{d_V}}$. Let 
\begin{align*}
       \xi(a) &= (1,\xi_1(a)',\ldots,\xi_\ell(a)')' \in \R^{\bar{N}},
\end{align*}
where $\bar{N} = \sum_{s = 0}^\ell N_s$.

Let $\mathcal{K}: \R^{1 + d_V} \to \R$ denote a $(1 + d_V)$-dimensional kernel. Let $\mathcal{K}_b(z) =  b^{-(1+d_V)} \mathcal{K}(z)$, where $b > 0$ is a scalar bandwidth. Let $b_N$ denote a sequence of bandwidths converging to zero.
Let
\begin{align*}
   \widehat{h}(z;\widehat\beta) &= \argmin_{h \in \R^{\bar{N}}} \sum_{j=1}^N  \left(Y_{jt} -  \sum_{0 \leq |r| \leq \ell} \left(\frac{Z_{jt}(\widehat{\beta}) - z}{b_N}\right)^{r}h_r \right)^2\mathcal{K}_{b_N}\left(\frac{Z_{jt}(\widehat{\beta}) - z}{b_N}\right)\\
   &= \argmin_{h \in \R^{\bar{N}}} \sum_{j=1}^N  \left(Y_{jt} -  \xi\left(\frac{Z_{jt}(\widehat{\beta}) - z}{b_N}\right)'h \right)^2\mathcal{K}_{b_N}\left(\frac{Z_{jt}(\widehat{\beta}) - z}{b_N}\right).
\end{align*}

As $\widehat{\beta} \pconv \beta_0$, the vector $\widehat{h}(z;\widehat{\beta})$ estimates coefficients in a Taylor expansion of degree $\ell$ of the conditional expectation of $Y_t$ given $Z_t(\beta_0) = z$. In particular, the first component of this vector, denoted by $\widehat{h}_1(z;\widehat{\beta}) = e_1' \widehat{h}(z;\widehat{\beta})$, is an estimator of the conditional mean of $Y_t$ given $(X_t'\beta_0,V)$. The vector $\widehat{h}(z;\widehat{\beta})$ is a least-squares estimator and can be written as
       \[ \widehat{h}(z;\widehat\beta) = S_N(z;\widehat{\beta})^{-1} T_N(z;\widehat{\beta}), \]
where
\begin{align*}
   S_N(z;\beta) &= \avgj \xi\left(\frac{Z_{jt}(\beta) - z}{b_N}\right)\xi\left(\frac{Z_{jt}(\beta) - z}{b_N}\right)' \mathcal{K}_{b_N} \left(\frac{Z_{jt}(\beta) - z}{b_N}\right)\\
   T_N(z;\beta) &= \avgj \xi\left(\frac{Z_{jt}(\beta) - z}{b_N}\right)  Y_{jt} \mathcal{K}_{b_N} \left(\frac{Z_{jt}(\beta) - z}{b_N}\right).
\end{align*}

In analogy to equation \eqref{eq:ASF_expression}, we average this conditional mean over the empirical marginal distribution of $V_i$ to obtain the ASF estimator:
\begin{align*}
       \widehat{\text{ASF}}_t(\ubx_t) &= \avg \widehat{h}_1(\ubx_t'\widehat{\beta},V_i ;\widehat{\beta})\hat\pi_{it},
\end{align*}
where $\widehat{h}_1(z;\widehat{\beta}) = e_1' \widehat{h}(z;\widehat{\beta})$ is the first component in $\widehat{h}(z;\widehat{\beta})$,  $\hat\pi_{it} = \1((\ubx_t'\widehat\beta,V_i)\in\mathcal{Z}_t)$ is a trimming function, and $\mathcal{Z}_t$ is an appropriately selected compact set in which the density $f_{Z_t(\beta)}(z)$ is bounded away from zero. This trimming function prevents issues with the invertibility of $S_N(z;\widehat\beta)$. Since $\mathcal{Z}_t$ is a fixed compact set, the parameter that is consistently estimated by $\widehat{\text{ASF}}_t$ is a trimmed ASF defined by
\begin{align*}
       \text{ASF}^\pi_t(\ubx_t) &\equiv \E[\E[Y_t|X_t'\beta_0 = \ubx_t'\beta_0, V]\pi_t]\\
       &= \int_{\mathcal{C}} \E[Y_t|X_t = \ubx_t, C=c] \prob((\ubx_t'\beta_0,V) \in \mathcal{Z}_t| C = c) \, dF_C(c),
\end{align*}
where we let $\pi_{it} =  \1((\ubx_t'\beta_0,V_i) \in \mathcal{Z}_t)$. Note that if 
$(\ubx_t'\beta_0,V) \in \mathcal{Z}_t$ with probability 1, $\text{ASF}^\pi_t(\ubx_t) = \text{ASF}_t(\ubx_t)$ and the trimming does not alter the estimand. By expanding $\mathcal{Z}_t$ along with the sample size at a slow enough rate, which is sometimes called a vanishing, or random, trimming approach, we expect that $\text{ASF}_t(\ubx_t)$ is consistently estimated by $\widehat{\text{ASF}}_t(\ubx_t)$. However, since fixed trimming is often employed in the partial mean literature,\footnote{See, for example, \cite{Newey1994} or more recently \cite{Lee2018}.} we let $\mathcal{Z}_t$ be fixed.

To understand the effect of trimming on the estimand, we consider the scenario where $\prob((\ubx_t'\beta_0,V) \in \mathcal{Z}_t| C = c)$ is bounded away from zero. Formally, assume that $\prob((\ubx_t'\beta_0,V) \in \mathcal{Z}_t| C ) \in [1-\underline{\varepsilon},1]$ with probability 1. Then, if $\text{ASF}^\pi_t(\ubx_t) \geq 0$, we can show that $\text{ASF}^\pi_t(\ubx_t) \in [(1-\underline{\varepsilon})\text{ASF}_t(\ubx_t),\text{ASF}_t(\ubx_t)]$, and thus
\begin{align*}
       \text{ASF}_t(\ubx_t) &\in \left[\text{ASF}^\pi_t(\ubx_t), \frac{\text{ASF}^\pi_t(\ubx_t)}{1-\underline{\varepsilon}}\right]. 
\end{align*}
These bounds are reversed when $\text{ASF}^\pi_t(\ubx_t) < 0$. Note that the bounds collapse to a point as $\underline{\varepsilon}$ approaches zero, and are narrow when $\underline{\varepsilon}$ is small.

We make the following assumptions to obtain the limiting distribution of the ASF. We begin with a standard assumption on the kernel.

\begin{estsec_assump} [Kernel] \label{assn:kernel}
The kernel $\mathcal{K}$ satisfies $\mathcal{K}(z) = K(u) \cdot \prod_{k=1}^{d_V} K(v_k)$ where $K: \R \to \R_{\geq 0}$ such that 
(i) $K(u)$ is equal to zero for all $u$ outside of a compact set, (ii) $K$ is twice continuously differentiable on $\R$ with all these derivatives being Lipschitz continuous, (iii) $\int_{-\infty}^\infty K(u) \, du = 1$, (iv) $K$ is symmetric.
\end{estsec_assump}   

Note that we do not require the use of higher-order kernels in this local polynomial regression.

To state the next assumption precisely, let $\mathcal{C}_m(\mathcal{A})$ denote the set of $m$-times continuously differentiable functions $f : \mathcal{A} \rightarrow \R$. Here $m$ is an integer and $\mathcal{A}$ is a subset of $\R^{1+d_V}$. Denote the differential operator by
       \[\nabla^\lambda = \frac{\partial^{| \lambda |}}{\partial z_1^{\lambda_1} \cdots \partial z_{1+d_V}^{\lambda_{1+d_V}}}, \]
where $\lambda = (\lambda_1,\ldots,\lambda_{1+d_V}) \in \{0,1,\ldots\}^{1+d_V}$ is comprised of nonnegative integers such that $ \sum_{k=1}^{1+d_V} \lambda_k = | \lambda |$. For a given set $\mathcal{A}$, let
\[
       \| f \|_m^{\mathcal{A}} = \max_{|\lambda| \leq m} \sup_{z \in \text{int} (\mathcal{A})} \| \nabla^{\lambda} f(z) \| .
\]

We omit the $\mathcal{A}$ superscript when it does not cause confusion. Next, we impose smoothness and regularity conditions on the distribution of $(Y_t,Z_t(\beta))$ for $\beta$ in a neighborhood of $\beta_0$.

\begin{estsec_assump} [Smoothness] \label{assn:smoothness}
Let $\mathcal{B}_\varepsilon = \{\beta \in \mathcal{B}: \|\beta - \beta_0\| \leq \varepsilon\}$.
\begin{itemize}
       \item [(i)] There exists $\varepsilon > 0$ such that for all $\beta \in \mathcal{B}_\varepsilon$, $Z_t(\beta)$ has a density $f_{Z_t(\beta)}(z)$ with respect to the Lebesgue measure;
       \item [(ii)] $f_{Z_t(\beta)}(z)$ and $\left\|\frac{\partial}{\partial\beta} f_{Z_t(\beta)}(z)\right\|$ are uniformly bounded and uniformly bounded away from zero for $z \in \mathcal{Z}_t$ and $\beta\in \mathcal{B}_\varepsilon$, where $\mathcal{Z}_t$ is a compact set;
       \item [(iii)] $ \left\|f_{Z_t(\beta_0)}(z)\right\|_{\ell + 2}^{\mathcal{Z}_t} <\infty$ and $ \left\| \E[Y_t|Z_t(\beta_0) = z]\right\|_{\ell + 2}^{\mathcal{Z}_t} < \infty$;
       \item [(iv)] $\ubx_t'\beta_0$ is in the interior of $\mathcal{Z}_{1t} \equiv \{e_1'z :z \in \mathcal{Z}_t\}$;
       \item [(v)] $f_{Z_t(\beta_0)|Y_t}(z|y)$ exists and is bounded for $y \in \supp(Y_t)$.
\end{itemize}  
\end{estsec_assump}

Assumptions (i) and (ii) ensure the boundedness and sufficient smoothness of the distribution of $f_{Z_t(\beta)}$ as a function of $\beta$ in a neighborhood of $\beta_0$. Assumption (iii) ensures additional smoothness in $z$ for the distribution of $Z_t(\beta_0)$. The degree of smoothness is linked to the degree of the polynomial in the local polynomial regression. Assumptions (iv) and (v) are standard technical assumptions. We also impose the following moment existence condition.

\begin{estsec_assump} [Moment existence] \label{assn:moment}
       Let $\E[\|X_t\|^2] < \infty$ and $\E[|Y_t|^n] < \infty$ for all $n \in \mathbb{N}$.
\end{estsec_assump}
We can relax the assumption that all moments of $Y_t$ exist at the cost of some additional notation and derivations: see the proof of Lemma \ref{lemma:convASF}.

The following rate conditions govern the bandwidth's convergence rate. Let  $\lceil x\rceil$ denote the smallest integer greater than or equal to $x$.
\begin{estsec_assump} [Bandwidth] \label{assn:ASFrates}
For some $\kappa, \delta >0$, let $b_N = \kappa \cdot N^{-\delta}$ where $\delta$ satisfies 
       \[ \max\left\{\frac{1}{4\left\lceil\frac{\ell + 1}{2}\right\rceil + 1}, 1-2\epsilon\right\} < \delta < \min\left\{\frac{2\epsilon}{3+2d_V}, \frac{1}{1+2d_V}\right\}. \]
\end{estsec_assump} 
A consequence of this assumption is that $\ell$ must increase as $d_V$ increases. In particular, we require $\ell > d_V$ when $\widehat{\beta}$ is $\sqrt{N}$-consistent. 

We can now state the main convergence result for the ASF.

\begin{theorem}[ASF asymptotics] \label{thm:ASF}
Suppose the assumptions for Part 1 of Theorem \ref{thm:ASF-APE_id} hold. Suppose Assumptions B\ref{assn:iid}--B\ref{assn:ASFrates} hold. Then,
\begin{align*}
       \sqrt{Nb_N} \left(\widehat{\text{ASF}}_t(\ubx_t) - \text{ASF}^\pi_t(\ubx_t) \right) &\dconv \mathcal{N}(0,\sigma^2_{\text{ASF}_t}(\ubx_t'\beta_0)),
\end{align*}
where
\begin{align*}
       &\sigma^2_{\text{ASF}_t}(u) = \E\left[ \Var(Y_t|X_t'\beta_0 = u,V)  \frac{f_V(V)}{f_{Z_t(\beta_0)}(u,V)} \1((u,V) \in \mathcal{Z}_t) \right]  \\
       & \cdot e_1'  \left(\int \xi(z)\xi(z)'\mathcal{K}(z) \, dz\right)^{-1}  \left[\int \left(\int \mathcal{K} \left(z\right) \xi \left(z \right) \, dv\right) \left(\int \mathcal{K} \left( z\right) \xi \left(z\right) \, dv\right)'  \, du\right] \left(\int \xi(z)\xi(z)'\mathcal{K}(z) \, dz\right)^{-1} e_1.
\end{align*}

\end{theorem}

To understand the limiting distribution of this estimator, we break down its sampling variation into four separate sources. The terms associated with three of these are asymptotically negligible under our assumptions. We can write
\begin{align*}
       \sqrt{Nb_N} \left(\widehat{\text{ASF}}_t(\ubx_t) - \text{ASF}^\pi_t(\ubx_t) \right) &= \sqrt{Nb_N} \left(\avg \left(\widehat{h}_1(\ubx_t'\widehat{\beta}, V_i; \widehat{\beta}) - \widehat{h}_1(\ubx_t'\widehat{\beta}, V_i; \beta_0)\right)\hat\pi_{it}\right)\\
       &+ \sqrt{Nb_N} \left(\avg \left(\widehat{h}_1(\ubx_t'\widehat{\beta}, V_i; \beta_0) - \widehat{h}_1(\ubx_t'\beta_0, V_i; \beta_0)\right) \hat\pi_{it} \right)\\
       &+ \sqrt{Nb_N} \left(\avg \widehat{h}_1(\ubx_t'\beta_0, V_i; \beta_0)(\hat\pi_{it} - \pi_{it})\right)\\
       &+ \sqrt{Nb_N} \left(\avg \widehat{h}_1(\ubx_t'\beta_0, V_i; \beta_0)\pi_{it} - \E[h_1(\ubx_t'\beta_0, V; \beta_0) \pi_t]\right). 
\end{align*}
The first term reflects the impact of the generated regressors $X_t'\widehat\beta$ being used instead of $X_t'\beta_0$. The bandwidth constraints involving $\epsilon$---the rate of convergence of $\widehat\beta$ to $\beta_0$---ensure this term is asymptotically negligible. The second term reflects the impact of the approximation of the evaluation point $\ubx_t'\beta_0$ by $\ubx_t'\widehat\beta$. Once again, $\epsilon$ plays a crucial role and this term is asymptotically negligible as it is of asymptotic order $O_p(\sqrt{Nb_N} a_N^{-1}) = o_p(1)$ by our assumptions. The third term pertains to the estimation of the trimming function $\pi_{it}$ by $\hat\pi_{it}$. This term is asymptotically dominated due to the superconsistency of $\hat\pi_{it}$ to $\pi_{it}$ uniformly in $i = 1, \ldots, N$. The fourth and final term asymptotically dominates the other three and converges in distribution to a mean-zero Gaussian variable at the $\sqrt{Nb_N}$ rate. Some of the technical tools we use to show this convergence in distribution build on \cite{Masry1996} and \citet*{KongLintonXia2010}.

The rate of convergence of $\widehat{\text{ASF}}_t(\ubx_t)$ when $\epsilon = 1/2$ is  $N^{\delta_{\text{ASF}}}$, where $\delta_{\text{ASF}}$ ranges in the interval $\left(\frac{1+d_V}{3 + 2 d_V}, \frac{1 + \ell}{3 + 2\ell}\right)$. In the case where $d_V = 1$ and $\ell = 2$, this range corresponds to $\left(\frac{2}{5}, \frac{3}{7}\right)$. Recall that 2/5 is the standard rate of convergence of univariate kernel estimation when using second-order kernels. Again, we note that this rate of convergence does not depend on either $T$ or $d_X$. We discuss various implementation details of this estimator and others in Appendix \ref{app_sec:implementation}.

\subsection{Semiparametric Estimation of the APE}\label{sec:ape-est}

We focus here on the case where $X_t^{(k)}$ is continuously distributed. When $X_t^{(k)}$ is discretely distributed, the APE is a difference between two ASFs, in which case Theorem \ref{thm:ASF} can be used to obtain its limiting distribution.

Let $\widehat{h}_2(z; \widehat{\beta}) = \frac{1}{b_N} e_{2+d_V}' \widehat{h}(z;\widehat{\beta})$ denote the $(2 + d_V)$-th  component of the local polynomial regression coefficient vector. By the definition of the above lexicographical order, this is an estimator of the derivative of the conditional mean of $Y_t$ given $(X_t'\beta_0,V) = (u,v)$ with respect to $u$. This estimated derivative is used in the APE estimator, which is defined as
\begin{align*}
       \widehat{\text{APE}}_{k,t}(\ubx_t) &= \widehat{\beta}^{(k)} \cdot \avg \widehat{h}_2(\ubx_t'\widehat{\beta}, V_i; \widehat\beta)\hat\pi_{it},
\end{align*}     
where $\widehat\beta^{(k)}$ denotes the $k$th component of $\widehat\beta$. 

As for the ASF, we use a trimming function in the estimator for technical reasons. Therefore, the estimator is consistent for a trimmed APE defined by $\text{APE}^\pi_{k,t}(\ubx_t) \equiv \E\left[ \frac{\partial}{\partial \ubx_t^{(k)}} \E[Y_t|X_t'\beta_0 = \ubx_t'\beta_0, V] \cdot \pi_t \right]$.
As for the ASF, the untrimmed APE is bounded by $\text{APE}_{k,t}(\ubx_t) \in \left[\text{APE}_{k,t}^\pi(\ubx_t), (1-\underline{\varepsilon})^{-1}\text{APE}^\pi_{k,t}(\ubx_t)\right]$
when $\prob(\ubx_t'\beta_0,V)\in\mathcal{Z}_t|C) \in [1-\underline\varepsilon,1]$ with probability 1 and the APE is positive: the bounds are reversed when it is negative.

The following theorem shows that the APE is $\sqrt{Nb_N^3}$-consistent, where $b_N$ is a bandwidth satisfying Assumption B\ref{assn:ASFrates}. Like the ASF, the APE's rate of convergence does not depend on the dimensions of $\X$.

\begin{theorem}[APE asymptotics] \label{thm:APE}
Suppose the assumptions for Part 2 of Theorem \ref{thm:ASF-APE_id} hold. Suppose Assumptions B\ref{assn:iid}--B\ref{assn:ASFrates} hold. Suppose $X_{t}^{(k)}$ is continuously distributed. Then,
\begin{align*}
       \sqrt{Nb_N^3} \left(\widehat{\text{APE}}_{k,t}(\ubx_t) - \text{APE}^\pi_{k,t}(\ubx_t) \right) &\dconv \mathcal{N}\left(0,(\beta_0^{(k)})^2 \cdot \sigma^2_{\text{APE}_t}(\ubx_t'\beta_0)\right),
\end{align*}
where
\begin{align*}
       &\sigma^2_{\text{APE}_t}(u) = \E\left[ \Var(Y_t|X_t'\beta_0 = u, V)  \frac{f_V(V)}{f_{Z_t(\beta_0)}(u,V)} \1((u, V) \in \mathcal{Z}_t) \right] e_{2 + d_V}'  \left(\int \xi(z)\xi(z)'\mathcal{K}(z) \, dz\right)^{-1} \\
       & \cdot  \left[ \int \left(\int \mathcal{K} \left(z\right) \xi \left(z \right) \, dv\right) \left(\int \mathcal{K} \left( z\right) \xi \left(z\right) \, dv\right)'  du \right] \left(\int \xi(z)\xi(z)'\mathcal{K}(z) \, dz\right)^{-1} e_{2 + d_V}.
\end{align*}
\end{theorem}

\medskip

We can decompose the APE's sample variation into five components. The first four components are analogous to those in the earlier ASF decomposition. In particular, the fourth component is
\begin{align*}
       &\widehat{\beta}^{(k)} \cdot \sqrt{Nb_N^3} \left(\avg \widehat{h}_2(\ubx_t'\beta_0, V_i; \beta_0)\pi_{it} - \E[h_2(\ubx_t'\beta_0, V; \beta_0)\pi_t]\right).
\end{align*}
This component converges in distribution to a mean-zero Gaussian distribution while dominating the other components. The fifth component is due to the presence of $\widehat{\beta}^{(k)}$ and is of the same order as $\sqrt{Nb_N^3}(\widehat{\beta}^{(k)} - \beta_0^{(k)}) = O_p\left(\sqrt{Nb_N^3} a_N\right) = o_p(1)$ by B\ref{assn:ASFrates}. 

The rate of convergence of $\widehat{\text{APE}}_{k,t}(\ubx_t)$ when $\epsilon = 1/2$ is $N^{\delta_{\text{APE}}}$, where $\delta_{\text{APE}}$ ranges in the interval $\left(\frac{d_V}{3 + 2 d_V}, \frac{\ell}{3 + 2\ell}\right)$. When $d_V = 1$ and $\ell = 2$, this range equals $\left(\frac{1}{5}, \frac{2}{7}\right)$. Recall that 2/7 is the standard rate of convergence of derivatives of univariate kernel estimators when using second-order kernels. Our estimator can approach this rate whenever $\ell \geq 2$, i.e., the local polynomial contains quadratic terms.

\subsection{Estimation of the LAR and AME}
The previous analysis focused on the estimation and inference for the ASF and APE using sample analog estimators. Under the assumptions of Theorem \ref{thm:LAR-AME_id}, the LAR and AME are also point identified via a function of the distribution of $(Y,\X)$. Here are their sample analogs:
\begin{align*}
       \widehat{\text{LAR}}_{k,t}(\ubbx) &= \widehat{\beta}^{(k)} \cdot \widehat{h}_2(\ubx_t'\widehat{\beta},v(\ubbx);\widehat{\beta})\\
       \widehat{\text{AME}}_{k,t} &= \avg \widehat{\text{LAR}}_{k,t}(\X_i)\hat\pi_{it} = \widehat{\beta}^{(k)} \cdot \avg \widehat{h}_2(X_{it}'\widehat{\beta},v(\X_i);\widehat{\beta})\hat\pi_{it}.
\end{align*}
Establishing their consistency and asymptotic distribution can be done using the same tools used to establish the same properties for the ASF and APE. Since their proofs are likely similar to those for the ASF and APE, we leave formal asymptotic analyses for future work. The rate of convergence of the LAR estimator should be the same as the nonparametric rate used to estimate $h_2$, while we expect the rate of convergence of the AME to be $\sqrt{N}$ when $\widehat{\beta}^{(k)}$ is $\sqrt{N}$-consistent. This is because the AME averages over all conditioning variables in the local regression of $Y_t$ on $Z_t(\widehat{\beta})$.

\subsection{Estimation with Estimated Indices}

We now briefly consider the estimation of these partial effects under the assumption that $C \;$ $\indep \X|V'\gamma_0$. Following the notations in Appendix \ref{sec:betaid-vgamma}, $V$ is $d_V\times 1$, $\gamma_0$ is $d_V\times d_{V^*}$, and the new index $V'\gamma_0$ is $1 \times d_{V^*}$. Let $\gamma^*_1,\cdots,\gamma^*_{d_{V^*}}$ be the \textit{non-zero elements} in each column of $\gamma_0$, and $V^*_1,\ldots,V^*_{d_{V^*}}$ be the corresponding elements in $V$, and we have $V'\gamma_0 = \left(V^{*\prime}_1 \gamma^*_1, \ldots, V^{*\prime}_{d_{V^*}}\gamma^*_{d_{V^*}}\right).$ Then, $\gamma^*=\left(\gamma^{*\prime}_1,\cdots,\gamma^{*\prime}_{d_{V^*}}\right)'$ contains all unknown elements in $\gamma_0$. 

Suppose $\theta_0 \equiv (\beta_0',\gamma^{*\prime})' $ is consistently estimated. For example, \cite{IchimuraLee1991}'s estimator is $\sqrt{N}$-consistent for $\theta_0$ under their regularity conditions. Let $Z_t(\theta) = (X_t'\beta,V'\gamma) \in \R^{1 + d_{V^*}}$ and let
\begin{align*}
       \widehat{h}(z;\widehat{\theta}) &= \argmin_{h \in \R^{\bar{N}}} \sum_{j=1}^N  \left(Y_{jt} -  \xi\left(\frac{Z_{jt}(\widehat{\theta}) - z}{b_N}\right)'h \right)^2\mathcal{K}_{b_N}\left(\frac{Z_{jt}(\widehat{\theta}) - z}{b_N}\right).
\end{align*}
Then, we can propose the following estimators:
\begin{align*}
       \widehat{\text{ASF}}_t(\ubx_t) &= \avg \widehat{h}_1(\ubx_t'\widehat{\beta},V_i'\widehat{\gamma} ;\widehat\theta)\hat\pi_{it}\\
       \widehat{\text{APE}}_{k,t}(\ubx_t) &= \widehat{\beta}^{(k)} \cdot \avg \widehat{h}_2(\ubx_t'\widehat{\beta}, V_i'\widehat{\gamma}; \widehat\theta)\hat\pi_{it}\\
       \widehat{\text{AME}}_{k,t} &= \widehat{\beta}^{(k)} \cdot \avg \widehat{h}_2(X_{it}'\widehat{\beta}, V_i'\widehat{\gamma}; \widehat\theta)\hat\pi_{it}.
\end{align*}
The indices $V'\gamma_0$ are usually of lower dimension than $V$, which means the function $\E[Y_t|X_t'\beta_0 = \ubx_t'\beta_0, V'\gamma_0 = v'\gamma_0]$ has lower dimension than $\E[Y_t|X_t'\beta_0 = \ubx_t'\beta_0, V = v]$, which helps satisfy the rate assumption B\ref{assn:ASFrates}. This comes at the cost of an additional generated regressor of the form $V_i'\widehat{\gamma}$. From examining Lemmas \ref{lemma:lem1}--\ref{lemma:lem5}, we expect these additional generated regressors do not impact the estimators' limiting distributions, but we leave a detailed asymptotic analysis for future work.

\section{Implementation Details}\label{app_sec:implementation}
\subsection{General Choices}\label{app_sec:gen-choice}

Here are a few practical concerns related to the implementation of our estimators. We explored some of these in more detail in our simulations (Appendix \ref{sec:MonteCarlo}) and empirical illustration (Section \ref{sec:application}).

\paragraph{Local polynomial regression.} First, a common practice in kernel-based methods is the standardization and orthogonalization of the conditioning variables, in our case $Z_t(\widehat{\beta}) = (X_t'\widehat\beta,V)$, before the nonparametric estimation step. The standardization leads to more comparable scales across different components of $Z_t(\widehat{\beta})$. The orthogonalization facilitates the practical implementation of using a product of one-dimensional kernels as our joint kernel, as formulated in Assumption B\ref{assn:kernel}. Also note that the orthogonalization, which can be done via a Cholesky decomposition, is performed on $V$ alone rather than all of $Z_t(\widehat{\beta})$. This is for technical reasons that $\ubx_t'\widehat\beta$ and $V$ should enter in the kernel as a product since the latter is averaged out based on its empirical distribution: see the proofs in Appendix \ref{app_sec:est-proofs}, such as the proof of Lemma \ref{lemma:lem1}.

Second, according to Assumption B\ref{assn:ASFrates}, the required polynomial order increases with $d_V$, the number of continuous index variables. When $d_V$ is 1 or 2, as in our Monte Carlo and empirical illustration, any $\ell \ge 2$ is sufficient. Larger values of $\ell$ reduce the bias in the nonparametric approximation but may cause overfitting, especially in small samples. Our estimates are generally not sensitive to $\ell$ around 2 to 4 in our Monte Carlo simulations and empirical illustration. We use $\ell = 3$ in the Monte Carlo simulations and for estimators conditioning on \(V'\gamma_0\) in the empirical illustration, and use $\ell = 2$ for estimators conditioning on \(V\) in the empirical illustration. The smaller $\ell$ is adopted for the latter case because we divide the observations into cells for discrete index variables, resulting in fewer observations in each cell: see Section \ref{sec:app-back-spec}.

Third, we modified the Gaussian kernel as follows to satisfy Assumption B\ref{assn:kernel}:
\[
K(u) = \begin{cases} \frac 1 {\sqrt{2\pi}}\exp(-u^2/2) &\text{for }|u| \le 5,  \\
        \frac {1} {\sqrt{2\pi}} \exp(-5^2/2) \cdot \left(4(6-|u|)^5-6(6-|u|)^4+3(6-|u|)^3\right) & \text{for } 5<|u|\le 6,  \\
        0 &\text{for }|u|>6.
\end{cases}
\]
This kernel is equivalent to the Gaussian kernel for $|u|\leq 5$ and their results are generally indistinguishable. The truncation at $\pm 6$ ensures the compact support assumed in B\ref{assn:kernel}.(i). The quintic polynomial for $5 < |u|\le 6$ guarantees the twice continuous differentiability assumed in B\ref{assn:kernel}.(ii). 

Fourth, one needs to select a bandwidth $b_N = \kappa \cdot N^{-\delta}$. In principle, we could choose different bandwidths for $\ubx_t'\widehat\beta$ and $V$ (after standardization and orthogonalization), but for simplicity, we keep it the same in the numerical exercises. 

Here we first choose  $\delta$ that satisfies our rate conditions in Assumption B\ref{assn:ASFrates}. We then find the scaling constant $\kappa^*$ using a bootstrap over a finite grid: see Appendix \ref{sec:bw-bs} for details. In our simulations and empirical illustration, $\kappa^*$ usually ranges from 0.6 to 4, and the estimated ASF and APE are generally robust with respect to the scaling constant $\kappa$ within a neighborhood of $\kappa^*$, such as within the range $[\kappa^*-0.2,\,\kappa^*+0.2]$. 

Lastly, the compact set  $\mathcal{Z}_t$ in the trimming function $\hat\pi_{it} = \1((\ubx_t'\widehat\beta,V_i)\in\mathcal{Z}_t)$ helps bound $f_{Z_t(\widehat\beta)}(z)$ away from zero. Candidate criteria could be: a lower bound directly on $\widehat{f}_{Z_t(\widehat\beta)}(z)=\avgj \mathcal{K}_{b_N} \left(\frac{Z_{jt}(\widehat\beta) - z}{b_N}\right)$, an upper bound on the condition number of $S_N(z;\widehat\beta)$, and a lower bound on its determinant. We incorporate all three criteria to construct the trimming set in the numerical implementation. 

\paragraph{Asymptotic variance estimation.} To conduct inference on the ASF and APE, one could, in principle, estimate $\sigma_{\text{ASF}_t}(\ubx_t'\beta_0)$ and $\sigma_{\text{APE}_t}(\ubx_t'\beta_0)$ analytically. This can be done by estimating $\Var(Y_t|X_t'\beta_0 = \ubx_t'\beta_0,V_i)$ via local polynomial regressions of $(Y_t,Y_t^2)$ on $(X_t'\widehat\beta,V)$, and replacing $f_{Z_t}(\ubx_t'\beta_0,V_i)$ by $\avgj \mathcal{K}_{b_N}\left(\frac{Z_{jt}(\widehat{\beta}) - (\ubx_t'\widehat\beta,V_i)}{b_N}\right)$, and $f_V(V_i)$ by  $\avgj \mathcal{K}^V_{b_N}\left(\frac{V_j - V_i}{b_N}\right)$. In the numerical implementation, we instead focus on bootstrap-based inference, which may better capture higher-order terms in the asymptotic expansion of our estimator. 

\paragraph{Multiple time periods.} Finally, note that the above estimator is for the ASF (or APE/AME), at period $t$, which may vary with $t$ in the population. If stationarity is further assumed, i.e., $(g_t,F_{U_t|C}) = (g_{t'},F_{U_{t'}|C})$ for all $t,t' \in \{1,\ldots,T\}$, then $\text{ASF}_t(\ubx_*) = \text{ASF}_{t'}(\ubx_*)$ for any pair of time periods assuming that $\ubx_* \in \supp(X_t) \cap \supp(X_{t'})$. Then, the ASF does not depend on $t$, and we can combine ASF estimates from multiple time periods to obtain a more precise estimate. For example, we can average the estimated ASFs over time:
\begin{align*}
        \overline{\widehat{\text{ASF}}}(\ubx_*) &= \frac{1}{T}\sum_{t=1}^T \widehat{\text{ASF}}_t(\ubx_*).
\end{align*}
We can further reduce its asymptotic variance by selecting weights that depend on $t$. However, weights that minimize the asymptotic variance of the weighted ASF depend on the inverse of an estimate of the joint asymptotic covariance matrix of $\widehat{\text{ASF}}_t(\ubx_*)$ across all $t=1,\cdots,T$. For simplicity, we propose the simple time average as our rule of thumb.

\subsection{Bandwidth Selection via Bootstrap}\label{sec:bw-bs}

Let us take the APE as an example. The bandwidth selection for the ASF and AME can be implemented in a similar fashion. Recall that the bandwidth $b_N$ equals $\kappa \cdot N^{-\delta}$ for a given $\delta > 0$ satisfying our rate conditions. Then, we want to select the tuning parameter $\kappa$ by  minimizing the integrated mean squared error
        \[ \text{IMSE}(\kappa) = \int_{\supp(X_t)} \E\left[\left(\widehat{\text{APE}}_{k,t}(\ubx_t;\kappa N^{-\delta}) - \text{APE}_{k,t}(\ubx_t; 0)\right)^2\right] dF_{X_t}(\ubx_t),\]
where $\widehat{\text{APE}}_{k,t}(\ubx_t;b)$ is our estimated APE with bandwidth $b$, and $\text{APE}_{k,t}(\ubx_t;b)$ denotes the probability limit of $\widehat{\text{APE}}_{k,t}(\ubx_t;b)$ for a fixed bandwidth $b$. Note that $\text{APE}_{k,t}(\ubx_t; 0)$ is the true APE.

Since the IMSE depends on unknown population quantities, we  first approximate \[\E\left[\left(\widehat{\text{APE}}_{k,t}(\ubx_t;\kappa N^{-\delta}) - \text{APE}_{k,t}(\ubx_t; 0)\right)^2\right]\]
 via
        \[  \frac{1}{S} \sum_{s=1}^S \left(\widehat{\text{APE}}_{k,t}^{\ast (s)}(\ubx_t;\kappa N^{-\delta}) - \widehat{\text{APE}}_{k,t}(\ubx_t;\kappa_0 N^{-\delta})\right)^2.\]
Here $\left\{\widehat{\text{APE}}_{k,t}^{\ast (s)}(\ubx_t;b)\right\}_{s=1}^S$ denote $S$ draws of the estimated APE according to its bootstrap distribution. We let $\kappa_0$ be a constant close to 0 and small relative to potential choices of $\kappa$. Note that we cannot set $\kappa_0 = 0$ since the estimated APE is defined only when $\kappa > 0$. We also approximate $F_{X_t}(\ubx_t)$ via the empirical distribution of $X_t$.

More specifically, the bandwidth constant $\kappa$ can be selected according to the following procedure.

\paragraph{Implementation procedure.}
\begin{enumerate}
        \item Generate a range of evaluation points $\ubx_{t,j}$, $j=1,\ldots,J$, with weights $\hat w(\ubx_{t,j})$ determined from the empirical distribution of $X_t$.
        \item Choose $\kappa_0$ to be a small value and estimate $\widehat{\text{APE}}_{k,t}(\ubx_{t,j};\kappa_0 N^{-\delta})$, $j=1,\ldots,J$, based on the original data $\{Y_i,\X_i\}_{i=1}^N$.
        \item Generate bootstrap samples $\{Y^{(s)}_i,\X_i^{(s)}\}_{i=1}^N$ for $s=1,\ldots,S$.
        \item For each bootstrap sample $s = 1,\ldots,S$ and each bandwidth $\kappa$ on grid $\{\kappa_1,\ldots,\kappa_K\}$, calculate $\widehat{\text{APE}}_{k,t}^{\ast (s)}(\ubx_{t,j};\kappa N^{-\delta})$ for $j=1,\ldots,J$.
        \item Choose $\kappa \in \{\kappa_1,\ldots,\kappa_K\}$ that minimizes \[\widehat{\text{IMSE}}(\kappa;w) = \sum_{j=1}^J \left(\frac{1}{S} \sum_{s=1}^S  \left(\widehat{\text{APE}}_{k,t}^{\ast (s)}(\ubx_{t,j};\kappa N^{-\delta}) - \widehat{\text{APE}}_{k,t}(\ubx_{t,j};\kappa_0 N^{-\delta})\right)^2\right) \hat w(\ubx_{t,j}).\]
\end{enumerate}
In the Monte Carlo simulations and empirical illustration, we choose the number of bootstrap samples to be $S=100$. We initialize $\kappa_0$ at 0.6 and increase it by 0.1 if a numerical issue occurs. The bandwidth grid ranges from $\kappa_0$ to 4 with increments of 0.1.

\subsection{Estimated Indices}\label{sec:est-ind}
When the conditioning variable(s) take the form \(V'\gamma_0\), we can implement the following three variations of the semiparametric estimator:
\begin{enumerate}
\item SP: the original three-step estimator.
\begin{enumerate}
\item First, estimate \(\beta_0\) (possibly with smoothed maximum score if $Y_t$ is binary).
\item Second, perform a local polynomial regression of \(Y_{it}\) on \((X_{it}'\widehat\beta,\; V_i)\).
\item Third, average over \(V_i\).
\end{enumerate}
\item SP (\(V'\gamma_0\)): a three-step estimator for estimated indices.
\begin{enumerate}
\item First, estimate \((\beta_0,\gamma_0)\) using \cite{IchimuraLee1991}.
\item Second, perform a local polynomial regression of \(Y_{it}\) on \((X_{it}'\widehat\beta,\; V_i'\widehat\gamma)\).
\item Third, average over \(V_i\).
\end{enumerate}
\item SP (\(V'\gamma_0\), iter.): a four-step estimator for estimated indices.
\begin{enumerate}
\item First, estimate \(\beta_0\) (possibly with smoothed maximum score if $Y_t$ is binary).
\item Second, plug in \(\widehat\beta\) into the objective function in \cite{IchimuraLee1991} to estimate \(\gamma_0\).
\item Third, perform a local polynomial regression of \(Y_{it}\) on \((X_{it}'\widehat\beta,\; V_i'\widehat\gamma)\).
\item Fourth, average over \(V_i\).
\end{enumerate}
\end{enumerate}
Note that: (i) SP (\(V'\gamma_0\)) and SP (\(V'\gamma_0\), iter.) assume the multiple index structure, which is more efficient when the assumption holds but is less robust to misspecification. (ii) SP (\(V'\gamma_0\), iter.) reduces the dimension of numerical optimization in \cite{IchimuraLee1991} and can achieve better numerical performance than SP (\(V'\gamma_0\)) for applications with higher dimensions of parameters.

\section{Extension to a Dynamic Model} \label{sec:extensions}

We now present an extension of our identification results to a dynamic panel model.
\subsection{Related Literature}\label{sec:dyn-panel-lit} 

There is a large literature on dynamic binary response models going back to \cite{Cox1958}. In particular, see \cite{Chamberlain1985}, \cite{Magnac2000}, \cite{HonoreKyriazidou2000}, and \cite{HonoreTamer2006} for results on the identification of common coefficients. For recent results under a logistic error distribution, see \cite{HonoreWeidner2020}, and \cite{Kitazawa2021} for identification results for common coefficients, and \cite{AguirregabiriaCarro2020} and \cite{DobronyiGuKim2021} for other functionals such as AMEs. \cite{KhanPonomarevaTamer2023} establish a number of identification results for $\beta_0$ without assuming logistic errors, ranging from point identification to sharp bounds under minimal assumptions.
\cite{Torgovitsky2019} also obtains partial identification results without parametric restrictions. See \cite{Aristodemou2021} and \cite{KhanOuyangTamer2021} for results on dynamic discrete response models.
Also see \cite{ArellanoBonhomme2017} for a review of nonlinear dynamic panel data models. 

\subsection{Model and Identification}
Our previous Assumption A\ref{assn:model}.(ii) in the main text rules out the dependence of $X_t$ on past $U_{t'}$, thus preventing $\X$ from containing lagged outcome variables. We consider a model that assumes weak or sequential exogeneity. We distinguish between predetermined and exogenous regressors and denote them by $X_t \equiv \begin{pmatrix}
X_{t,\text{pre}} & X_{t,\text{exog}} \end{pmatrix}$. Let $\X_{\text{exog}} = (X_{1,\text{exog}},\ldots,X_{T,\text{exog}})$ denote all past, current, and future values of the exogenous regressor, and let $\X^t_\text{pre} = (X_{1,\text{pre}},\ldots,X_{t,\text{pre}})$ denote all current and past values of the predetermined regressors. We assume that errors are conditionally independent of past, current, and future values of the exogenous regressors, as well as past and current values of the predetermined regressors.

\begin{assump}{A\ref{assn:model}$^\dagger$.(ii)}(Sequential exogeneity)  For each $t = 1,\ldots,T$, $U_t \indep (\X_{\text{exog}},\X^t_\text{pre})|C$.
\end{assump}
This assumption replaces A\ref{assn:model}.(ii) and allows the future covariates to depend on the current error term $U_t$. In particular, it allows for the inclusion of lagged dependent variables in $\X$. 

We also maintain Assumption A\ref{assn:betaid} that states $\beta_0$ is point-identified. When compared to strict exogeneity, point-identification of the common coefficients under sequential exogeneity can be more challenging: see the literature in Appendix \ref{sec:dyn-panel-lit}. To fix ideas, let us consider a relatively simple dynamic binary model as a running example, where the only predetermined regressor is the lagged dependent variable, i.e., $X_{t,\text{pre}} = Y_{t-1}$. For notational simplicity, denote $\widetilde X_{t}=X_{t,\text{exog}}$. Let
\begin{align}
       Y_{it} &= \1(\widetilde X_{it}'\widetilde\beta_0 + \rho_0 Y_{i\, t-1} + C_i - U_{it} \geq 0),\label{eq:mod-dyn-binary-panel}
\end{align}
and define $\beta_0 = (\widetilde{\beta}_0,\rho_0)$.
Versions of this binary outcome model with lagged dependent variables have been studied in \cite{Chamberlain1985} and \cite{HonoreKyriazidou2000}, where they study the identification of $\beta_0$. Its identification generally requires the presence of units whose covariate values do not change over time, known as ``stayers,'' which rules out the inclusion of time dummies in $\widetilde{X}_{t}$. It also requires a minimum number of time periods which is usually greater than 2. As shown in \cite{HonoreKyriazidou2000}, identification of $\beta_0$ can be achieved when $U_{t}$ follows a logistic distribution. Furthermore, identification of $\beta_0$ can still be attained when $U_{t}$ does not follow a logistic distribution, given additional conditions, such as the existence of a continuous regressor with full support. For example, \cite{KhanPonomarevaTamer2023} provide such a result in their Theorem 3.  
 
Given the identification of $\beta_0$, we can make a modified index assumption to help identify partial effects.
\begin{assump}{A\ref{assn:index}$^\dagger$}(Dynamic index sufficiency) For $t \in \{1,\ldots,T\}$, given  $V^t = v_t(\X_{\text{exog}},\X^t_\text{pre}) \in\R^{d_V}$, where $v_t$ is known, let
$C|(\X_{\text{exog}},\X^t_\text{pre}) \deq C|V^t$.
\end{assump}
This assumption replaces A\ref{assn:index} and allows the index to depend on all regressors except for future values of the predetermined regressor. The following theorem shows the identification of our partial effects in these models.

\begin{theorem}[Identification under weak exogeneity]\label{thm:ID_weakexog}
Let $t \in \{1,\ldots,T\}$, $\ubx_t \in \supp(X_t)$, and $\ubbx^t \in \supp(\X_{\text{exog}},\X^t_\text{pre})$. Let Assumptions A\ref{assn:model}--A\ref{assn:index} hold with A\ref{assn:model}$^\dagger$.(ii) replacing A\ref{assn:model}.(ii), and A\ref{assn:index}$^\dagger$ replacing A\ref{assn:index}. Then,
\begin{enumerate}
       \item $\text{ASF}_t(\ubx_t) = \E[\E[Y_t|X_t'\beta_0 = \ubx_t'\beta_0,V^t]]$ is point identified when $\supp(V^t|X_t'\beta_0 = \ubx_t'\beta_0) = \supp(V^t)$;
       \item Let the partial derivative of $\text{ASF}_t(\ubx_t)$ with respect to $\ubx_t^{(k)}$ exist.  $\text{APE}_{k,t}(\ubx_t) = \E[\frac{\partial}{\partial \ubx_t^{(k)}}\E[Y_t|X_t'\beta_0 = \ubx_t'\beta_0,V^t]]$ is point identified when $\supp(V^t|X_t'\beta_0 = u) = \supp(V^t)$ for all $u$ in a neighborhood of $\ubx_t'\beta_0$;
\item $\text{LAR}_{k,t}(\ubbx^t)= \frac{\partial}{\partial \ubx_t^{(k)}} \E[Y_t|X_t'\beta_0 = \ubx_t'\beta_0, V = v] |_{v =v_t(\ubbx^t)}$ is point identified when the derivative exists and when $v_t(\ubbx^t) \in \supp(V^t|X_t'\beta_0 = u)$ for all $u$ in a neighborhood of $\ubx_t'\beta_0$;
                \item $\text{AME}_{k,t}= \E[\frac{\partial}{\partial X_t^{(k)}} \E[Y_t|X_t'\beta_0, V^t]]$ is point identified if the above condition on $\text{LAR}_{k,t}(\ubbx^t)$ holds for all $\ubbx^t \in \supp(\X_{\text{exog}},\X^t_\text{pre})$ up to a $\prob_{\X_{\text{exog}},\X^t_\text{pre}}$-measure zero set.
       
\end{enumerate}
\end{theorem}

In the dynamic binary outcome model in equation \eqref{eq:mod-dyn-binary-panel} above, to identify the ASF at time $t=1$, one can consider an index that depends on $(\widetilde{X}_1,\ldots,\widetilde{X}_T,Y_0)$, where $Y_0$ is the initial time period outcome. Specifically, let $V^1 = (\widetilde{v}(\widetilde\X),Y_0)$, where $\widetilde\X=\X_{\text{exog}}$ for notational simplicity. Assume that $\beta_0$ is identified, for example, from the identification results in \cite{HonoreKyriazidou2000}. Define $\mathcal V^1 = \supp(V^1)$ and $\mathcal V^1(\widetilde{\ubx}_1'\widetilde\beta_0 + \uby_0 \rho_0 ) = \supp(V^1|\widetilde{X}_1'\beta_0 + Y_0 \rho_0 =\widetilde{\ubx}_1'\widetilde\beta_0 + \uby_0 \rho_0 )$. Then,
\begingroup
\allowdisplaybreaks
\begin{align*}
       &\text{ASF}_1(\widetilde{\ubx}_1,\uby_0) =  \E[\1(U_1 \leq \widetilde{\ubx}_1'\widetilde\beta_0 + \uby_0 \rho_0 + C )]\\
       &= \int_{\mathcal V^1} \E[\1(U_1 \leq \widetilde{\ubx}_1'\widetilde\beta_0 + \uby_0 \rho_0 + C )|\tilde{v}(\widetilde{\X}) = v,Y_0 = y_0]\; dF_{\widetilde{v}(\widetilde{\X}),Y_0}(v,y_0)\\
       &= \int_{\mathcal V^1(\widetilde{\ubx}_1'\widetilde\beta_0 + \uby_0 \rho_0 )} \E[\1(U_1 \leq \widetilde{\ubx}_1'\widetilde\beta_0 + \uby_0 \rho_0 + C )|\tilde{v}(\widetilde{\X}) = v,Y_0 = y_0]\; dF_{\widetilde{v}(\widetilde{\X}),Y_0}(v,y_0)\\
       &= \int_{\mathcal V^1(\widetilde{\ubx}_1'\widetilde\beta_0 + \uby_0 \rho_0 )} \E[\1(U_1 \leq \widetilde{\ubx}_1'\widetilde\beta_0 + \uby_0 \rho_0 + C )|\widetilde{X}_1'\beta_0 + Y_0 \rho_0 = \widetilde{\ubx}_1'\widetilde\beta_0 + \uby_0\rho_0, \tilde{v}(\widetilde{\X}) = v,Y_0 = y_0]\; dF_{\widetilde{v}(\widetilde{\X}),Y_0}(v,y_0)\\
       &= \int_{\mathcal V^1(\widetilde{\ubx}_1'\widetilde\beta_0 + \uby_0 \rho_0 )} \E[\1(U_1 \leq \widetilde{X}_1'\widetilde\beta_0 +Y_0 \rho_0 + C )|\widetilde{X}_1'\beta_0 + Y_0 \rho_0 = \widetilde{\ubx}_1'\widetilde\beta_0 + \uby_0\rho_0, \tilde{v}(\widetilde{\X}) = v,Y_0 = y_0]\; dF_{\widetilde{v}(\widetilde{\X}),Y_0}(v,y_0)\\
       &= \E[\E[Y_1|\widetilde{X}_1'\beta_0 + Y_0 \rho_0 = \widetilde{\ubx}_1'\widetilde\beta_0 + \uby_0\rho_0, \tilde{v}(\widetilde{\X}),Y_0 ]].
\end{align*}
\endgroup
The first equality follows from $U_1 \indep (\widetilde{X}_1,Y_0)|C$, the second from iterated expectations, and the third from the support assumption in the statement of Theorem \ref{thm:ID_weakexog}. The fourth follows from $(C,U_1) \indep (\widetilde{X}_1'\beta_0 + Y_0 \rho_0)|V^{1}$, which is implied by Assumptions A\ref{assn:model}$^\dagger$.(ii) and A\ref{assn:index}$^\dagger$, and the proof is similar to step 1 in the proof of Theorem \ref{thm:ASF-APE_id}. Finally, the last two equalities follow directly.

One can also identify the APE or AME under the appropriate support conditions on the index variables. 
Finding sufficient index variables in dynamic models is potentially more delicate than in static ones because the exchangeability of covariates across time is an unlikely justification in dynamic models. We leave an analysis of this task for future work.

\clearpage
\section{Monte Carlo Simulations} \label{sec:MonteCarlo}

We conduct two sets of Monte Carlo simulation experiments based on binary panel data models with the conditioning variable being \(V\) in Case 1 and \(V'\gamma_0\) in Case 2. We focus on the former while streamlining the discussion of the latter, as their main messages are similar. Both cases account for two key features: multidimensional index variables and flexible error distributions. For Monte Carlo simulations with logistic errors, please see the previous version of this paper \citep{liu2021identification}.

\subsection{Case 1: Conditioning on \(V\)}

The Monte Carlo design is summarized in Table \ref{tbl:sim-gen-dgp}. Note that both $X_{t}$ and $V$ are 2-by-1 vectors. Covariates $X^{(k)}_{t},\;k=1,2,$ are drawn from a bivariate standard normal distribution, which satisfies the support conditions in Theorem \ref{thm:ASF-APE_id}. Our choices of $N=1500$ and $T=10$ are directly comparable with the dataset in our empirical illustration on female labor force participation, in which $N = 1461$ and $T = 9$. We use ``DGP xy'' to indicate the data-generating process (DGP) with $f_{C|V}$ being type x and $f_{U_t}$ being type y. The distribution of individual effects, $f_{C|V}$, is skewed in DGP 1y and bimodal in DGP 2y.\footnote{Many empirical applications feature skewed and/or multimodal distributions of unobserved individual heterogeneity.  For example, \cite{liu2018density} estimated the latent productivity distribution of young firms, which exhibits a long right tail since good ideas are scarce. Also, \citet*{FisherJensenTkac2019}  found two modes in the underlying skill distribution of mutual fund management---a primary mode with average ability and a secondary mode with poor performance.} For the error term, we consider error distributions $f_{U_t}$ that exhibit skewness in DGP x1 and fat-tails in DGP x2. 

\begin{table}
       \caption{Monte Carlo Design - Case 1}
       \label{tbl:sim-gen-dgp}
       \begin{center}
               \begin{tabular}{ll} \hline \hline
                               Model:& $Y_{it} = \1\left(X_{it}'\beta_0 + C_i - U_{it} \geq 0 \right)$ \\
Common param.:&$\beta_0=(1,2)'$\\
                               Covariates:& $X_{it}\sim \mathcal{N}\left(0_{2\times 1},I_2\right)$ \\
                               Index:& $V_i=\frac 1 T\sum_{t=1}^T X_{it} $ \\
                               
                               Sample Size:& $N=1500$, $T=10$ \\
                               \# Repetitions:& $N_{sim}=100$ \\ \hline
                               $f_{C|V}$:&\\
                               \hspace{1em}DGP 1y, skewed: & $C_i|V_i \sim \left(\sum_{k=1}^2\left(V_{i}^{(k)}\right)^2+1\right)\cdot \mathcal{SN}(0,1,10)$\\
\hspace{1em}DGP 2y, bimodal: & $C_i|V_i \sim \frac 1 2 \mathcal{N}\left(\sum_{k=1}^2\left(V_{i}^{(k)}\right)^2+2,\;1\right)+\frac 1 2 \mathcal{N}\left(-\sum_{k=1}^2\left(V_{i}^{(k)}\right)^2-2,\;1\right)$\\ \vspace{-0.51em}
\\

\multicolumn{2}{l}{$f_{U_t}$, with $\mathbb{E}\left(U_{t}\right)=0$ and $\Var\left(U_{t}\right)=1$:}\\
\hspace{1em}DGP x1, skewed: & $U_{it} \sim \frac{1}{9}\mathcal{N}\left(2,\frac{1}{2}\right)+\frac{8}{9}\mathcal{N}\left(-\frac{1}{4},\frac{1}{2}\right)$\\
\hspace{1em}DGP x2, fat-tailed: & $U_{it} \sim \frac{1}{5}\mathcal{N}\left(0,4\right)+\frac{4}{5}\mathcal{N}\left(0,\frac{1}{4}\right)$\\
                               \hline
                               \\
                       \end{tabular}

\includegraphics[width=0.8\textwidth]{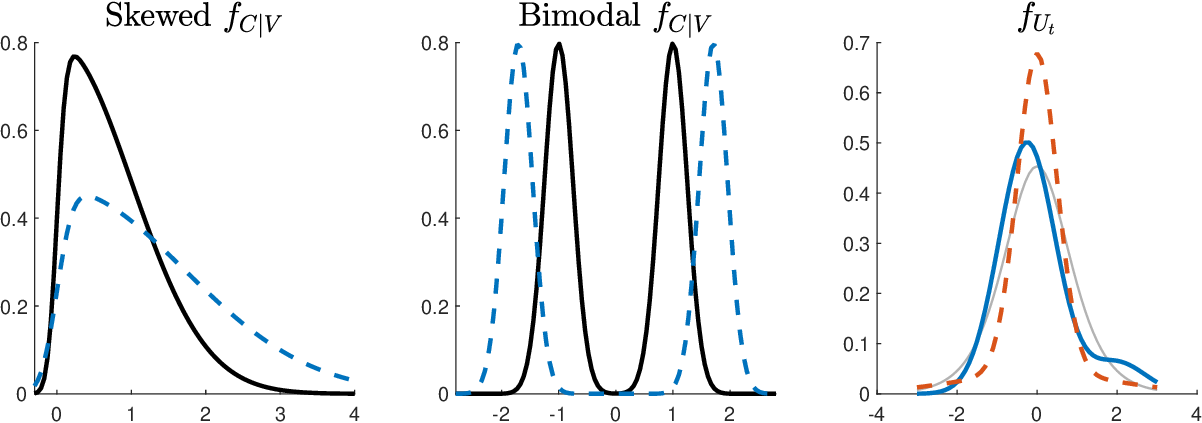}
                      
       \end{center}
       {\footnotesize {\em Notes: }$\mathcal{SN}(\xi,\omega,\alpha)$ denotes a skewed normal distribution with location parameter $\xi$, scale parameter $\omega$, and shape parameter $\alpha$, and its pdf is given by $f(x)={\frac  {2}{\omega }}\phi \left({\frac  {x-\xi }{\omega }}\right)\Phi \left(\alpha \left({\frac  {x-\xi }{\omega }}\right)\right)$, where $\phi(\cdot)$ and $\Phi(\cdot)$ denote the pdf and cdf of a standard normal distribution. The two left panels depict $f_{C|V}$. The black solid and blue dashed lines are conditional on $\sqrt{\sum_{k=1}^2\left(V_{i}^{(k)}\right)^2}=0$ and $0.5$, respectively. The rightmost panel depicts $f_{U_t}$. The blue solid and red dashed lines are $f_{U_t}$ in DGPs G.x1 (skewed) and G.x2 (fat-tailed), respectively. For reference, the thin gray line plots a rescaled logistic distribution with zero mean and unit variance. 
        }\setlength{\baselineskip}{4mm}
\end{table}

We evaluate the estimated ASF and APE based on a collection of $\ubx_t=\left(\ubx_t^{(1)},\ubx_t^{(2)}\right)'$. We fix $\ubx_t^{(1)}$ at its population mean (i.e., $\ubx_t^{(1)}=0$) and vary $\ubx_t^{(2)}\in[-1,\;1]$, which covers 68\% of the distribution of $X^{(2)}_{t}$. Given the non-logistic error distributions, we estimate $\beta_0$ using a smoothed maximum score estimator as in \cite{CharlierMelenbergSoest1995} and \cite{Kyriazidou1995}, employing a fourth-order cdf kernel to satisfy the bandwidth requirement in Assumption B\ref{assn:ASFrates}. We normalize $|\widehat\beta^{(1)}|=1$ since the identification of $\beta_0$ is up to scale. We use a local cubic regression (i.e., polynomial order $\ell = 3$) to estimate the conditional expectation of $Y_{t}$ evaluated at $(\ubx_t'\widehat\beta,V)$. Finally, given the DGPs, the ASFs and APEs do not change over time, so we average the estimated ASFs and APEs across time periods. See Section \ref{app_sec:gen-choice} for more details.

Figure \ref{fig:sim-gen-ape-hair} compares the estimated APEs to the true APEs based on 100 Monte Carlo repetitions in each setup, and Figure \ref{fig:sim-gen-ape-stat} plots the biases, standard deviations, and root mean square errors (RMSEs). Figures \ref{fig:sim-gen-asf-hair} and \ref{fig:sim-gen-asf-stat} show corresponding graphs for the ASF estimates. We see that the proposed semiparametric estimator better captures the peak in the skewed case and the valley in the bimodal case, whereas the RE and CRE reverse the valley in the bimodal case due to their parametric restrictions. As expected, the semiparametric estimator generates smaller biases and larger standard deviations than the RE and CRE. The improvement in bias dominates the deterioration in standard deviation for most covariate values in all setups. The difference between the RE and CRE is relatively negligible---their parametric assumptions in $f_{C|V}$ seem too restrictive and lead to considerable misspecification biases given current DGPs.

In Table \ref{tbl:sim-gen-ape}, the first three columns summarize the APE estimator's performance by computing weighted averages of biases, standard deviations, and RMSEs across the collection of evaluation points $\ubx_t$ with weights proportional to $f_{X_{t}}(\ubx_t)$. Similar to what we observed in Figures \ref{fig:sim-gen-ape-hair} and \ref{fig:sim-gen-ape-stat}, the semiparametric estimator yields the smallest RMSE in all cases. The last three columns present the minimum, median, and maximum of the ratios of $\text{RMSE}(\ubx_t)$ to the true $\text{APE}(\ubx_t)$. The minimum, median, and maximum are taken over the collection of evaluation points $\ubx_t$. We see that the ratios range between 2.5\% and 120\% across all setups. Therefore, the RMSEs are generally sizeable compared to the true APEs, and thus the more precise semiparametric estimator would indeed make a significant difference. The RE and CRE have lower \textit{minimal} ratios, which occurs at $\ubx_t$'s where the grey bands ``intersect'' with true APE curves; at the same time, the semiparametric estimator largely reduces the \textit{median} and \textit{maximal} ratios. For example, in DGP 22, the median (maximal) ratio of the semiparametric estimator is less than 1/3 (1/4) of its RE and CRE counterparts.

We also examine the estimation of the common parameter and the ASF in Table \ref{tbl:sim-gen-beta-asf}. The structure of the ASF part of the table is the same as Table \ref{tbl:sim-gen-ape} for the APE. The ratios of $\text{RMSE}(\ubx_t)$ to the true $\text{ASF}(\ubx_t)$ are generally smaller than their APE counterparts, and the semiparametric estimator again dominates the RE and CRE.

For $\widehat\beta$, the nonparametric smoothed maximum score estimator produces less biased but noisier estimates, and their RMSEs are larger than those of the RE and CRE. Nevertheless, the semiparametric estimator still better traces the shapes of the ASFs, hence providing the most accurate ASF estimates. Its RMSEs are around or less than half that of the RE and CRE. To take a closer look at how the $\beta_0$ estimation affects the APE estimation, we further examine an infeasible semiparametric estimator with known $\beta_0$ (see Table 8 in the previous version of this paper, \cite{liu2021identification}). Results show that the smoothed maximum score estimates of $\beta_0$ slightly increase the absolute value of the bias, the standard deviation, and the RMSE, but the difference is minor---the flexible semiparametric estimator of the APE partially absorbs the effect of the slightly imprecisely estimated $\beta_0$.
\clearpage
\begin{figure}
       \caption{Estimated APEs vs True APEs - Monte Carlo Case 1}
       \label{fig:sim-gen-ape-hair}
       \begin{center}
       \begin{tabular}{cc}
       &\textbf{Semiparam.\hspace{1.25in}RE\hspace{1.55in}CRE}\\
       \rotatebox{90}{\hspace{.6in}\textbf{DGP 11}}
       &\adjustbox{trim={0.08\width} {0\height} {0.08\width} {0\height},clip}
       {\includegraphics[width=0.99\textwidth]{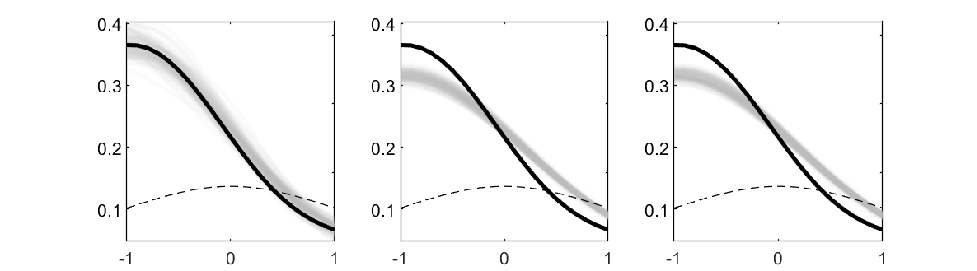}}\\
       \rotatebox{90}{\hspace{.6in} \textbf{DGP 12}}
       &\adjustbox{trim={0.08\width} {0\height} {0.08\width} {0\height},clip}
       {\includegraphics[width=0.99\textwidth]{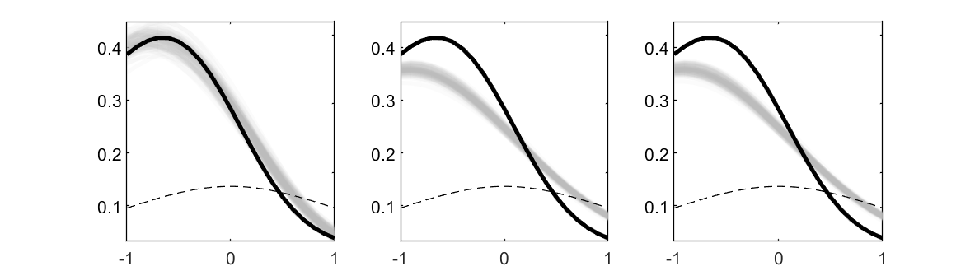}}\\
       \rotatebox{90}{\hspace{.6in} \textbf{DGP 21}}
       &\adjustbox{trim={0.08\width} {0\height} {0.08\width} {0\height},clip}
       {\includegraphics[width=0.99\textwidth]{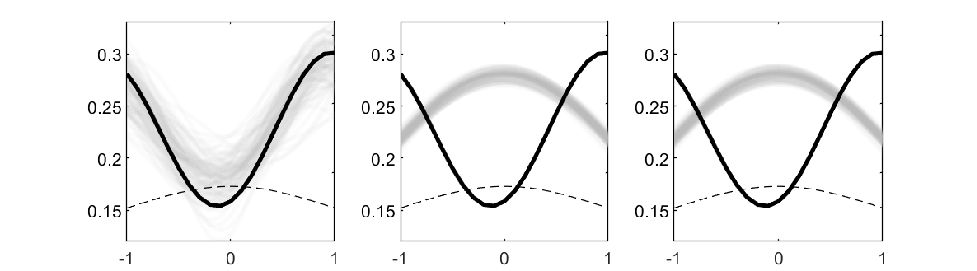}}\\
       \rotatebox{90}{\hspace{.6in} \textbf{DGP 22}}
       &\adjustbox{trim={0.08\width} {0\height} {0.08\width} {0\height},clip}
       {\includegraphics[width=0.99\textwidth]{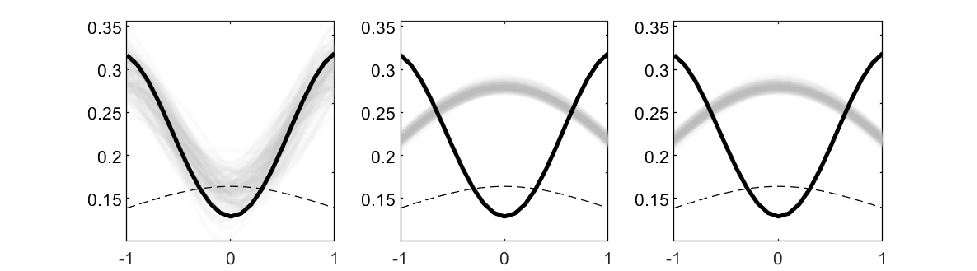}}
       \end{tabular}
       \end{center}
       {\footnotesize {\em Notes:} The x-axes are potential values $\ubx_t^{(2)}$. The black solid lines are the true APEs. The gray bands are collections of lines where each line corresponds to the estimated APE based on one simulation repetition. The thin dashed lines at the bottom of all panels show $f_{X_{t}^{(2)}}\left(\ubx_t^{(2)}\right)$. 
        }\setlength{\baselineskip}{4mm}                                     
\end{figure}

\begin{figure}
       \caption{Bias, Standard Deviation, and RMSE in APE Estimation - Monte Carlo Case 1}
       \label{fig:sim-gen-ape-stat}
       \begin{center}
       \begin{tabular}{cc}
       &\textbf{Semiparam.\hspace{1.35in}RE\hspace{1.6in}CRE}\vspace{1em}\\
       \rotatebox{90}{\hspace{.6in}\textbf{DGP 11}}
       &\adjustbox{trim={0\width} {0\height} {0\width} {0\height},clip}
       {\includegraphics[width=.85\textwidth]{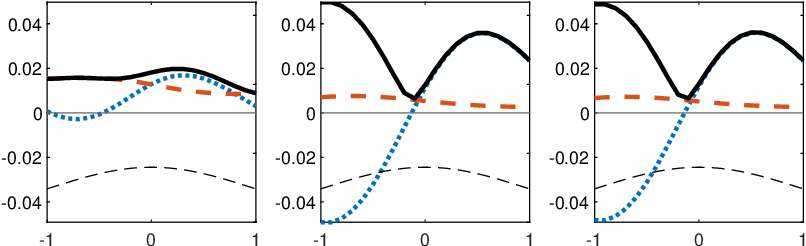}}\vspace{1em}\\
       \rotatebox{90}{\hspace{.6in} \textbf{DGP 12}}
       &\adjustbox{trim={0\width} {0\height} {0\width} {0\height},clip}
       {\includegraphics[width=.85\textwidth]{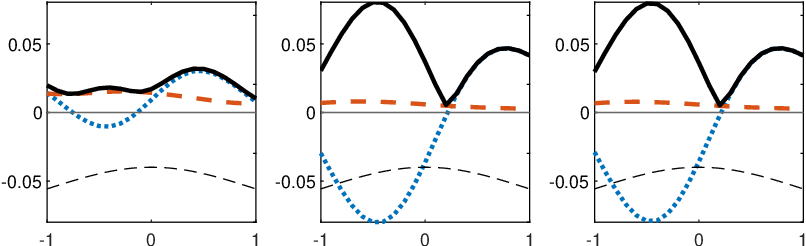}}\vspace{1em}\\
       \rotatebox{90}{\hspace{.6in} \textbf{DGP 21}}
       &\adjustbox{trim={0\width} {0\height} {0\width} {0\height},clip}
       {\includegraphics[width=.85\textwidth]{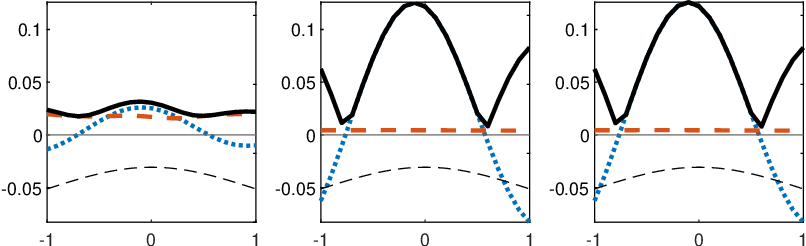}}\vspace{1em}\\
       \rotatebox{90}{\hspace{.6in} \textbf{DGP 22}}
       &\adjustbox{trim={0\width} {0\height} {0\width} {0\height},clip}
       {\includegraphics[width=.85\textwidth]{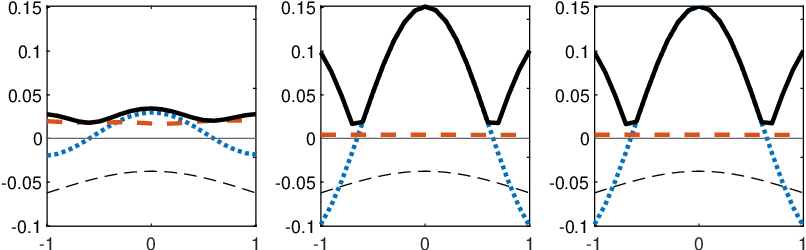}}\vspace{1em}
       \end{tabular}
       \end{center}
       {\footnotesize {\em Notes:} The x-axes are potential values $\ubx_t^{(2)}$. The black solid / blue dotted / red dashed lines represent the RMSEs / biases / standard deviations of the APE estimates. The thin dashed lines at the bottom of all panels show $f_{X_{t}^{(2)}}\left(\ubx_t^{(2)}\right)$. 
        }\setlength{\baselineskip}{4mm}                                     
\end{figure}

\begin{table}
       \caption{APE Estimation - Monte Carlo Case 1}
       \label{tbl:sim-gen-ape}
       \begin{center}
   \begin{tabular}{llrrrrrr} \\  \hline \hline
& &      \textbar Bias\textbar & SD & RMSE & Min & Med. & Max \\ \hline
\multirow{3}{*}{DGP 11}
& Semiparam. &   0.013 &   0.012 &   \textbf{0.016} &     4.2\% &     8.4\% &    15.7\%  \\ 
& RE         &   0.028 &   0.005 &   0.029 &     2.7\% &    13.6\% &    39.1\%  \\ 
& CRE        &   0.028 &   0.005 &   0.029 &     2.7\% &    13.3\% &    39.1\%  \\ 
&&&&&&&\\ 
\multirow{3}{*}{DGP 12}
& Semiparam. &   0.018 &   0.012 &   \textbf{0.020} &     3.3\% &     6.0\% &    35.2\%  \\ 
& RE         &   0.047 &   0.006 &   0.047 &     2.6\% &    18.3\% &   107.5\%  \\ 
& CRE        &   0.046 &   0.006 &   0.047 &     2.5\% &    18.4\% &   107.3\%  \\ 
&&&&&&&\\ 
\multirow{3}{*}{DGP 21}
& Semiparam. &   0.019 &   0.018 &   \textbf{0.023} &     7.2\% &     8.5\% &    20.5\%  \\ 
& RE         &   0.071 &   0.004 &   0.071 &     3.1\% &    23.8\% &    81.5\%  \\ 
& CRE        &   0.071 &   0.004 &   0.071 &     3.0\% &    23.7\% &    81.7\%  \\ 
&&&&&&&\\ 
\multirow{3}{*}{DGP 22}
& Semiparam. &   0.022 &   0.019 &   \textbf{0.026} &     7.4\% &     9.3\% &    26.6\%  \\ 
& RE         &   0.086 &   0.004 &   0.086 &     6.3\% &    31.1\% &   116.9\%  \\ 
& CRE        &   0.086 &   0.004 &   0.086 &     6.2\% &    31.0\% &   117.2\%  \\     

\hline
   \end{tabular}       
\end{center}
       {\footnotesize {\em Notes:} \textbar Bias\textbar\ indicates the absolute value of the bias. The reported \textbar Bias\textbar, SD, and RMSE are weighted averages across the collection of evaluation points $\ubx_t$, where the weights are proportional to $f_{X_{t}}(\ubx_t)$. The bold entries indicate the best estimator (i.e., with the smallest RMSE) for each DGP. The last three columns are the minimum/median/maximum of $\text{RMSE}(\ubx_t)/\text{APE}(\ubx_t)\times100\%$ over $\ubx_t$.}\setlength{\baselineskip}{4mm}
\end{table}

\begin{figure}
       \caption{Estimated ASFs vs True ASFs - Monte Carlo Case 1}
       \label{fig:sim-gen-asf-hair}
       \begin{center}
       \begin{tabular}{cc}
       &\textbf{Semiparam.\hspace{1.25in}RE\hspace{1.55in}CRE}\\
       \rotatebox{90}{\hspace{.6in}\textbf{DGP 11}}
       &\adjustbox{trim={0.08\width} {0\height} {0.08\width} {0\height},clip}
       {\includegraphics[width=0.99\textwidth]{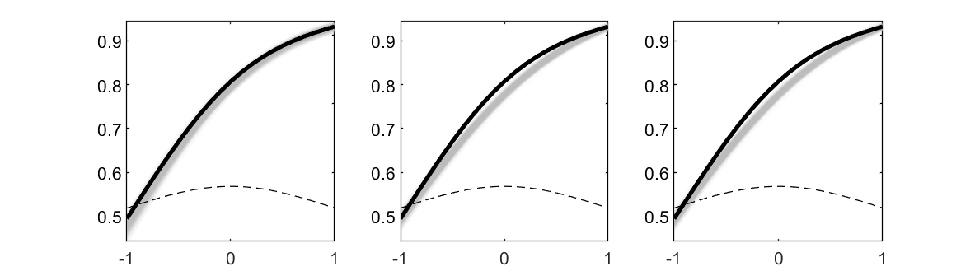}}\\
       \rotatebox{90}{\hspace{.6in} \textbf{DGP 12}}
       &\adjustbox{trim={0.08\width} {0\height} {0.08\width} {0\height},clip}
       {\includegraphics[width=0.99\textwidth]{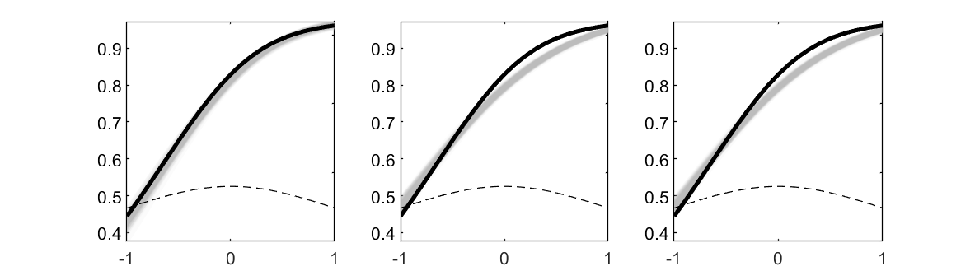}}\\
       \rotatebox{90}{\hspace{.6in} \textbf{DGP 21}}
       &\adjustbox{trim={0.08\width} {0\height} {0.08\width} {0\height},clip}
       {\includegraphics[width=0.99\textwidth]{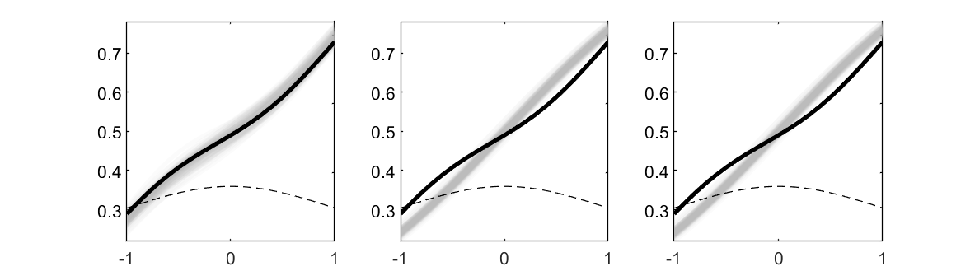}}\\
       \rotatebox{90}{\hspace{.6in} \textbf{DGP 22}}
       &\adjustbox{trim={0.08\width} {0\height} {0.08\width} {0\height},clip}
       {\includegraphics[width=0.99\textwidth]{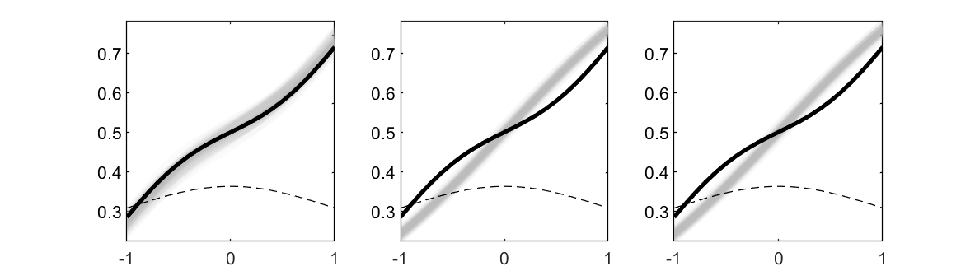}}
               
       \end{tabular}
       \end{center}
       {\footnotesize {\em Notes:} The x-axes are potential values $\ubx_t^{(2)}$. The black solid lines are the true ASFs. The gray bands are collections of lines where each line corresponds to the estimated ASF based on one simulation repetition. The thin dashed lines at the bottom of all panels show $f_{X_{t}^{(2)}}\left(\ubx_t^{(2)}\right)$. 
        }\setlength{\baselineskip}{4mm}                                     
\end{figure}

\begin{figure}
       \caption{Bias, Standard Deviation, and RMSE in ASF Estimation - Monte Carlo Case 1}
       \label{fig:sim-gen-asf-stat}
       \begin{center}
       \begin{tabular}{cc}
       &\textbf{Semiparam.\hspace{1.35in}RE\hspace{1.6in}CRE}\vspace{1em}\\
       \rotatebox{90}{\hspace{.6in}\textbf{DGP 11}}
       &\adjustbox{trim={0\width} {0\height} {0\width} {0\height},clip}
       {\includegraphics[width=.85\textwidth]{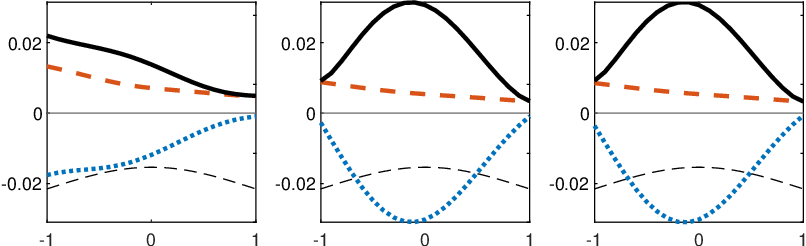}}\vspace{1em}\\
       \rotatebox{90}{\hspace{.6in} \textbf{DGP 12}}
       &\adjustbox{trim={0\width} {0\height} {0\width} {0\height},clip}
       {\includegraphics[width=.85\textwidth]{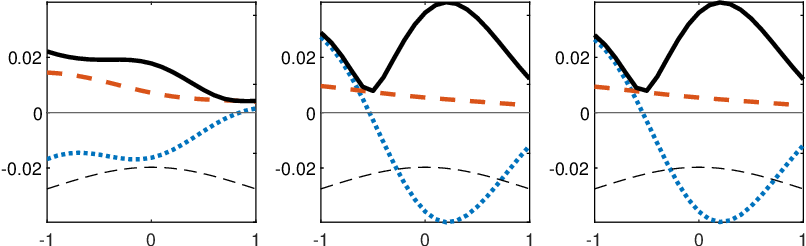}}\vspace{1em}\\
       \rotatebox{90}{\hspace{.6in} \textbf{DGP 21}}
       &\adjustbox{trim={0\width} {0\height} {0\width} {0\height},clip}
       {\includegraphics[width=.85\textwidth]{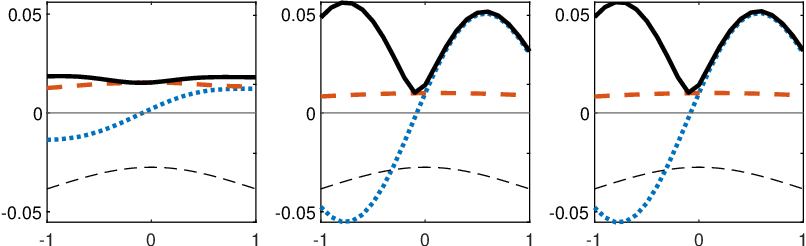}}\vspace{1em}\\
       \rotatebox{90}{\hspace{.6in} \textbf{DGP 22}}
       &\adjustbox{trim={0\width} {0\height} {0\width} {0\height},clip}
       {\includegraphics[width=.85\textwidth]{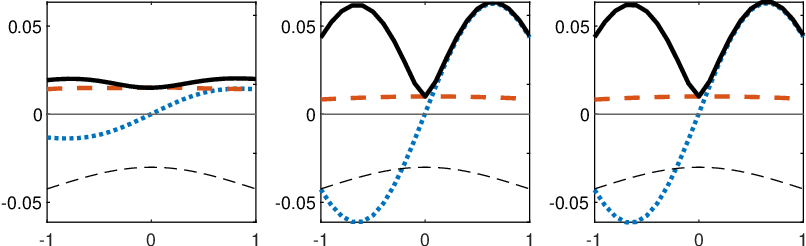}}\vspace{1em}
       \end{tabular}
       \end{center}
       {\footnotesize {\em Notes:} The x-axes are potential values $\ubx_t^{(2)}$. The black solid / blue dotted / red dashed lines represent the RMSEs / biases / standard deviations of the ASF estimates. The thin dashed lines at the bottom of all panels show $f_{X_{t}^{(2)}}\left(\ubx_t^{(2)}\right)$. 
        }\setlength{\baselineskip}{4mm}   
                                          
\end{figure}

\begin{table}
       \caption{Estimation of Common Parameter and ASF - Monte Carlo Case 1}
       \label{tbl:sim-gen-beta-asf}
       \begin{center}
   \begin{tabular}{ll|rrr|rrrrrr} \\  \hline \hline   
   &  &\multicolumn{3}{c|}{$\widehat\beta^{(2)}$} &\multicolumn{6}{c}{ASF}\\
& &      Bias & SD & RMSE &\textbar Bias\textbar & SD & RMSE & Min & Med. & Max \\ \hline
\multirow{3}{*}{DGP 11}
& Semiparam. &   0.011 &   0.031 &   0.033 &   0.011 &   0.008 &   \textbf{0.013} &     0.5\% &     1.7\% &     4.4\%  \\ 
& RE         &   0.004 &   0.023 &   0.023 &   0.020 &   0.006 &   0.021 &     0.4\% &     2.8\% &     4.1\%  \\ 
& CRE        &   0.005 &   0.022 &   0.023 &   0.020 &   0.006 &   0.021 &     0.4\% &     2.8\% &     4.2\%  \\ 
&&&&&&&&&&\\ 
\multirow{3}{*}{DGP 12}
& Semiparam. &   0.012 &   0.026 &   0.028 &   0.013 &   0.008 &   \textbf{0.015} &     0.4\% &     2.1\% &     5.0\%  \\ 
& RE         &   0.005 &   0.019 &   0.020 &   0.025 &   0.006 &   0.026 &     1.2\% &     3.4\% &     6.5\%  \\ 
& CRE        &   0.006 &   0.019 &   0.020 &   0.025 &   0.006 &   0.026 &     1.2\% &     3.4\% &     6.3\%  \\ 
&&&&&&&&&&\\ 
\multirow{3}{*}{DGP 21}
& Semiparam. &   0.015 &   0.064 &   0.065 &   0.014 &   0.014 &   \textbf{0.017} &     2.5\% &     3.2\% &     6.4\%  \\ 
& RE         &   0.007 &   0.041 &   0.042 &   0.037 &   0.010 &   0.038 &     2.2\% &     7.7\% &    16.8\%  \\ 
& CRE        &   0.008 &   0.043 &   0.043 &   0.037 &   0.010 &   0.039 &     2.2\% &     7.7\% &    16.9\%  \\ 
&&&&&&&&&&\\ 
\multirow{3}{*}{DGP 22}
& Semiparam. &   0.011 &   0.072 &   0.073 &   0.014 &   0.015 &   \textbf{0.018} &     2.8\% &     3.3\% &     6.8\%  \\ 
& RE         &   0.004 &   0.043 &   0.043 &   0.044 &   0.010 &   0.045 &     2.0\% &     9.5\% &    16.9\%  \\ 
& CRE        &   0.005 &   0.044 &   0.044 &   0.044 &   0.010 &   0.045 &     2.0\% &     9.5\% &    17.0\%  \\
\hline
   \end{tabular}       
\end{center}
       {\footnotesize {\em Notes:}  For the RE and CRE, we normalize $\widehat\beta$ such that $|\widehat\beta^{(1)}|=1$ to ensure comparability across estimators. \textbar Bias\textbar\ indicates the absolute value of the bias. The reported \textbar Bias\textbar, SD, and RMSE of the ASF are weighted averages across the collection of evaluation points $\ubx_t$, where the weights are proportional to $f_{X_{t}}(\ubx_t)$. The bold entries indicate the best ASF estimator (i.e., with the smallest RMSE) for each DGP. The last three columns are the minimum/median/maximum of $\text{RMSE}(\ubx_t)/\text{ASF}(\ubx_t)\times100\%$ over $\ubx_t$.}\setlength{\baselineskip}{4mm}
\end{table}

\clearpage
\subsection{Conditioning on \(V'\gamma_0\), Estimated Indices}\label{sec:appendix-sim}
The Monte Carlo design is described in Table \ref{tbl:sim-gen-dgp-2}, which is modified from Case 1. Now the distributions of individual effects, $f_{C|V}$, depend on a linear combination of \(V\).
Here we consider the three variations of the semiparametric estimator discussed in Appendix \ref{sec:est-ind}: SP, SP (\(V'\gamma_0\)), and SP (\(V'\gamma_0\), iter.).
In the current setup, there is no misspecification for all three variations of the semiparametric estimator.

Figure \ref{fig:sim-gen-ape-hair-2} shows the estimated APEs, and Figure \ref{fig:sim-gen-ape-stat-2} depicts their corresponding biases, standard deviations, and RMSEs. Similarly, Figures \ref{fig:sim-gen-asf-hair-2} and \ref{fig:sim-gen-asf-stat-2} present the estimated ASFs and their corresponding statistics. Table \ref{tbl:sim-gen-ape-2} reports these statistics for the APE estimators, and Table \ref{tbl:sim-gen-beta-asf-2} for the common parameter and the ASF. In terms of estimation performance, the differences across the three variations of the semiparametric estimator are relatively small and, similar to Case 1, they dominate the RE and CRE.

\begin{table}[t]
       \caption{Monte Carlo Design - Case 2}
       \label{tbl:sim-gen-dgp-2}
       \begin{center}
               \begin{tabular}{ll} \hline \hline
                               Model:& $Y_{it} = \1\left(X_{it}'\beta_0 + C_i - U_{it}>0 \right)$ \\
Common param.:&$\beta_0=(1,2)'$, $\gamma_0 = (1,1)'$\\
                               Covariates:& $X_{it}\sim \mathcal{N}\left(0_{2\times 1},I_2\right)$ \\
                               Index:& $V_i'\gamma_0 = 
                               \frac 1 T \sum_{t=1}^T X_{it}'\gamma_0 $ \\
                               
                               Sample Size:& $N=1500$, $T=10$ \\
                               \# Repetitions:& $N_{sim}=100$ \\ \hline
                               $f_{C|V}$:&\\
                               \hspace{1em}DGP 1y, skewed: & $C_i|V_i \sim V_i'\gamma_0+ \left((V_i'\gamma_0)^2+1\right)\cdot \mathcal{SN}(0,1,10)$\\
\hspace{1em}DGP 2y, bimodal: & $C_i|V_i \sim V_i'\gamma_0  +\frac 1 2 \mathcal{N}\left(\left(V_i'\gamma_0\right)^2+2,1\right)+\frac 1 2 \mathcal{N}\left(-\left(V_i'\gamma_0\right)^2-2,1\right)$\\ \vspace{-0.5em}
\\

\multicolumn{2}{l}{$f_{U_t}$, with $\mathbb{E}\left(U_{it}\right)=0$ and $\Var\left(U_{it}\right)=1$:}\\
\hspace{1em}DGP x1, skewed: & $U_{it} \sim \frac{1}{9}\mathcal{N}\left(2,\frac{1}{2}\right)+\frac{8}{9}\mathcal{N}\left(-\frac{1}{4},\frac{1}{2}\right)$\\
\hspace{1em}DGP x2, fat-tailed: & $U_{it} \sim \frac{1}{5}\mathcal{N}\left(0,4\right)+\frac{4}{5}\mathcal{N}\left(0,\frac{1}{4}\right)$\\
                               \hline
                               \\
                       \end{tabular}
                      
       \end{center}
       {\footnotesize {\em Notes: }See the description in Table~\ref{tbl:sim-gen-dgp}.
        }\setlength{\baselineskip}{4mm}
\end{table}
\begin{figure}[p]
       \caption{Estimated APEs vs True APEs - Monte Carlo Case 2}
       \label{fig:sim-gen-ape-hair-2}
       \begin{center}
       \begin{tabular}{cc}
       &\textbf{SP ($V'\gamma_0$, iter.)\hspace{1.2in}RE\hspace{1.55in}CRE}\\
       \rotatebox{90}{\hspace{.6in}\textbf{DGP 11}}
       &\adjustbox{trim={0.08\width} {0\height} {0.08\width} {0\height},clip}
       {\includegraphics[width=0.99\textwidth]{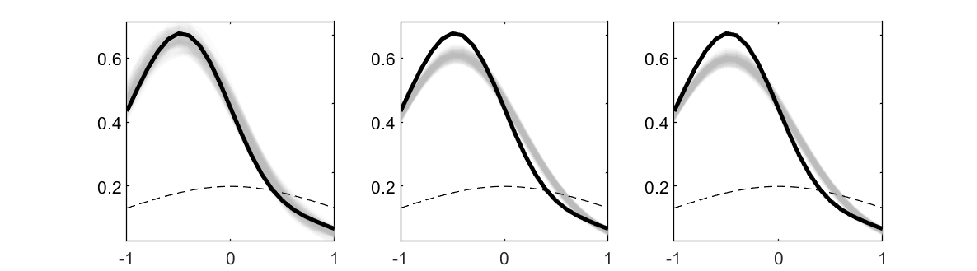}}\\
       \rotatebox{90}{\hspace{.6in} \textbf{DGP 12}}
       &\adjustbox{trim={0.08\width} {0\height} {0.08\width} {0\height},clip}
       {\includegraphics[width=0.99\textwidth]{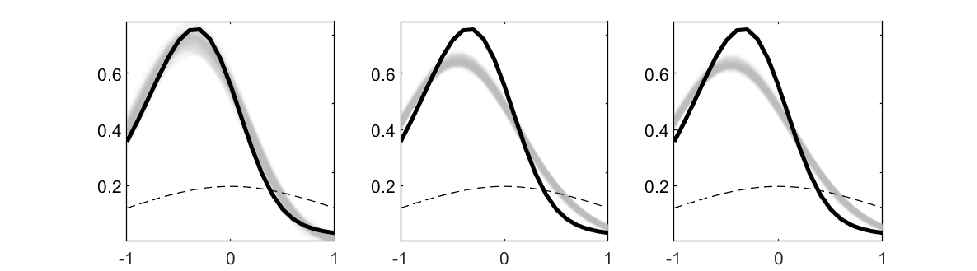}}\\
       \rotatebox{90}{\hspace{.6in} \textbf{DGP 21}}
       &\adjustbox{trim={0.08\width} {0\height} {0.08\width} {0\height},clip}
       {\includegraphics[width=0.99\textwidth]{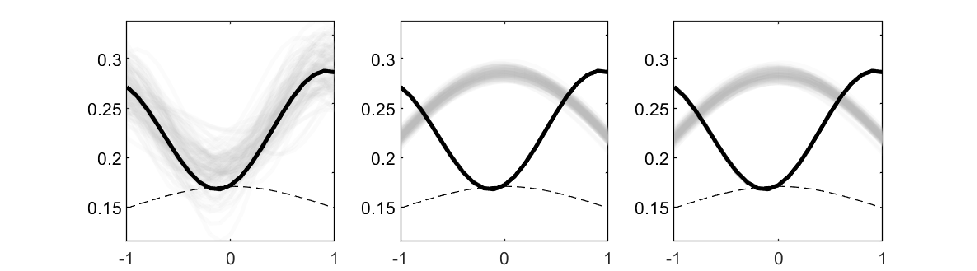}}\\
       \rotatebox{90}{\hspace{.6in} \textbf{DGP 22}}
       &\adjustbox{trim={0.08\width} {0\height} {0.08\width} {0\height},clip}
       {\includegraphics[width=0.99\textwidth]{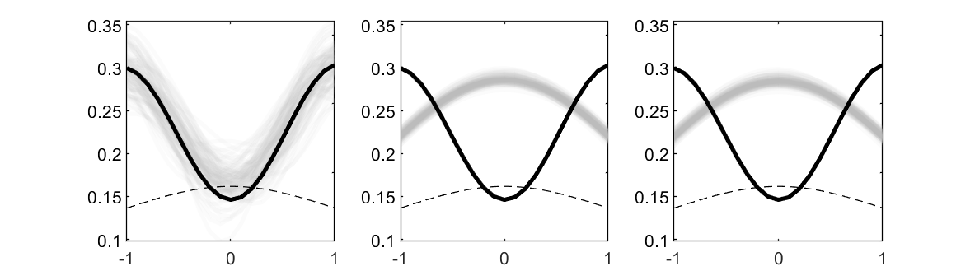}}
               
       \end{tabular}
       \end{center}
       {\footnotesize {\em Notes:} The x-axes are potential values $\ubx_t^{(2)}$. The black solid lines are the true APEs. The gray bands are collections of lines where each line corresponds to the estimated APE based on one simulation repetition. The thin dashed lines at the bottom of all panels show $f_{X_{t}^{(2)}}\left(\ubx_t^{(2)}\right)$. 
        }\setlength{\baselineskip}{4mm}                                     
\end{figure}

\begin{figure}[p]
       \caption{Bias, Standard Deviation, and RMSE in APE Estimation - Monte Carlo Case 2}
       \label{fig:sim-gen-ape-stat-2}
       \begin{center}
       \begin{tabular}{cc}
       &\textbf{SP ($V'\gamma_0$, iter.)\hspace{1.3in}RE\hspace{1.6in}CRE}\vspace{1em}\\
       \rotatebox{90}{\hspace{.6in}\textbf{DGP 11}}
       &\adjustbox{trim={0\width} {0\height} {0\width} {0\height},clip}
       {\includegraphics[width=.85\textwidth]{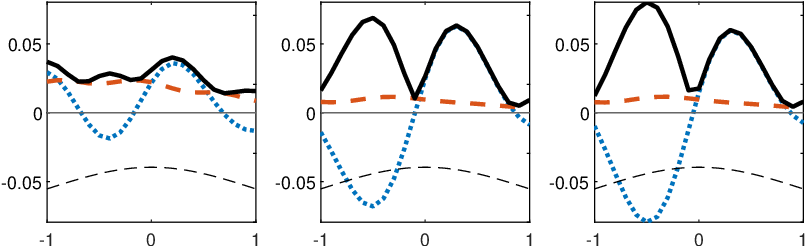}}\vspace{1em}\\
       \rotatebox{90}{\hspace{.6in} \textbf{DGP 12}}
       &\adjustbox{trim={0\width} {0\height} {0\width} {0\height},clip}
       {\includegraphics[width=.85\textwidth]{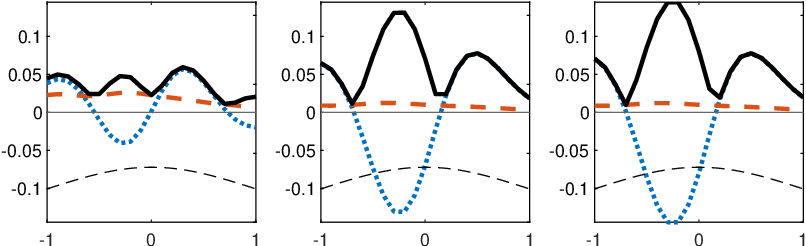}}\vspace{1em}\\
       \rotatebox{90}{\hspace{.6in} \textbf{DGP 21}}
       &\adjustbox{trim={0\width} {0\height} {0\width} {0\height},clip}
       {\includegraphics[width=.85\textwidth]{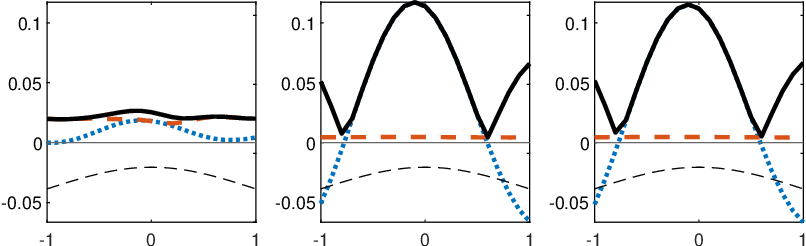}}\vspace{1em}\\
       \rotatebox{90}{\hspace{.6in} \textbf{DGP 22}}
       &\adjustbox{trim={0\width} {0\height} {0\width} {0\height},clip}
       {\includegraphics[width=.85\textwidth]{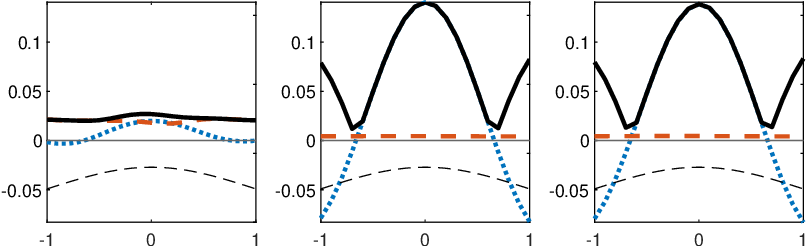}}\vspace{1em}
       \end{tabular}
       \end{center}
       {\footnotesize {\em Notes:} The x-axes are potential values $\ubx_t^{(2)}$. The black solid / blue dotted / red dashed lines represent the RMSEs / biases / standard deviations of the APE estimates. The thin dashed lines at the bottom of all panels show $f_{X_{t}^{(2)}}\left(\ubx_t^{(2)}\right)$. 
        }\setlength{\baselineskip}{4mm}   
                                          
\end{figure}
\begin{figure}[p]
       \caption{Estimated ASFs vs True ASFs - Monte Carlo Case 2}
       \label{fig:sim-gen-asf-hair-2}
       \begin{center}
       \begin{tabular}{cc}
       &\textbf{SP ($V'\gamma_0$, iter.)\hspace{1.2in}RE\hspace{1.55in}CRE}\\
       \rotatebox{90}{\hspace{.6in}\textbf{DGP 11}}
       &\adjustbox{trim={0.08\width} {0\height} {0.08\width} {0\height},clip}
       {\includegraphics[width=0.99\textwidth]{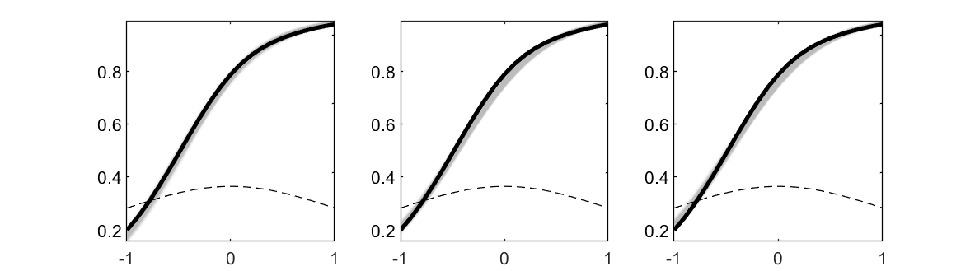}}\\
       \rotatebox{90}{\hspace{.6in} \textbf{DGP 12}}
       &\adjustbox{trim={0.08\width} {0\height} {0.08\width} {0\height},clip}
       {\includegraphics[width=0.99\textwidth]{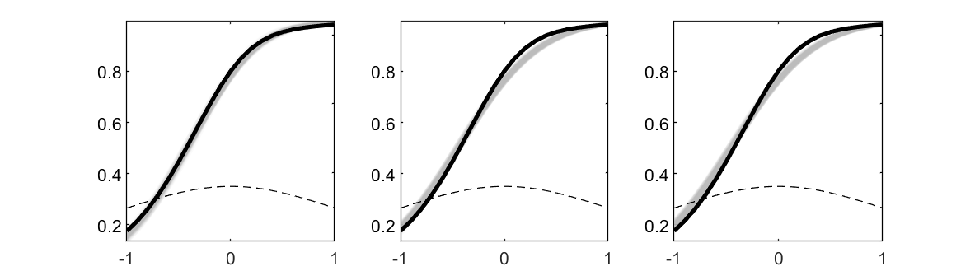}}\\
       \rotatebox{90}{\hspace{.6in} \textbf{DGP 21}}
       &\adjustbox{trim={0.08\width} {0\height} {0.08\width} {0\height},clip}
       {\includegraphics[width=0.99\textwidth]{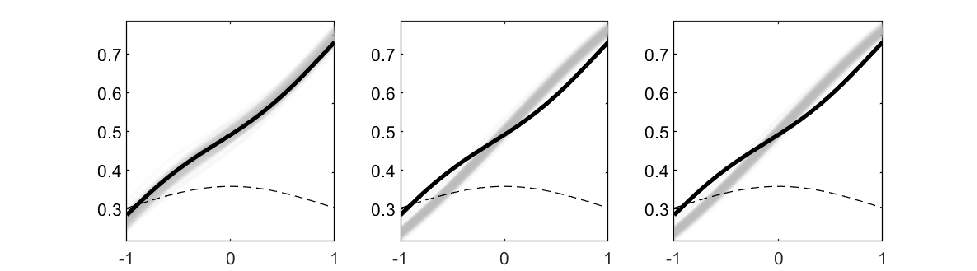}}\\
       \rotatebox{90}{\hspace{.6in} \textbf{DGP 22}}
       &\adjustbox{trim={0.08\width} {0\height} {0.08\width} {0\height},clip}
       {\includegraphics[width=0.99\textwidth]{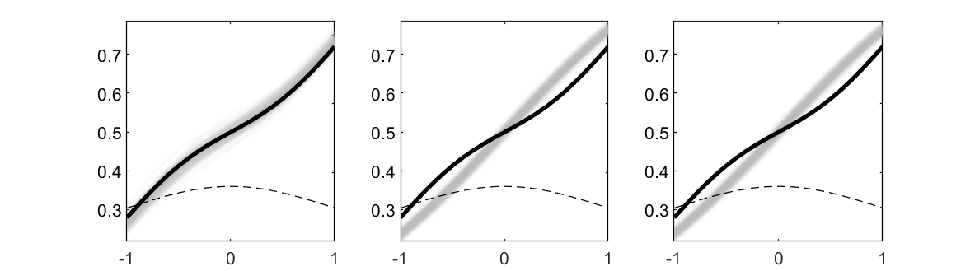}}
               
       \end{tabular}
       \end{center}
       {\footnotesize {\em Notes:} The x-axes are potential values $\ubx_t^{(2)}$. The black solid lines are the true ASF. The gray bands are collections of lines where each line corresponds to the estimated ASF based on one simulation repetition. The thin dashed lines at the bottom of all panels show $f_{X_{t}^{(2)}}\left(\ubx_t^{(2)}\right)$.
        }\setlength{\baselineskip}{4mm}                                     
\end{figure}

\begin{figure}[p]
       \caption{Bias, Standard Deviation, and RMSE in ASF Estimation - Monte Carlo Case 2}
       \label{fig:sim-gen-asf-stat-2}
       \begin{center}
       \begin{tabular}{cc}
       &\textbf{SP ($V'\gamma_0$, iter.)\hspace{1.3in}RE\hspace{1.6in}CRE}\vspace{1em}\\
       \rotatebox{90}{\hspace{.6in}\textbf{DGP 11}}
       &\adjustbox{trim={0\width} {0\height} {0\width} {0\height},clip}
       {\includegraphics[width=.85\textwidth]{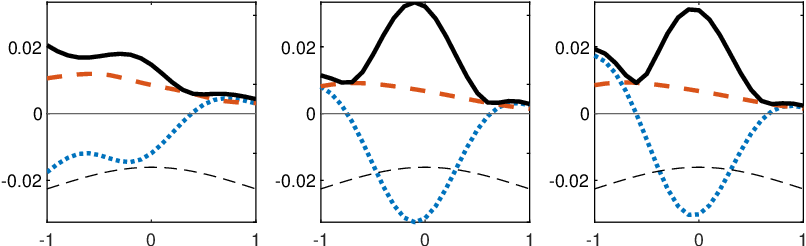}}\vspace{1em}\\
       \rotatebox{90}{\hspace{.6in} \textbf{DGP 12}}
       &\adjustbox{trim={0\width} {0\height} {0\width} {0\height},clip}
       {\includegraphics[width=.85\textwidth]{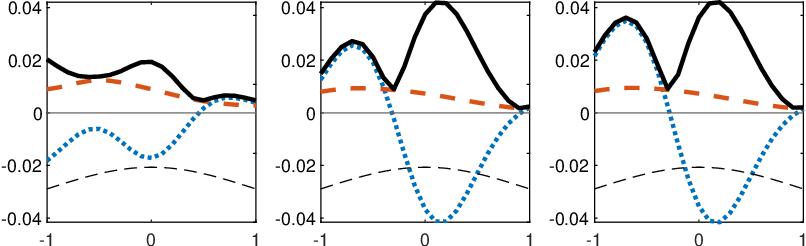}}\vspace{1em}\\
       \rotatebox{90}{\hspace{.6in} \textbf{DGP 21}}
       &\adjustbox{trim={0\width} {0\height} {0\width} {0\height},clip}
       {\includegraphics[width=.85\textwidth]{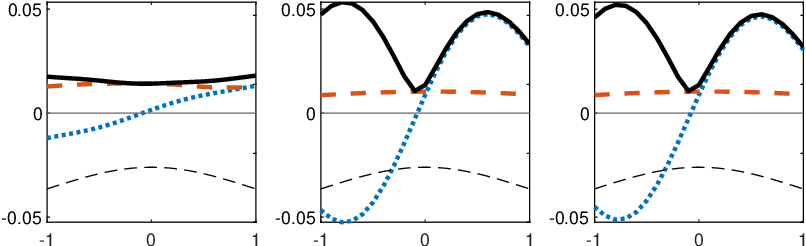}}\vspace{1em}\\
       \rotatebox{90}{\hspace{.6in} \textbf{DGP 22}}
       &\adjustbox{trim={0\width} {0\height} {0\width} {0\height},clip}
       {\includegraphics[width=.85\textwidth]{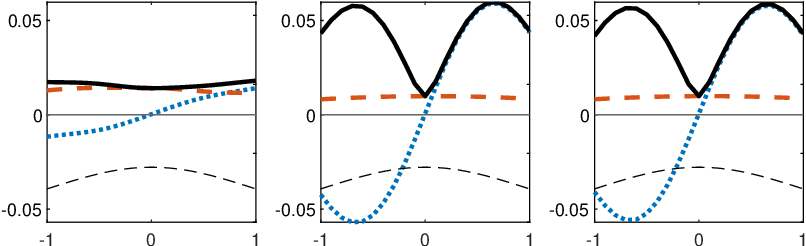}}\vspace{1em}
       \end{tabular}
       \end{center}
       {\footnotesize {\em Notes:} The x-axes are potential values $\ubx_t^{(2)}$. The black solid / blue dotted / red dashed lines represent the RMSEs / biases / standard deviations of the ASF estimates. The thin dashed lines at the bottom of all panels show $f_{X_{t}^{(2)}}\left(\ubx_t^{(2)}\right)$. 
        }\setlength{\baselineskip}{4mm}   
                                          
\end{figure}

\begin{table}[t]
       \caption{APE Estimation - Monte Carlo Case 2}
       \label{tbl:sim-gen-ape-2}
       \begin{center}
   \begin{tabular}{llrrrrrr} \\  \hline \hline
& &      \textbar Bias\textbar & SD & RMSE & Min & Med. & Max \\ \hline
\multirow{5}{*}{DGP 11}
& SP ($V'\gamma_0$, iter.) &   0.023 &   0.019 &   0.027 &     3.5\% &     8.5\% &    23.6\%  \\ 
& SP ($V'\gamma_0$)        &   0.021 &   0.023 &   \textbf{0.026} &     4.1\% &     7.8\% &    23.3\%  \\ 
& SP             &   0.024 &   0.019 &   0.028 &     3.1\% &     8.4\% &    30.3\%  \\ 
& RE             &   0.039 &   0.008 &   0.040 &     2.0\% &     9.4\% &    31.3\%  \\ 
& CRE            &   0.041 &   0.008 &   0.042 &     2.8\% &    10.6\% &    31.1\%  \\ 
&&&&&&&\\ 
\multirow{5}{*}{DGP 12}
& SP ($V'\gamma_0$, iter.) &   0.032 &   0.019 &   0.036 &     3.4\% &    11.6\% &    66.5\%  \\ 
& SP ($V'\gamma_0$)        &   0.024 &   0.024 &   \textbf{0.030} &     3.9\% &     8.7\% &    59.2\%  \\ 
& SP             &   0.033 &   0.019 &   0.037 &     2.8\% &    10.2\% &    78.9\%  \\ 
& RE             &   0.064 &   0.009 &   0.065 &     2.0\% &    17.3\% &    98.6\%  \\ 
& CRE            &   0.069 &   0.009 &   0.070 &     1.7\% &    19.2\% &   100.6\%  \\ 
&&&&&&&\\ 
\multirow{5}{*}{DGP 21}
& SP ($V'\gamma_0$, iter.) &   0.018 &   0.019 &   0.022 &     7.0\% &     9.0\% &    15.8\%  \\ 
& SP ($V'\gamma_0$)        &   0.017 &   0.020 &   \textbf{0.021} &     7.6\% &     9.0\% &    13.5\%  \\ 
& SP             &   0.018 &   0.018 &   0.022 &     6.9\% &     7.8\% &    17.1\%  \\ 
& RE             &   0.064 &   0.005 &   0.065 &     1.7\% &    20.9\% &    69.7\%  \\ 
& CRE            &   0.063 &   0.005 &   0.064 &     2.0\% &    20.2\% &    68.6\%  \\ 
&&&&&&&\\ 
\multirow{5}{*}{DGP 22}
& SP ($V'\gamma_0$, iter.) &   0.019 &   0.020 &   \textbf{0.023} &     6.8\% &     9.4\% &    18.4\%  \\ 
& SP ($V'\gamma_0$)        &   0.019 &   0.023 &   0.024 &     7.9\% &    11.1\% &    16.4\%  \\ 
& SP             &   0.019 &   0.019 &   \textbf{0.023} &     7.3\% &     8.7\% &    20.9\%  \\ 
& RE             &   0.078 &   0.004 &   0.078 &     4.5\% &    26.3\% &    96.1\%  \\ 
& CRE            &   0.077 &   0.004 &   0.077 &     5.0\% &    26.6\% &    94.8\%  \\

\hline
   \end{tabular}       
\end{center}
       {\footnotesize {\em Notes:} \textbar Bias\textbar\ indicates the absolute value of the bias. The reported \textbar Bias\textbar, SD, and RMSE are weighted averages across the collection of evaluation points $\ubx_t$, where the weights are proportional to $f_{X_{t}}(\ubx_t)$. The bold entries indicate the best estimator (i.e., with the smallest RMSE) for each DGP. The last three columns are the minimum/median/maximum of $\text{RMSE}(\ubx_t)/\text{APE}(\ubx_t)\times100\%$ over $\ubx_t$.}\setlength{\baselineskip}{4mm}
\end{table}

\begin{table}[p]
       \caption{Estimation of Common Parameter and ASF - Monte Carlo Case 2}
       \label{tbl:sim-gen-beta-asf-2}
       \begin{center}
   \begin{tabular}{ll|rrr|rrrrrr} \\  \hline \hline   
   &  &\multicolumn{3}{c|}{$\widehat\beta^{(2)}$} &\multicolumn{6}{c}{ASF}\\
& &      Bias & SD & RMSE &\textbar Bias\textbar & SD & RMSE & Min & Med. & Max \\ \hline
\multirow{5}{*}{DGP 11}
& SP ($V'\gamma_0$, iter.) &   0.019 &   0.050 &   0.053 &   0.011 &   0.008 &   \textbf{0.012} &     0.4\% &     1.9\% &    10.3\%  \\ 
& SP ($V'\gamma_0$)        &   0.006 &   0.068 &   0.068 &   0.018 &   0.010 &   0.020 &     0.5\% &     1.7\% &    24.7\%  \\ 
& SP             &   0.019 &   0.050 &   0.053 &   0.019 &   0.009 &   0.021 &     0.6\% &     2.4\% &    24.8\%  \\ 
& RE             &  -0.030 &   0.036 &   0.047 &   0.016 &   0.006 &   0.017 &     0.3\% &     2.9\% &     5.8\%  \\ 
& CRE            &   0.013 &   0.038 &   0.040 &   0.015 &   0.006 &   0.016 &     0.2\% &     2.7\% &     9.8\%  \\ 
&&&&&&&&&&\\ 
\multirow{5}{*}{DGP 12}
& SP ($V'\gamma_0$, iter.) &   0.006 &   0.042 &   0.042 &   0.011 &   0.008 &   \textbf{0.013} &     0.5\% &     2.4\% &    11.5\%  \\ 
& SP ($V'\gamma_0$)        &   0.006 &   0.067 &   0.067 &   0.019 &   0.010 &   0.021 &     0.4\% &     2.1\% &    26.9\%  \\ 
& SP             &   0.006 &   0.042 &   0.042 &   0.020 &   0.009 &   0.022 &     0.6\% &     3.1\% &    27.9\%  \\ 
& RE             &  -0.036 &   0.030 &   0.047 &   0.022 &   0.006 &   0.023 &     0.2\% &     3.9\% &     9.6\%  \\ 
& CRE            &   0.005 &   0.032 &   0.032 &   0.024 &   0.006 &   0.025 &     0.2\% &     3.6\% &    13.7\%  \\ 
&&&&&&&&&&\\ 
\multirow{5}{*}{DGP 21}
& SP ($V'\gamma_0$, iter.) &   0.010 &   0.063 &   0.064 &   0.012 &   0.013 &   0.015 &     2.5\% &     2.9\% &     6.2\%  \\ 
& SP ($V'\gamma_0$)        &  -0.057 &   0.120 &   0.132 &   0.011 &   0.013 &   \textbf{0.014} &     2.0\% &     2.9\% &     4.5\%  \\ 
& SP             &   0.010 &   0.063 &   0.064 &   0.014 &   0.015 &   0.018 &     2.9\% &     3.1\% &     7.7\%  \\ 
& RE             &  -0.007 &   0.041 &   0.042 &   0.035 &   0.010 &   0.037 &     2.2\% &     7.3\% &    16.7\%  \\ 
& CRE            &   0.003 &   0.042 &   0.042 &   0.035 &   0.010 &   0.036 &     2.2\% &     7.1\% &    16.2\%  \\ 
&&&&&&&&&&\\ 
\multirow{5}{*}{DGP 22}
& SP ($V'\gamma_0$, iter.) &   0.002 &   0.070 &   0.069 &   0.012 &   0.013 &   0.016 &     2.5\% &     2.8\% &     6.2\%  \\ 
& SP ($V'\gamma_0$)        &  -0.058 &   0.110 &   0.123 &   0.011 &   0.014 &   \textbf{0.015 }&     2.2\% &     2.8\% &     5.3\%  \\ 
& SP             &   0.002 &   0.070 &   0.069 &   0.014 &   0.015 &   0.018 &     3.0\% &     3.1\% &     7.7\%  \\ 
& RE             &  -0.010 &   0.043 &   0.044 &   0.041 &   0.009 &   0.043 &     2.0\% &     8.8\% &    16.4\%  \\ 
& CRE            &   0.000 &   0.044 &   0.044 &   0.041 &   0.009 &   0.042 &     2.0\% &     8.6\% &    16.0\%  \\ 

\hline
   \end{tabular}       
\end{center}
       {\footnotesize {\em Notes:}  For the RE and CRE, we normalize $\widehat\beta$ such that $|\widehat\beta^{(1)}|=1$ to ensure comparability across estimators. \textbar Bias\textbar\ indicates the absolute value of the bias. The reported \textbar Bias\textbar, SD, and RMSE of the ASF are weighted averages across the collection of evaluation points $\ubx_t$, where the weights are proportional to $f_{X_{t}}(\ubx_t)$. The bold entries indicate the best ASF estimator (i.e., with the smallest RMSE) for each DGP. The last three columns are the minimum/median/maximum of $\text{RMSE}(\ubx_t)/\text{ASF}(\ubx_t)\times100\%$ over $\ubx_t$.}\setlength{\baselineskip}{4mm}
\end{table}

\clearpage
\section{Additional Figures and Tables for the Empirical Illustration}\label{sec:appendix-app}

Figure \ref{fig:app-dist-obs} plots the distributions of the covariates, and Table \ref{tbl:app-descriptive-stats} summarizes the corresponding descriptive statistics.

Figure \ref{fig:app-dtime} depicts the estimated coefficients on time dummies which capture time-variation in aggregate participation rates. Point estimates of the time profiles are generally parallel to each other (from top to bottom: the smoothed maximum score, RE, and CRE) and yield higher participation rates after 1983, which coincides with the beginning of the Great Moderation. Most of the time-variation within each estimator and differences across estimators are insignificant at the 10\% level, and standard errors generally increase with time for all three estimators. The smoothed maximum score yields the widest confidence band, as expected.

Figure \ref{fig:app-asf-ape-alt} plots the estimated ASF and APE based on alternative specifications. In the benchmark specification for the results in the main paper, we construct the indices based on the initial value of the covariates \(X_{i1}\) and use our original three-step semiparametric estimator without estimated indices (see Supplemental Appendix \ref{sec:est-ind} for a detailed comparison across different variations of the semiparametric estimator). To assess the sensitivity of our empirical findings to these choices, we conduct robustness checks with various alternative specifications. Specifically, we consider (i) \(V_i\) constructed from \(X_{i1}\) or \(\overline X_i=\frac 1 T \sum_t X_{it}\), and (ii) with or without estimated indices ($V'\gamma_0$). Comparing with the benchmark specification in Figure \ref{fig:app-asf-ape}, we see that, in general, the estimates do not change much as we vary the timing of \(V\) or incorporate estimated indices. For robustness checks with alternative coarsening schemes and the local logit estimator, see the previous version of this paper \citep{liu2021identification}.

\begin{figure}[t]
        \caption{Distribution of Observables - Female Labor Force Participation}
        \label{fig:app-dist-obs}
        \begin{center}
        \includegraphics[width=.8\textwidth]{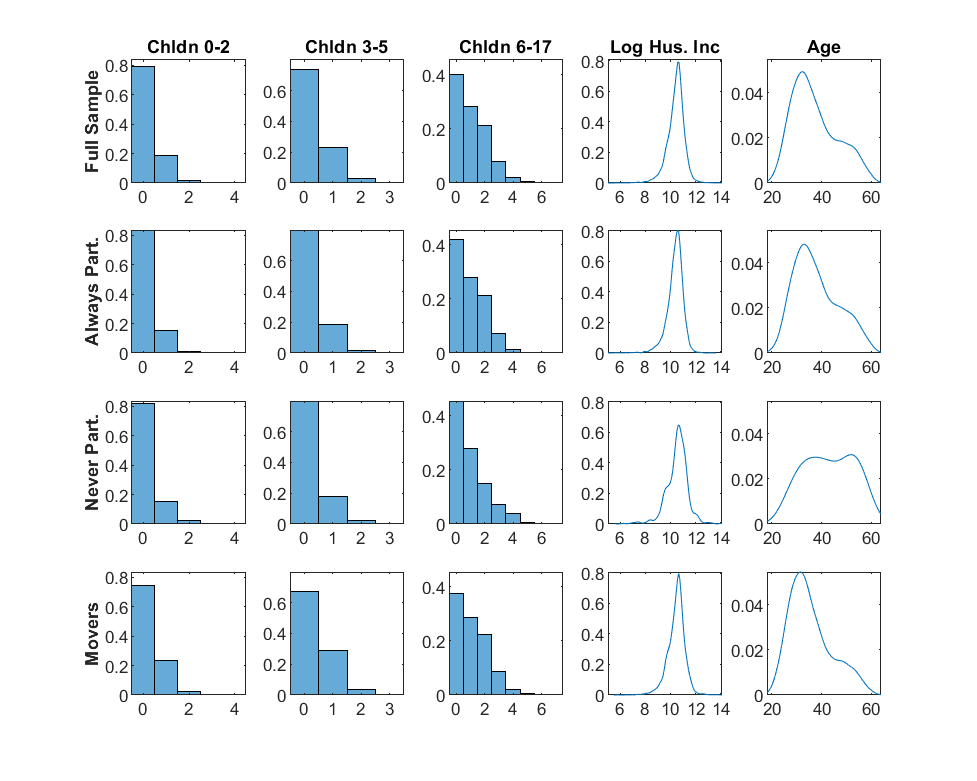}
        \end{center}\vspace{-2em}
        {\footnotesize {\em Notes:} The sample consists of  $N=1461$ married women observed for $T=9$ years from the PSID between 1980--1988. See \cite{Fernandez-Val2009} for details.}\setlength{\baselineskip}{4mm}                                     
\end{figure}

\begin{table}[t]
        \caption{Descriptive Statistics - Female Labor Force Participation}
        \label{tbl:app-descriptive-stats}
        \begin{center}
    \begin{tabular}{lrrrrrrr} \\  \hline \hline
 &       25\%   &      Med.    &      75\%    &  Mean  &     SD        &       Skew.  & Kurt.\\ \hline
\multicolumn{8}{l}{\textit{(a) Full Sample, \#obs = $N\times T$ = 13,149}}\\                                                                                                                             
Participate     &       -       &       -       &       -       &       0.72    &       0.45    &       -       &       -       \\      
Children 0--2        &       0       &       0       &       0       &       0.23    &       0.47    &       1.99    &       6.79    \\      
Children 3--5        &       0       &       0       &       1       &       0.29    &       0.51    &       1.60    &       4.85    \\      
Children 6--17       &       0       &       1       &       2       &       1.05    &       1.10    &       0.91    &       3.46    \\      
Log Husband's Income      &       10.09   &       10.51   &       10.83   &       10.43   &       0.69    &       -0.89   &       7.27    \\      
Age     &       30.00   &       35.00   &       43.00   &       37.30   &       9.22    &       0.56    &       2.50    \\      \\
\multicolumn{8}{l}{\textit{(b) Always Participate, \%obs = 46.27\%}}\\                                                                                                                          
Children 0--2        &       0       &       0       &       0       &       0.18    &       0.41    &       2.25    &       7.56    \\      
Children 3--5        &       0       &       0       &       0       &       0.23    &       0.46    &       1.93    &       6.12    \\      
Children 6--17       &       0       &       1       &       2       &       1.00    &       1.06    &       0.91    &       3.47    \\      
Log Husband's Income      &       10.08   &       10.47   &       10.77   &       10.37   &       0.65    &       -1.36   &       8.89    \\      
Age     &       31.00   &       36.00   &       44.00   &       37.98   &       9.04    &       0.51    &       2.45    \\      \\
\multicolumn{8}{l}{\textit{(c) Never Participate, \%obs = 8.28\%}}\\                                                                                                                            
Children 0--2        &       0       &       0       &       0       &       0.21    &       0.47    &       2.35    &       8.50    \\      
Children 3--5        &       0       &       0       &       0       &       0.23    &       0.48    &       2.05    &       6.79    \\      
Children 6--17       &       0       &       1       &       2       &       0.99    &       1.19    &       1.30    &       4.54    \\      
Log Husband's Income      &       10.13   &       10.62   &       11.04   &       10.53   &       0.85    &       -0.74   &       6.52    \\      
Age     &       35.00   &       43.00   &       52.00   &       42.98   &       10.09   &       -0.06   &       1.90    \\      \\
\multicolumn{8}{l}{\textit{(d) Movers, \%obs = 45.45\%}}\\                                                                                                                               
Participate     &       -       &       -       &       -       &       0.57    &       0.49    &       -       &       -       \\      
Children 0--2        &       0       &       0       &       1       &       0.28    &       0.51    &       1.70    &       5.74    \\      
Children 3--5        &       0       &       0       &       1       &       0.36    &       0.56    &       1.27    &       3.82    \\      
Children 6--17       &       0       &       1       &       2       &       1.11    &       1.11    &       0.83    &       3.18    \\      
Log Husband's Income      &       10.11   &       10.55   &       10.87   &       10.47   &       0.69    &       -0.59   &       5.81    \\      
Age     &       29.00   &       34.00   &       40.00   &       35.57   &       8.71    &       0.73    &       2.88    \\      

\hline
    \end{tabular}       
\end{center}
        {\footnotesize {\em Notes:} The sample consists of  $N=1461$ married women observed for $T=9$ years from the PSID between 1980--1988. ``Movers'' refers to women who have participated in the labor market in some years but not all years. See \cite{Fernandez-Val2009} for details.}\setlength{\baselineskip}{4mm}
\end{table}

\begin{figure}[t]
        \caption{Estimated Coefficients on Time Dummies - Female Labor Force Participation}
        \label{fig:app-dtime}
        \begin{center}
        \includegraphics[width=.8\textwidth]{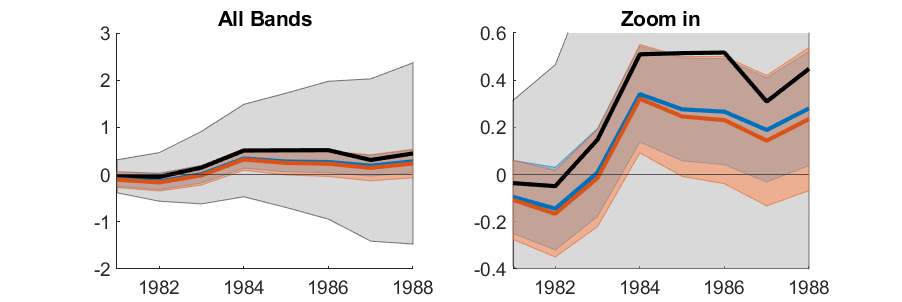}
        \end{center}
        {\footnotesize {\em Notes:} The black/blue/orange solid lines represent point estimates of the coefficients on time dummies using the smoothed maximum score/RE/CRE. The bands with corresponding colors indicate the 90\% symmetric percentile-$t$ confidence intervals based on bootstrap standard deviations. The right panel further zooms in on y-axis values between $-0.4$ and 0.6.}\setlength{\baselineskip}{4mm}                                     
\end{figure}

\begin{figure}[t]
        \caption{Estimated ASF and APE - Female Labor Force Participation, Alternative Specifications}
        \label{fig:app-asf-ape-alt}
        \begin{center}
        \begin{tabular}{cc}
        \rotatebox{90}{\hspace{.7in}\textbf{$V=v(\overline X)$, SP}}
        &\adjustbox{trim={.05\width} {0\height} {.05\width} {0\height},clip}
        {\includegraphics[width=.8\textwidth]{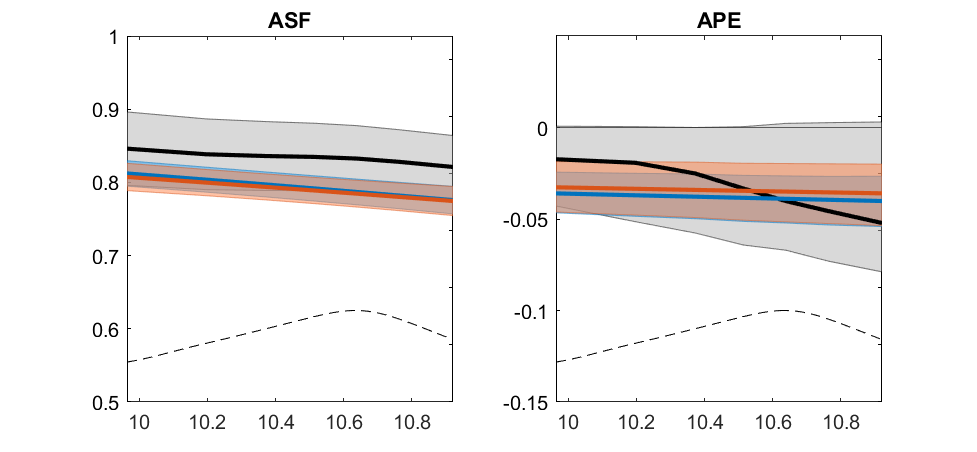}}\\
        \rotatebox{90}{\hspace{.2in} \textbf{$V=v(X_{1})$, SP ($V'\gamma_0$, iter.)}}
        &\adjustbox{trim={.05\width} {0\height} {.05\width} {0.07\height},clip}
        {\includegraphics[width=.8\textwidth]{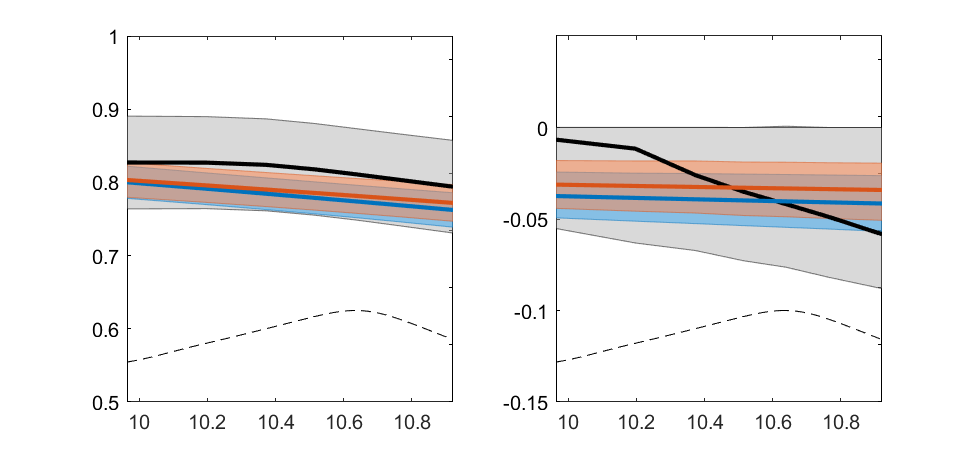}}\\
        \rotatebox{90}{\hspace{.2in} \textbf{$V=v(\overline X)$, SP ($V'\gamma_0$, iter.)}}
        &\adjustbox{trim={.05\width} {0\height} {0.05\width} {0.07\height},clip}
        {\includegraphics[width=.8\textwidth]{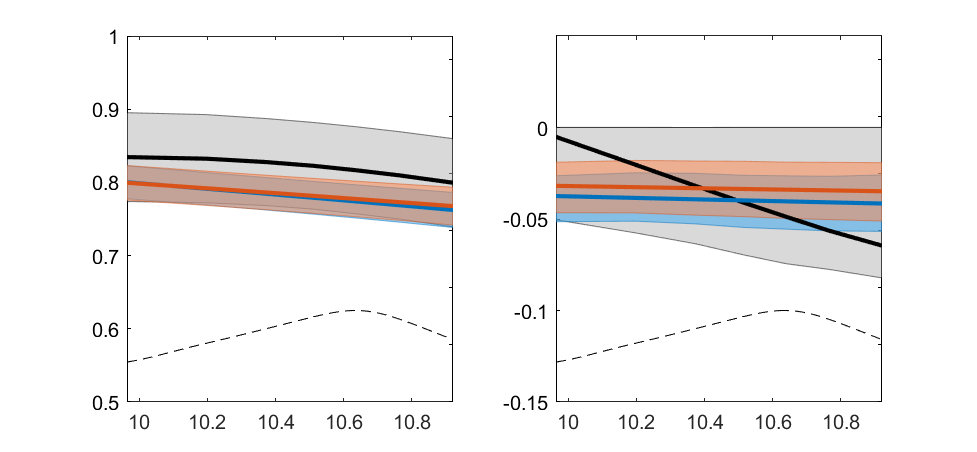}}
        \end{tabular}
        \end{center}
        {\footnotesize {\em Notes:} The x-axes are potential values of log husband's income. The blue/orange solid lines represent point estimates of the ASF and APE using the RE/CRE. The bands with corresponding colors indicate the 90\% bootstrap confidence intervals. The thin dashed lines at the bottom of all panels show the distribution of log husband's income.  
         }\setlength{\baselineskip}{4mm}                                     
\end{figure}

\clearpage
\section{Proofs}\label{supp-app-sec:proofs}
\subsection{Proofs for Appendix \ref{sec:estimation2}}\label{app_sec:est-proofs}

We now present a sequence of lemmas that are used to prove our two main theorems of Appendix \ref{sec:estimation2}: Theorem \ref{thm:ASF} and Theorem \ref{thm:APE}. When applied to matrices, let $\|\cdot\|$ denote the spectral norm.

\begin{lemma}[Convergence of $S_N$] \label{lemma:lem1}
Suppose B\ref{assn:iid}--B\ref{assn:ASFrates} hold. Then,
\begin{align*}
       \sup_{z \in \mathcal{Z}_t} \left\|S_N(z;\widehat{\beta}) - S_N(z;\beta_0)\right\| &= o_p\left(\frac{1}{\sqrt{Nb_N}}\right).
\end{align*}
\end{lemma}

\begin{proof}[Proof of Lemma \ref{lemma:lem1}]
Select the same generic entry from matrices $S_N(z; \widehat\beta)$ and $S_N(z; \beta_0)$. These entries can respectively be written as
\begin{align*}
       S_N^{\tau,\tau'}(z;\widehat{\beta}) &\equiv \avgj \left(\frac{Z_{jt}(\widehat{\beta}) - z}{b_N}\right)^{\tau}  \left(\frac{Z_{jt}(\widehat{\beta}) - z}{b_N}\right)^{\tau'} \mathcal{K}_{b_N} \left(\frac{Z_{jt}(\widehat{\beta}) - z}{b_N}\right) 
\end{align*}
and
\begin{align*}
       S_N^{\tau,\tau'}(z;\beta_0) &\equiv \avgj \left(\frac{Z_{jt}(\beta_0) - z}{b_N}\right)^{\tau}  \left(\frac{Z_{jt}(\beta_0) - z}{b_N}\right)^{\tau'}\mathcal{K}_{b_N}\left(\frac{Z_{jt}(\beta_0) - z}{b_N}\right),
\end{align*}
where $\tau$, $\tau'$ are vectors of exponents which satisfy $0 \leq |\tau|, |\tau'| \leq \ell$. Let $\tau_1$ and $\tau_1'$ denote the first components of $\tau$ and $\tau'$, and let $\tau_{-1}$ and $\tau_{-1}'$ denote vectors with all other components of $\tau$ and $\tau'$. We can write
\begin{align*}
       &S_N^{\tau,\tau'}(z;\widehat{\beta}) - S_N^{\tau,\tau'}(z;\beta_0)\\
   &= \avgj   \left[\left(\frac{X_{jt}'\widehat{\beta}  -u}{b_N}\right)^{\tau_1 + \tau_1'} \frac{1}{b_N} K\left(\frac{X_{jt}'\widehat\beta - u}{b_N}\right) - \left(\frac{X_{jt}'\beta_0  -u}{b_N}\right)^{\tau_1 + \tau_1'} \frac{1}{b_N} K\left(\frac{X_{jt}'\beta_0 - u}{b_N}\right)\right] \\
   &\cdot\left(\frac{V_j - v}{b_N}\right)^{\tau_{-1} + \tau_{-1}'} \mathcal{K}^V_{b_N}\left(\frac{V_j - v}{b_N}\right)\\
   &= \avgj   \left[\frac{1}{b_N}\Gamma\left(\frac{X_{jt}'\widehat{\beta}  -u}{b_N}\right) - \frac{1}{b_N}\Gamma\left(\frac{X_{jt}'\beta_0  -u}{b_N}\right)\right] \left(\frac{V_j - v}{b_N}\right)^{\tau_{-1} + \tau_{-1}'} \mathcal{K}^V_{b_N}\left(\frac{V_j - v}{b_N}\right)
\end{align*}
where $\mathcal{K}^V_{b_N}(v) = b_N^{-d_V} \cdot \prod_{k=1}^{d_V} K(v_k)$, and $\Gamma(u) \equiv u^{\tau_1 + \tau_1'} K(u)$ for generic $u \in \R$. 

By B\ref{assn:kernel}, $\Gamma$ is continuously differentiable. A first-order Taylor expansion yields
\begin{align*}
       S_N^{\tau,\tau'}(z;\widehat{\beta}) - S_N^{\tau,\tau'}(z;\beta_0)
       &= \avgj \frac{1}{b_N^2} \gamma\left(\frac{X_{jt}'\widetilde{\beta} - u}{b_N}\right)  \left(\frac{V_j - v}{b_N}\right)^{\tau_{-1} + \tau_{-1}'}  \mathcal{K}^V_{b_N}\left(\frac{V_j - v}{b_N}\right)X_{jt}'(\widehat{\beta} - \beta_0)
\end{align*}
where $\widetilde\beta$ is such that $X_{jt}'\widetilde\beta$ is between $X_{jt}'\widehat\beta$ and $X_{jt}'\beta_0$, and where $\gamma(u) \equiv \Gamma'(u) = (\tau_1 + \tau_1')u^{\tau_1 + \tau'_1 - 1}K(u) + u^{\tau_1 + \tau_1'} K'(u)$.

Since $\prob(\widehat\beta \in \mathcal{B}_\varepsilon) \to 1$ as $N \to \infty$, with probability arbitrarily close to 1, we have that
\begin{align}
       &\sup_{z \in \mathcal{Z}_t} \left|S_N^{\tau,\tau'}(z;\widehat{\beta}) - S_N^{\tau,\tau'}(z;\beta_0)\right|\notag\\
       &\leq \frac{1}{b_N^2}\sup_{z \in \mathcal{Z}_t}\left\|\avgj   \gamma\left(\frac{X_{jt}'\widetilde{\beta}  -u}{b_N}\right) \left(\frac{V_j - v}{b_N}\right)^{\tau_{-1} + \tau_{-1}'} \mathcal{K}^V_{b_N}\left(\frac{V_j - v}{b_N}\right) X_{jt} \right. \notag\\
       &\left.- \E\left[   \gamma\left(\frac{X_{t}'\widetilde{\beta}  -u}{b_N}\right) \left(\frac{V - v}{b_N}\right)^{\tau_{-1} + \tau_{-1}'} \mathcal{K}^V_{b_N}\left(\frac{V - v}{b_N}\right)X_{t}\right]\right\|\|\widehat{\beta}-\beta_0\| \notag\\
       &+ \sup_{z \in \mathcal{Z}_t} \frac{1}{b_N^2}\left\| \E\left[ \mathcal{K}^V_{b_N}\left(\frac{V - v}{b_N}\right) \left(\frac{V - v}{b_N}\right)^{\tau_{-1} + \tau_{-1}'} \gamma\left(\frac{X_{t}'\widetilde{\beta}  -u}{b_N}\right)X_{t}\right]\right\|\|\widehat{\beta}-\beta_0\| \notag \\
       &\leq \frac{1}{b_N^2}\sup_{z \in \mathcal{Z}_t, \beta \in \mathcal{B}_\varepsilon,b \in (0,\bar{b}]}\left\|\avgj   \gamma\left(\frac{X_{jt}'\beta  -u}{b}\right) \left(\frac{V_j - v}{b}\right)^{\tau_{-1} + \tau_{-1}'} \mathcal{K}^V_{b}\left(\frac{V_j - v}{b}\right) X_{jt} \right. \notag\\
       &\left.- \E\left[  \gamma\left(\frac{X_{t}'\beta  - u}{b}\right)  \left(\frac{V - v}{b}\right)^{\tau_{-1} + \tau_{-1}'} \mathcal{K}^V_{b}\left(\frac{V - v}{b}\right) X_{t}\right]\right\|\|\widehat{\beta}-\beta_0\| \label{eq:lemma1:1}\\
       &+ \sup_{z \in \mathcal{Z}_t,\beta \in \mathcal{B}_\varepsilon} \frac{1}{b_N^2} \left\| \E\left[   \gamma\left(\frac{X_{t}'\beta  -u}{b_N}\right) \left(\frac{V - v}{b_N}\right)^{\tau_{-1} + \tau_{-1}'} \mathcal{K}^V_{b_N}\left(\frac{V - v}{b_N}\right) X_{t}\right]\right\|\|\widehat{\beta}-\beta_0\|, \label{eq:lemma1:2}
\end{align}
where $\bar{b} > 0$.
To obtain the stochastic order of term \eqref{eq:lemma1:1}, define the class of functions
\begin{align*}
       \widetilde{\mathcal{F}} &= \left\{ \gamma\left(\frac{X_t'\beta - u}{b}\right): u \in \R, \beta \in \mathcal{B}_\varepsilon, b \in (0,\bar{b}]\right\}.
\end{align*}
These functions are of the form $\gamma(X_t'c + d)$ where $c = \beta/b$ and $d = - u/b$. Since $K$ has a bounded domain and is twice continuously differentiable with bounded derivatives (Assumption B\ref{assn:kernel}), the function $\gamma(u)$ is of bounded variation on $\R$. By \cite{NolanPollard1987} Lemma 22.(ii), the above class of functions is Euclidean. It is also bounded since $K$ is bounded. Similarly, the classes
\begin{align*}
       \mathcal{F}_{V_k} &= \left\{\left(\frac{V_k - v_k}{b}\right)^{\tau_{k+1} + \tau_{k+1}'} K\left(\frac{V_k - v_k}{b}\right): v_k \in \R, b \in (0,\bar{b}]\right\}
\end{align*}
are Euclidean and bounded for $k = 1, \ldots, d_V$ by the same argument as above. Here $\tau_{k+1}$ and $\tau'_{k+1}$ denote the $(k+1)$th components of $\tau$ and $\tau'$. The product of bounded Euclidean classes is also bounded and Euclidean, hence
\begin{align*}
       \mathcal{F}_{V} &= \left\{ \gamma\left(\frac{X_t'\beta  -u}{b}\right) \left(\frac{V - v}{b}\right)^{\tau_{-1} + \tau_{-1}'} \mathcal{K}^V\left(\frac{V - v}{b}\right): z \in \mathcal{Z}_t, \beta \in \mathcal{B}_\varepsilon, b \in (0,\bar{b}]\right\}
\end{align*}
is bounded and Euclidean. By B\ref{assn:moment}, $\E[\|X_t\|^2]<\infty$. Hence, by Lemma 2.14 (ii) in \cite{PakesPollard1989}, the class
\begin{align*}
       \mathcal{F} &= \left\{\gamma\left(\frac{X_t'\beta - u}{b}\right) \left(\frac{V - v}{b}\right)^{\tau_{-1} + \tau_{-1}'} \mathcal{K}^V\left(\frac{V - v}{b}\right) X_t: z \in \mathcal{Z}_t, \beta \in \mathcal{B}_\varepsilon, b \in (0,\bar{b}]\right\}
\end{align*}
is also Euclidean, and hence Donsker. Therefore, by the continuous mapping theorem,
\begin{align*}
       &\frac{1}{\sqrt{N}b_N^{2+d_V}}\sup_{z \in \mathcal{Z}_t, \beta \in \mathcal{B}_\varepsilon, b \in (0,\bar{b}]}\left\|\frac{1}{\sqrt{N}}\sum_{j=1}^N\left\{ \gamma\left(\frac{X_{jt}'\beta  -u}{b}\right) \left(\frac{V_j - v}{b}\right)^{\tau_{-1} + \tau_{-1}'} \mathcal{K}^V\left(\frac{V_j - v}{b}\right)  X_{jt} \right.\right.\\
       &\left.\left.- \E\left[ \gamma\left(\frac{X_{t}'\beta - u}{b}\right) \left(\frac{V - v}{b}\right)^{\tau_{-1} + \tau_{-1}'} \mathcal{K}^V \left(\frac{V - v}{b}\right) X_{t}\right]\right\}\right\|\\
       &= \frac{1}{\sqrt{Nb_N^{4+2d_V}}} \cdot O_p(1)\\
       &= O_p\left( (Nb_N^{4+2d_V})^{-1/2} \right).
\end{align*}

Thus, term \eqref{eq:lemma1:1} can be written as
\begin{align*}
       &\frac{1}{b_N^2}\sup_{z \in \mathcal{Z}_t, \beta \in \mathcal{B}_\varepsilon,b \in (0,\bar{b}]}\left\|\avgj   \gamma\left(\frac{X_{jt}'\beta  -u}{b}\right) \left(\frac{V_j - v}{b}\right)^{\tau_{-1} + \tau_{-1}'} \mathcal{K}^V_{b}\left(\frac{V_j - v}{b}\right) X_{jt} \right. \notag\\
       &\left.- \E\left[  \gamma\left(\frac{X_{t}'\beta  - u}{b}\right)  \left(\frac{V - v}{b}\right)^{\tau_{-1} + \tau_{-1}'} \mathcal{K}^V_{b}\left(\frac{V - v}{b}\right) X_{t}\right]\right\|\|\widehat{\beta}-\beta_0\|\\
       &= O_p\left( (Nb_N^{4+2d_V})^{-1/2} \right) \cdot O_p(a_N^{-1})\\
       &= o_p\left( (Nb_N)^{-1/2}\right),
\end{align*}
where the last line follows from $a_N^2 b_N^{3+2d_V} \to \infty$ as $N\to\infty$ (Assumption B\ref{assn:ASFrates}).

To bound term \eqref{eq:lemma1:2}, we first note that
\begin{align*}
       &\frac{1}{b_N^2} \left\| \E\left[ \gamma\left(\frac{X_{t}'\beta  -u}{b_N}\right) \left(\frac{V - v}{b_N}\right)^{\tau_{-1} + \tau_{-1}'} \mathcal{K}^V_{b_N}\left(\frac{V - v}{b_N}\right) X_{t}\right]\right\|\\
       &=  \left\| \E\left[\frac{\partial}{\partial\beta} \left(\frac{Z_t(\beta) - z}{b_N}\right)^{\tau + \tau'} \mathcal{K}_{b_N}\left(\frac{Z_t(\beta) - z}{b_N}\right)\right]\right\|\\
       &=  \left\| \int \frac{\partial}{\partial\beta} \left(\frac{\widetilde{z} - z}{b_N}\right)^{\tau + \tau'} \mathcal{K}_{b_N}\left(\frac{\widetilde{z} - z}{b_N}\right)f_{Z_t(\beta)}(\widetilde{z}) \, d\widetilde{z} \right\|\\
       &= \left\|  \int  a^{\tau + \tau'} \mathcal{K}\left(a\right) \frac{\partial}{\partial\beta} f_{Z_t(\beta)}(z + ab_N) \, da\right\|.
\end{align*}
The last equality follows from the change of variables $\widetilde{z} = z + ab_N$. We then have that
\begin{align*}
       \sup_{z \in \mathcal{Z}_t,\beta \in \mathcal{B}_\varepsilon} \left\|  \int  a^{\tau + \tau'} \mathcal{K}\left(a\right)\frac{\partial}{\partial\beta} f_{Z_t(\beta)}(z + ab_N)\, da\right\|
       &\leq \sup_{z \in \mathcal{Z}_t,\beta \in \mathcal{B}_\varepsilon} \left\|\frac{\partial}{\partial\beta} f_{Z_t(\beta)}(z)\right\| \left|  \int  a^{\tau + \tau'} \mathcal{K}\left(a\right) \, da\right|\\
       &<\infty.
\end{align*}
To see that the last inequality holds, recall Assumption B\ref{assn:smoothness}.(ii), and that $\mathcal{K}$ is a bounded function with compact support, hence $a^{\tau + \tau'}\mathcal{K}(a)$ is bounded with compact support. Therefore, term \eqref{eq:lemma1:2} is of order $O(1) \cdot \|\widehat{\beta} - \beta_0\| = O_p(a_N^{-1}) = o_p\left( (Nb_N)^{-1/2} \right)$ since, by B\ref{assn:ASFrates}, $N b_N a_N^{-2} \to 0$ as $N \to \infty$.

Combining the rates of convergence of terms \eqref{eq:lemma1:1} and \eqref{eq:lemma1:2}, we obtain
\begin{align*}
       \sup_{z \in \mathcal{Z}_t} \left|S_N^{\tau,\tau'}(z;\widehat{\beta}) - S_N^{\tau,\tau'}(z;\beta_0)\right| &= o_p\left(\frac{1}{\sqrt{Nb_N}}\right)
\end{align*}
Since this rate of convergence applies uniformly in $z \in \mathcal{Z}_t$ to a generic element of $S_N^{\tau,\tau'}(z;\widehat{\beta}) - S_N^{\tau,\tau'}(z;\beta_0)$, it also applies uniformly in $z \in \mathcal{Z}_t$ to the matrix norm of $S_N(z;\widehat{\beta}) - S_N(z;\beta_0)$, which concludes the proof.
\end{proof}

Define 
       \[S(z;\beta_0) = \int \xi(a)\xi(a)'\mathcal{K}(a)\, da \cdot f_{Z_t(\beta_0)}(z).\]

\begin{lemma}[Convergence of $S_N$ to $S$] \label{lemma:lem1b}
Suppose B\ref{assn:iid}--B\ref{assn:ASFrates} hold. Then,
\begin{align*}
       \sup_{z \in \mathcal{Z}_t} \left\|S_N(z;\beta_0) - S(z;\beta_0)\right\| &=  O_p\left(\left(\frac{\log(N)}{Nb_N^{1+d_V}}\right)^{1/2}\right) + O(b_N).
\end{align*}
\end{lemma}

\begin{proof}[Proof of Lemma \ref{lemma:lem1b}]
This is Corollary 1.(ii) in \cite{Masry1996} with $\theta = 1$ (in his notation), therefore we verify its assumptions. His condition 1(b) holds by B\ref{assn:smoothness}.(iv). His conditions 2 and 3 hold by B\ref{assn:kernel} and B\ref{assn:smoothness}.(iii). Finally, the rate conditions of Theorem 2 in \cite{Masry1996} hold by B\ref{assn:ASFrates}. Therefore, all assumptions of his corollary hold and the above result holds.
\end{proof}

\begin{lemma}[Convergence of $T_N$] \label{lemma:lem2}
Suppose B\ref{assn:iid}--B\ref{assn:ASFrates} hold. Then,
\begin{align*}
       \sup_{z \in \mathcal{Z}_t} \left\|T_N(z;\widehat{\beta}) - T_N(z;\beta_0)\right\| &= o_p\left(\frac{1}{\sqrt{Nb_N}}\right).
\end{align*}
\end{lemma}

\begin{proof}[Proof of Lemma \ref{lemma:lem2}]
Select the same generic component from $T_N(z;\widehat{\beta})$ and $T_N(z;\beta_0)$. These components can respectively be written as
\begin{align*}
       T_N^{\tau}(z;\widehat{\beta}) &\equiv \avgj \left(\frac{Z_{jt}(\widehat{\beta}) - z}{b_N}\right)^{\tau}  Y_{jt} \mathcal{K}_{b_N}\left(\frac{Z_{jt}(\widehat{\beta}) - z}{b_N}\right) \\
       T_N^{\tau}(z;\beta_0) &\equiv \avgj \left(\frac{Z_{jt}(\beta_0) - z}{b_N}\right)^{\tau} Y_{jt} \mathcal{K}_{b_N}\left(\frac{Z_{jt}(\beta_0) - z}{b_N}\right),
\end{align*}
where $\tau$ is a vector of exponents which satisfies $0 \leq |\tau| \leq \ell$. Again let $\tau_1$ denote the first component of $\tau$ and let $\tau_{-1}$  denote all other components of $\tau$. Let $\Gamma(u) \equiv u^{\tau_1} K(u)$ and $\gamma(u) \equiv \Gamma'(u) = \tau_1 u^{\tau_1 - 1} K(u) + u^{\tau_1}K'(u)$. As in the proof of Lemma \ref{lemma:lem1}, we write
\begin{align*}
       &T_N^{\tau}(z;\widehat{\beta}) - T_N^{\tau}(z;\beta_0)\\
       &= \avgj Y_{jt} \left[\frac{1}{b_N}\Gamma\left(\frac{X_{jt}'\widehat{\beta}  -u}{b_N}\right) - \frac{1}{b_N}\Gamma\left(\frac{X_{jt}'\beta_0  -u}{b_N}\right)\right] \left(\frac{V_j - v}{b_N}\right)^{\tau_{-1}} \mathcal{K}^V_{b_N}\left(\frac{V_j - v}{b_N}\right)\\
       &= \avgj Y_{jt} \frac{1}{b_N^2}\gamma\left(\frac{X_{jt}'\widetilde\beta - u}{b_N}\right) \left(\frac{V_j - v}{b_N}\right)^{\tau_{-1}} \mathcal{K}^V_{b_N}\left(\frac{V_j - v}{b_N}\right) X_{jt}'(\widehat{\beta}-\beta_0)
\end{align*}

By the same arguments as in the proof of Lemma \ref{lemma:lem1}, and by $\E[Y_{jt}^2] < \infty$, we can show that
\begin{align*}
       &\sup_{z\in\mathcal{Z}_t} \left|T_N^{\tau}(z;\widehat{\beta}) - T_N^{\tau}(z;\beta_0)\right|\\
       &\leq \frac{1}{b_N^2}\sup_{z \in \mathcal{Z}_t, \beta \in \mathcal{B}_\varepsilon, b \in (0,\bar{b}]}\left\|\avgj  Y_{jt} X_{jt} \gamma\left(\frac{X_{jt}'\beta  -u}{b}\right) \left(\frac{V_j - v}{b}\right)^{\tau_{-1} } \mathcal{K}^V_{b}\left(\frac{V_j - v}{b}\right) \right.\\
       &\left.- \E\left[   Y_{t}X_{t} \gamma\left(\frac{X_{t}'\beta  - u}{b}\right) \left(\frac{V - v}{b}\right)^{\tau_{-1}} \mathcal{K}^V_{b}\left(\frac{V - v}{b}\right)\right]\right\|\|\widehat{\beta}-\beta_0\|\\
       &+ \sup_{z \in \mathcal{Z}_t,\beta \in \mathcal{B}_\varepsilon} \frac{1}{b_N^2} \left\| \E\left[   Y_{t}X_{t} \gamma\left(\frac{X_{t}'\beta  - u}{b_N}\right) \left(\frac{V - v}{b_N}\right)^{\tau_{-1}} \mathcal{K}^V_{b_N}\left(\frac{V - v}{b_N}\right)\right]\right\|\|\widehat{\beta}-\beta_0\|\\
       &= O_p\left(\frac{1}{\sqrt{Nb_N^{4+2d_V}}}\right) \cdot O_p(a_N^{-1}) + O(1)\cdot O_p(a_N^{-1})\\
       &= o_p\left(\frac{1}{\sqrt{Nb_N}}\right)
\end{align*}
holds with probability arbitrarily close to 1 as $N\to\infty$ since $\prob(\widehat{\beta} \in \mathcal{B}_\varepsilon) \to 1$. The last equality follows from B\ref{assn:ASFrates}.

Since this rate of convergence applies uniformly in $z \in \mathcal{Z}_t$ to generic components of the vector $T_N(z; \widehat\beta) - T_N(z; \beta_0)$, it applies to its vector norm uniformly in $z \in \mathcal{Z}_t$ as well, which concludes the proof.
\end{proof}

Let 
       \[ T(z;\beta_0) = \int \xi(a) \mathcal{K}(a)\, da \cdot  \E[Y_t| Z_t(\beta_0) = z] f_{Z_t(\beta_0)}(z). \]
Also, recall that $Z_t \equiv Z_t(\beta_0)$.

\begin{lemma}[Convergence of $T_N$ to $T$] \label{lemma:lem2b}
Suppose B\ref{assn:iid}--B\ref{assn:ASFrates} hold. Then,
\begin{align*}
       \sup_{z \in \mathcal{Z}_t} \left\|T_N(z;\beta_0) - T(z;\beta_0)\right\| &= O_p\left(\left(\frac{\log(N)}{Nb_N^{1+d_V}}\right)^{1/2}\right) + O(b_N).
\end{align*}
\end{lemma}

\begin{proof}[Proof of Lemma \ref{lemma:lem2b}]
By the triangle inequality,
\begin{align*}
       \sup_{z \in \mathcal{Z}_t} \left\|T_N(z;\beta_0) - T(z;\beta_0)\right\| &\leq \sup_{z \in \mathcal{Z}_t} \left\|T_N(z;\beta_0) - \E[T_N(z;\beta_0)]\right\| + \sup_{z \in \mathcal{Z}_t} \left\|\E[T_N(z;\beta_0)] - T(z;\beta_0)\right\|.
\end{align*}

Generic components of $T_N(z;\beta_0) - \E[T_N(z;\beta_0)]$ can be written as
\begin{align*}
       &\sup_{z \in \mathcal{Z}_t} \left|\avgj \left(\frac{Z_{jt} - z}{b_N}\right)^\tau Y_{jt}\mathcal{K}_{b_N}\left(\frac{Z_{jt} - z}{b_N}\right) - \E\left[\left(\frac{Z_t - z}{b_N}\right)^\tau Y_t \mathcal{K}_{b_N}\left(\frac{Z_t - z}{b_N}\right)\right]\right|.
\end{align*}
By an argument similar to that used in Corollary 1.(ii) in \cite{Masry1996} or in Lemma B.ii.(2) in \cite{RotheFirpo2019}, this term is of order $O_p\left(\left(\frac{\log(N)}{Nb_N^{1+d_V}}\right)^{1/2}\right).$

Next, note that generic elements of $\E[T_N(z;\beta_0)]$ are of the form
\begin{align*}
       &\E\left[\left(\frac{Z_{t} - z}{b_N}\right)^\tau Y_t \mathcal{K}_{b_N}\left(\frac{Z_{t} - z}{b_N}\right)\right]\\
       &= \int \left(\frac{\tilde{z} - z}{b_N}\right)^\tau \E[Y_t|Z_t = \tilde{z}] \mathcal{K}_{b_N}\left(\frac{\tilde{z} - z}{b_N}\right) f_{Z_t}(\tilde{z})\, d\tilde{z}\\
       &= \int a^\tau \mathcal{K}(a) \E[Y_t|Z_t = z + ab_N] f_{Z_t}(z + a b_N) \, da\\
       &\leq \E[Y_t|Z_t = z] f_{Z_t}(z)\int a^\tau \mathcal{K}(a) \, da + b_N \sup_{z\in\mathcal{Z}_t} \left\| \frac{\partial}{\partial z} \left(\E[Y_t|Z_t = z] f_{Z_t}(z)\right) \right\| \cdot \left\|\int a^\tau \mathcal{K}(a) \cdot a \, da\right\|.
\end{align*}
The second equality follows from a change in variables. Note that $\E[Y_t|Z_t = z] f_{Z_t}(z) \int a^\tau \mathcal{K}(a) \, da $ is the corresponding element of $T(z;\beta_0)$. Therefore,
\begin{align*}
       &\sup_{z\in\mathcal{Z}_t} \left|\int a^\tau \mathcal{K}(a) \E[Y_t|Z_t = z + ab_N] f_{Z_t}(z + ab_N) \, da -  \E[Y_t|Z_t = z] f_{Z_t}(z)\int a^\tau \mathcal{K}(a) \, da\right|\\
       &\leq b_N \sup_{z\in\mathcal{Z}_t} \left\| \frac{\partial}{\partial z} \left(\E[Y_t|Z_t = z] f_{Z_t}(z)\right) \right\| \cdot \left\|\int a^\tau \mathcal{K}(a) \cdot a \, da\right\|.
\end{align*}

By B\ref{assn:kernel}, $\left\|\int a^\tau \mathcal{K}(a) \cdot a \, da\right\|<\infty$. By B\ref{assn:smoothness}.(iii), we have that $\sup_{z\in\mathcal{Z}_t} \left\| \frac{\partial}{\partial z} \left( \E[Y_t|Z_t = z] f_{Z_t}(z) \right) \right\|<\infty$. Therefore, 
\[
       \sup_{z \in \mathcal{Z}_t} \left\|\E[T_N(z;\beta_0)] - T(z;\beta_0)\right\| = O(b_N)
\]
and
\begin{align*}
       \sup_{z \in \mathcal{Z}_t} \left\|T_N(z;\beta_0) - T(z;\beta_0)\right\| &= O_p\left(\left(\frac{\log(N)}{Nb_N^{1+d_V}}\right)^{1/2}\right) + O(b_N).
\end{align*}
\end{proof}

\begin{lemma}[Convergence of $S_N$ part 2] \label{lemma:lem3}
Suppose B\ref{assn:iid}--B\ref{assn:ASFrates} hold. Then,
\begin{align*}
       \sup_{z \in \mathcal{Z}_t} \left\|\frac{\partial}{\partial u} S_N(z;\beta_0)\right\| &= o_p\left(\frac{a_N}{\sqrt{Nb_N}}\right).
\end{align*}
\end{lemma}

\begin{proof}[Proof of Lemma \ref{lemma:lem3}]
As in the proof of Lemma \ref{lemma:lem1}, consider a generic entry of $S_N(z;\beta_0)$, which we write as
\begin{align*}
       S_N^{\tau,\tau'}(z;\beta_0) &= \avgj \left(\frac{Z_{jt} - z}{b_N}\right)^{\tau+\tau'}  \mathcal{K}_{b_N}\left(\frac{Z_{jt} - z}{b_N}\right).
\end{align*}
Its derivative with respect to $u$, the first element of $z$, is
\begin{align*}
       \frac{\partial}{\partial u} S_N^{\tau,\tau'}(z;\beta_0) &= \frac{-1}{b_N^{2+d_V}} \avgj \gamma\left(\frac{X_{jt}'\beta_0 - u}{b_N}\right) \left(\frac{V_j - v}{b_N}\right)^{\tau_{-1} + \tau_{-1}'} \mathcal{K}^V\left(\frac{V_j - v}{b_N}\right)
\end{align*}
where $\gamma(u) = (\tau_1 + \tau_1')u^{\tau_1 + \tau'_1 - 1}K(u) + u^{\tau_1 + \tau_1'} K'(u)$. 

Therefore, we have that
\begin{align}
       &\sup_{z\in\mathcal{Z}_t} \left|\frac{\partial}{\partial u} S_N^{\tau,\tau'}(z;\beta_0)\right| = \sup_{z\in\mathcal{Z}_t}\left|\frac{-1}{b_N^{2+d_V}} \avgj \gamma\left(\frac{X_{jt}'\beta_0 - u}{b_N}\right) \left(\frac{V_j - v}{b_N}\right)^{\tau_{-1} + \tau_{-1}'} \mathcal{K}^V\left(\frac{V_j - v}{b_N}\right) \right| \notag \\
       &\leq \sup_{z\in\mathcal{Z}_t, b\in (0,\bar{b}]} \frac{1}{\sqrt{N} b_N^{2+d_V}}\left|\frac{1}{\sqrt{N}}\sum_{j=1}^N \left\{ \gamma\left(\frac{X_{jt}'\beta_0  -u}{b}\right) \left(\frac{V_j - v}{b}\right)^{\tau_{-1} + \tau_{-1}'} \mathcal{K}^V\left(\frac{V_j - v}{b}\right) \right.\right. \notag \\
       &\left.\left. - \E\left[ \gamma\left(\frac{X_{t}'\beta_0  -u}{b}\right) \left(\frac{V - v}{b}\right)^{\tau_{-1} + \tau_{-1}'} \mathcal{K}^V\left(\frac{V - v}{b}\right) \right]\right\}\right| \label{eq:lemma3:1}\\
       &+ \sup_{z\in\mathcal{Z}_t} \frac{1}{b_N^2} \left|\E\left[ \gamma\left(\frac{X_{t}'\beta_0  -u}{b}\right) \left(\frac{V - v}{b}\right)^{\tau_{-1} + \tau_{-1}'} \mathcal{K}^V_{b_N}\left(\frac{V - v}{b}\right) \right]\right|. \label{eq:lemma3:2}
\end{align}

The class
\begin{align*}
       \left\{ \gamma\left(\frac{X_t'\beta_0  -u}{b}\right) \left(\frac{V - v}{b}\right)^{\tau_{-1} + \tau_{-1}'} \mathcal{K}^V\left(\frac{V - v}{b}\right): z \in \mathcal{Z}_t, b \in (0,\bar{b}]\right\}
\end{align*}
is a subset of $\mathcal{F}_V$ which is Euclidean, therefore it is also Euclidean and hence Donsker. We therefore have that term \eqref{eq:lemma3:1} is of order $O_p\left(\frac{1}{\sqrt{Nb_N^{4+2d_V}}}\right)$. 

We can bound term \eqref{eq:lemma3:2} as follows,
\begin{align*}
       &\sup_{z \in \mathcal{Z}_t} \frac{1}{b_N^2} \left| \E\left[ \gamma\left(\frac{X_{t}'\beta_0  -u}{b_N}\right) \left(\frac{V - v}{b_N}\right)^{\tau_{-1} + \tau_{-1}'} \mathcal{K}^V_{b_N}\left(\frac{V - v}{b_N}\right)\right]\right|\\
       &= \sup_{z \in \mathcal{Z}_t} \left| \E\left[\frac{\partial}{\partial u} \left(\frac{Z_t - z}{b_N}\right)^{\tau + \tau'} \mathcal{K}_{b_N}\left(\frac{Z_t - z}{b_N}\right)\right]\right|\\
               &= \sup_{z \in \mathcal{Z}_t} \left| \int \frac{\partial}{\partial u}  \left(\frac{\tilde{z} - z}{b_N}\right)^{\tau + \tau'} \mathcal{K}_{b_N}\left(\frac{\tilde{z} - z}{b_N}\right) f_{Z_t(\beta_0)}(\tilde{z}) \, d \tilde{z} \right|\\
               &= \sup_{z \in \mathcal{Z}_t} \left|  \int \frac{\partial}{\partial u} a^{\tau + \tau'} \mathcal{K}\left(a\right)f_{Z_t}(z + a b_N) \, da\right|\\
               &\leq \sup_{z \in \mathcal{Z}_t} \left|\frac{\partial}{\partial u} f_{Z_t}(z)\right| \left|\int a^{\tau + \tau'} \mathcal{K}\left(a\right) \, da\right|\\
               &= O(1).
\end{align*}
The third equality follows from the change of variables $\tilde{z} = z + a b_N$. The final line follows from B\ref{assn:kernel} and B\ref{assn:smoothness}.(iii).

Therefore, 
\begin{align*}
       \sup_{z\in\mathcal{Z}_t} \left|\frac{\partial}{\partial z_1} S_N^{\tau,\tau'}(z;\beta_0)\right| &= O_p\left(\frac{1}{\sqrt{Nb_N^{4+2d_V}}}\right) + O(1)\\
       &= o_p\left(\frac{a_N}{\sqrt{Nb_N}}\right)
\end{align*}
since, as $N\to\infty$, $\frac{1}{\sqrt{Nb_N^{4+2d_V}}} \cdot \frac{\sqrt{Nb_N}}{a_N} = O(N^{\epsilon - \delta(3/2+d_V)}) = o(1)$ by B\ref{assn:ASFrates}, and since $ \frac{\sqrt{Nb_N}}{a_N} \cdot O(1) = O(N^{1/2 - \epsilon - \delta/2}) = o(1)$, also by B\ref{assn:ASFrates}. Since this holds for a generic entry of the matrix $\frac{\partial}{\partial u} S_N(z;\beta_0)$, it holds for its matrix norm as well, which concludes this lemma.
\end{proof}

\begin{lemma}[Convergence of $T_N$ part 2] \label{lemma:lem4}
Suppose B\ref{assn:iid}--B\ref{assn:ASFrates} hold. Then,
\begin{align*}
       \sup_{z \in \mathcal{Z}_t} \left\|\frac{\partial}{\partial u} T_N(z;\beta_0)\right\| &= o_p\left(\frac{a_N}{\sqrt{Nb_N}}\right).
\end{align*}
\end{lemma}

\begin{proof}[Proof of Lemma \ref{lemma:lem4}]
As in the proof of Lemma \ref{lemma:lem2}, consider a generic component of the vector $T_N(z;\beta_0)$. Denote this element by
\begin{align*}
       T_N^{\tau}(z;\beta_0) &= \avgj \left(\frac{Z_{jt} - z}{b_N}\right)^{\tau} Y_{jt} \mathcal{K}_{b_N}\left(\frac{Z_{jt} - z}{b_N}\right).
\end{align*}
Its derivative with respect to $u$ is
\begin{align*}
       \frac{\partial}{\partial u} T_N^{\tau}(z;\beta_0) &= \frac{-1}{b_N^{2+d_V}} \avgj  Y_{jt}\gamma\left(\frac{X_{jt}'\beta  -u}{b}\right) \left(\frac{V_j - v}{b}\right)^{\tau_{-1}} \mathcal{K}^V\left(\frac{V_j - v}{b}\right).
\end{align*}
where $\gamma(u) = \tau_1 u^{\tau_1 - 1} K(u) + u^{\tau_1} K'(u)$. The rest of the proof follows directly from the arguments used in the proofs of Lemmas \ref{lemma:lem2} and \ref{lemma:lem3}.
\end{proof}

\begin{lemma}[Convergence of indicators] \label{lemma:lem5}
Suppose B\ref{assn:iid}--B\ref{assn:ASFrates} hold. Suppose $\widetilde\beta \pconv \beta_0$. Let $\pi_{it}(\beta) \equiv \1((\ubx_t'\beta,V_i) \in \mathcal{Z}_t)$. Then,
       \[\prob\left(\sup_{i=1,\ldots,N} \left|\pi_{it}(\widetilde\beta) - \pi_{it}(\beta_0)\right| = 0\right) \to 1\]
       as $N\to \infty$.
\end{lemma}

\begin{proof}[Proof of Lemma \ref{lemma:lem5}]
We note that
\begin{align*}
       \sup_{i=1,\ldots,N} |\pi_{it}(\widetilde{\beta}) - \pi_{it}(\beta_0)| &= \sup_{i=1,\ldots,N}\left(\1((\ubx_t'\widetilde{\beta},V_i) \in \mathcal{Z}_t, (\ubx_t'\beta_0, V_i) \notin \mathcal{Z}_t) + \1((\ubx_t'\widetilde{\beta},V_i) \notin \mathcal{Z}_t, (\ubx_t'\beta_0, V_i) \in \mathcal{Z}_t)\right)\\
       &\leq \sup_{i=1,\ldots,N}\left(\1(\ubx_t'\widetilde{\beta} \in \mathcal{Z}_{1t}, \ubx_t'\beta_0 \notin \mathcal{Z}_{1t}) + \1(\ubx_t'\widetilde{\beta} \notin \mathcal{Z}_{1t}, \ubx_t'\beta_0 \in \mathcal{Z}_{1t})\right)\\
       &= \1(\ubx_t'\widetilde{\beta} \in \mathcal{Z}_{1t}, \ubx_t'\beta_0 \notin \mathcal{Z}_{1t}) + \1(\ubx_t'\widetilde{\beta} \notin \mathcal{Z}_{1t}, \ubx_t'\beta_0 \in \mathcal{Z}_{1t}),
\end{align*}
where $\mathcal{Z}_{1t} = \{z_1 = e_1'z: z\in\mathcal{Z}_t\}$. By B\ref{assn:smoothness}.(v), $\ubx_t'\beta_0 \in \mathcal{Z}_{1t}$, and therefore $\1(\ubx_t'\widetilde{\beta} \in \mathcal{Z}_{1t}, \ubx_t'\beta_0 \notin \mathcal{Z}_{1t}) = 0$, and $\1(\ubx_t'\widetilde{\beta} \notin \mathcal{Z}_{1t}, \ubx_t'\beta_0 \in \mathcal{Z}_{1t})  = \1(\ubx_t'\widetilde{\beta} \notin \mathcal{Z}_{1t})$. 

By assumption, $\widetilde{\beta}$ converges in probability to $\beta_0$. By Theorem 18.9.(v) in \cite{Vaart1998}, $\prob(\ubx_t'\widetilde{\beta} \in \mathcal{Z}_{1t}) \to \1(\ubx_t'\beta_0 \in \mathcal{Z}_{1t}) = 1$ since $\ubx_t'\beta_0$ is not in the boundary of $\mathcal{Z}_{1t}$ by B\ref{assn:smoothness}.(v).

Therefore, 
\begin{align*}
       \prob\left(\sup_{i=1,\ldots,N} |\pi_{it}(\widetilde{\beta}) - \pi_{it}(\beta_0)| = 0 \right) &\geq \prob(\1(\ubx_t'\widetilde{\beta} \notin \mathcal{Z}_{1t}) = 0) = \prob(\ubx_t'\widetilde{\beta} \in \mathcal{Z}_{1t}) \to \prob(\ubx_t'\beta_0 \in \mathcal{Z}_{1t}) = 1
\end{align*}
as $N \to \infty$.
\end{proof}

\begin{lemma}[ASF convergence in distribution]\label{lemma:convASF}
Suppose B\ref{assn:iid}--B\ref{assn:ASFrates} hold. Then,
\begin{align*}
       \sqrt{Nb_N}\left(\avg \widehat{h}_1(\ubx_t'\beta_0,V_i;\beta_0)\pi_{it} - \E[h_1(\ubx_t'\beta_0,V;\beta_0)\pi_t]\right) &\dconv \mathcal{N}(0,\sigma^2_{\text{ASF}_t}(\ubx_t'\beta_0)).
\end{align*}

\end{lemma}

\begin{proof}[Proof of Lemma \ref{lemma:convASF}]
       This proof builds on the proof of Corollary 2 in \cite{KongLintonXia2010} (KLX hereafter). First, we verify that Assumptions A1--A7 of KLX hold under ours. Their A1 holds with our squared loss function, and we note that $\psi(\varepsilon_i) \equiv -2(Y_{it} - \E[Y_{t}|Z_{it}] )$ in their notation. By our A\ref{assn:moment}, $\E[|\psi(\varepsilon_i)|^{\nu_1}]<\infty$ holds for arbitrary large $\nu_1$. Their A2 holds immediately. Their A3 holds by our B\ref{assn:kernel}. Their A4 and A5 hold by  our B\ref{assn:smoothness}.(iii). Their A6 holds if 
\begin{align*}
       Nb_N^{1+d_V}/\log(N) &\to \infty\\
       Nb_N^{1+d_V + 2(\ell + 1)}/\log(N) &= O(1)\\
               N^{\nu_2/8 - \lambda_1 - 1/4} b_N^{(1+d_V)(\nu_2/8 -\lambda_1 + 3/4)} \log(N)^{-\nu_2/8 + 5/4 + \lambda_1}
       &\to \infty,
\end{align*}

for some $2 < \nu_2 \leq \nu_1$. Since $b_N = \kappa \cdot N^{-\delta}$, these conditions are equivalent to
\begin{align*}
       1 - \delta(1+d_V) &> 0\\
       1 - \delta(3 + 2\ell + d_V) &\leq 0\\
       \nu_2/8 - \lambda_1 - 1/4 - \delta(1+d_V)(\nu_2/8 -\lambda_1 + 3/4) &> 0.
\end{align*}
Since $\nu_1$ can be made arbitrarily large, $\nu_2$ can also taken to be arbitrarily large, and the last inequality is equivalent to
\begin{align*}
       \delta &< \frac{1}{1+d_V}.
\end{align*}    
By our B\ref{assn:ASFrates}, these rate conditions all hold. Finally, their A7 holds by our B\ref{assn:smoothness}.(v). Since these assumptions hold for $\lambda_1 = 1$, we can use equation (13) in KLX and their Corollary 1 to write
\begin{align*}
       \widehat{h}_1(z;\beta_0) &= h_1(z;\beta_0) + B_{1,N}(z) + \avgj \phi_{1,jN}(z) + R_{1,N}(z)
\end{align*}
where $B_{1,N}(z)$ is a bias term satisfying $\sup_{z\in\mathcal{Z}_t} |B_{1,N}(z)| = O(b_N^{\ell+1})$ if $\ell$ is odd or $O(b_N^{\ell+2})$ if $\ell$ is even, where $\phi_{1,jN}(z)$ are mean-zero random variables, and where $R_{1,N}(z)$ is a higher-order term satisfying $\sup_{z\in\mathcal{Z}_t} |R_{1,N}(z)| = O_p\left(\frac{\log(N)}{Nb_N^{1+d_V}}\right)$. 
       
Second, we note that
\begin{align}
       &\sqrt{Nb_N} \left(\avg \widehat{h}_1(\ubx_t'\beta_0,V_i;\beta_0)\pi_{it} - \E[h_1(\ubx_t'\beta_0, V; \beta_0) \pi_t]\right)\notag\\
       &= \sqrt{Nb_N}\avg \left(\widehat{h}_1(\ubx_t'\beta_0, V_i; \beta_0) - h_1(\ubx_t'\beta_0, V_i; \beta_0) \right) \pi_{it} \label{eq:lemma8:1}\\
       &+ \sqrt{b_N} \cdot \frac{1}{\sqrt{N}} \sum_{i=1}^N \left(h_1(\ubx_t'\beta_0,V_i;\beta_0)\pi_{it} - \E[h_1(\ubx_t'\beta_0,V;\beta_0)\pi_t]\right). \label{eq:lemma8:2}
\end{align}

To analyze term \eqref{eq:lemma8:1}, we use the fact that
\begin{align*}
       &\sqrt{Nb_N}\avg \left(\widehat{h}_1(\ubx_t'\beta_0, V_i; \beta_0) - h_1(\ubx_t'\beta_0, V_i; \beta_0) \right) \pi_{it} \\
       &= \sqrt{Nb_N} \avg B_{1,N}(\ubx_t'\beta_0, V_i)\pi_{it}\\
       & + \sqrt{Nb_N} \avgij \phi_{1,jN}(\ubx_t'\beta_0, V_i)\pi_{it} + \sqrt{Nb_N} \avg R_{1,N}(\ubx_t'\beta_0, V_i)\pi_{it}.
\end{align*}

When $\ell$ is odd, $\sqrt{Nb_N} \avg B_{1,N}(\ubx_t'\beta_0, V_i)\pi_{it}$ is $o(1)$ because
\begin{align*}
       \left|\sqrt{Nb_N} \avg B_{1,N}(\ubx_t'\beta_0,V_i)\pi_{it}\right| &\leq \sqrt{Nb_N} \cdot\sup_{z\in\mathcal{Z}_t} \left|B_{1,N}(z)\right|\\
       &= \sqrt{Nb_N}\cdot O(b_N^{\ell+1})\\
       &= O(\sqrt{Nb_N^{2\ell+3}})
\end{align*}
and by B\ref{assn:ASFrates}. A similar derivation applies when $\ell$ is even.

We now show that term $\sqrt{Nb_N} \avgij \phi_{1,jN}(\ubx_t'\beta_0,V_i)\pi_{it}$ converges in distribution to a normal distribution. By standard arguments from \cite{Masry1996}, which are also referred to in the proof of Corollary 2 in KLX, we have that
\begin{align*}
       &\avgij \phi_{1,jN}(\ubx_t'\beta_0,V_i)\pi_{it}\\
       &= \frac{-1}{Nb_N} \sum_{i=1}^N (Y_{it} - \E[Y_t|Z_{it}]) f_V(V_i) \1( (\ubx_t'\beta_0, V_i) \in \mathcal{Z}_t)\\
       & \cdot e_1'  S_N(\ubx_t'\beta_0,V_i;\beta_0)^{-1}  \int \mathcal{K}\left(\frac{X_{it}'\beta_0 - \ubx_t'\beta_0}{b_N}, v\right)\xi\left(\frac{X_{it}'\beta_0 - \ubx_t'\beta_0}{b_N}, v\right) \, dv\left(1 + O_p\left(\left(\frac{\log(N)}{Nb_N^{d_V}}\right)^{1/2}\right)\right)\\
       &= \frac{-1}{Nb_N} \sum_{i=1}^N (Y_{it} - \E[Y_{t}|Z_{it}]) f_V(V_i) \1( (\ubx_t'\beta_0,V_i) \in \mathcal{Z}_t)\\
       & \cdot  e_1'  S_N(\ubx_t'\beta_0,V_i;\beta_0)^{-1}  \int \mathcal{K}\left(\frac{X_{it}'\beta_0 - \ubx_t'\beta_0}{b_N}, v\right)\xi\left(\frac{X_{it}'\beta_0 - \ubx_t'\beta_0}{b_N}, v\right) \, dv + o_p(1).
\end{align*}

We now calculate the asymptotic variance of \footnotesize
       \[ \frac{-1}{Nb_N} \sum_{i=1}^N (Y_{it} - \E[Y_{t}|Z_{it}]) f_V(V_i) \1( (\ubx_t'\beta_0,V_i) \in \mathcal{Z}_t) \cdot  e_1' S_N(\ubx_t'\beta_0, V_i;\beta_0)^{-1} \int \mathcal{K}\left(\frac{X_{it}'\beta_0 - \ubx_t'\beta_0}{b_N}, v\right)\xi\left(\frac{X_{it}'\beta_0 - \ubx_t'\beta_0}{b_N}, v\right) \, dv .\]
\normalsize
We have that \footnotesize
\begin{align*}
       &\Var\left(\frac{-1}{Nb_N} \sum_{i=1}^N (Y_t - \E[Y_t|Z_t]) f_V(V) \1( (\ubx_t'\beta_0,V) \in \mathcal{Z}_t) \cdot  e_1' S_N(\ubx_t'\beta_0, V; \beta_0)^{-1}  \int \mathcal{K}\left(\frac{X_{t}'\beta_0 - \ubx_t'\beta_0}{b_N}, v\right) \xi\left(\frac{X_t'\beta_0 - \ubx_t'\beta_0}{b_N}, v \right) \, dv  \right)\\
       &= \frac{1}{N b_N^2} \E\left[ (Y_t - \E[Y_t|Z_t])^2  f_V(V)^2 \1((\ubx_t'\beta_0,V) \in \mathcal{Z}_t) e_1'  S_N(\ubx_t'\beta_0,V;\beta_0)^{-1}  \left(\int \mathcal{K} \left(\frac{X_t'\beta_0 - \ubx_t'\beta_0}{b_N}, v \right) \xi \left(\frac{X_t'\beta_0 - \ubx_t'\beta_0}{b_N}, v \right) \, dv\right) \right.\\
       &\left. \left(\int \mathcal{K} \left( \frac{X_t'\beta_0 - \ubx_t'\beta_0}{b_N}, v \right) \xi \left(\frac{X_t'\beta_0 - \ubx_t'\beta_0}{b_N}, v \right) \, dv\right)' S_N(\ubx_t'\beta_0,V;\beta_0)^{-1} e_1 \right]
\end{align*}\normalsize
Recall that $S_N(z;\beta_0) = S(z;\beta_0) + o_p(1) = \int \xi(a)\xi(a)' \mathcal{K}(a) \, da \cdot f_{Z_t(\beta_0)}(z) + o_p(1)$ uniformly in $z \in \mathcal{Z}_t$ by Lemma \ref{lemma:lem1b}. Therefore, the above expression \footnotesize
\begin{align*}
               &= \frac{1}{N b_N^2} \E\left[ \Var(Y_t|Z_t(\beta_0))  \frac{f_V(V)^2}{f_{Z_t(\beta_0)}(\ubx_t'\beta_0,V)^2} \1((\ubx_t'\beta_0,V) \in \mathcal{Z}_t) e_1'  \left(\int \xi(a)\xi(a)'\mathcal{K}(a) \, da\right)^{-1}   \right.\\
       & \left(\int \mathcal{K} \left(\frac{X_t'\beta_0 - \ubx_t'\beta_0}{b_N}, v \right) \xi \left(\frac{X_t'\beta_0 - \ubx_t'\beta_0}{b_N}, v \right) \, dv\right) \left. \cdot \left(\int \mathcal{K} \left( \frac{X_t'\beta_0 - \ubx_t'\beta_0}{b_N}, v \right) \xi \left(\frac{X_t'\beta_0 - \ubx_t'\beta_0}{b_N}, v \right) \, dv\right)' \right.\\
       &\left. \left(\int \xi(a)\xi(a)'\mathcal{K}(a) \, da\right)^{-1} e_1 \right] + o((Nb_N)^{-1})\\
       &= \frac{1}{N b_N^2} \E\left[ \int \Var(Y_t| X_t'\beta_0 = \tilde{u}, V)  \frac{f_V(V)^2}{f_{Z_t(\beta_0)}(\ubx_t'\beta_0,V)^2} \1((\ubx_t'\beta_0,V) \in \mathcal{Z}_t) e_1'  \left(\int \xi(a)\xi(a)' \mathcal{K}(a) \, da\right)^{-1}  \ \right.\\
       &\left(\int \mathcal{K} \left(\frac{\tilde{u} - \ubx_t'\beta_0}{b_N}, v \right) \xi \left(\frac{\tilde{u} - \ubx_t'\beta_0}{b_N}, v \right) \, dv\right)  \cdot \left(\int \mathcal{K} \left( \frac{\tilde{u} - \ubx_t'\beta_0}{b_N}, v \right) \xi \left(\frac{\tilde{u} - \ubx_t'\beta_0}{b_N}, v \right) \, dv\right)' \\
       &\left. \left(\int \xi(a)\xi(a)'\mathcal{K}(a) \, da\right)^{-1} e_1 \, f_{X_t'\beta_0|V}(\tilde{u}|V) \, d\tilde{u}\right] + o((Nb_N)^{-1})\\
       &= \frac{1}{N b_N} \E\left[ \int \Var(Y_t|X_t'\beta_0 = \ubx_t'\beta_0 + b_N u,V)  \frac{f_V(V)^2}{f_{Z_t(\beta_0)}(\ubx_t'\beta_0,V)^2} \1((\ubx_t'\beta_0,V) \in \mathcal{Z}_t) e_1'  \left(\int \xi(a)\xi(a)'\mathcal{K}(a) \, da\right)^{-1}   \right.\\
       &\left. \cdot \left(\int \mathcal{K} \left(z\right) \xi \left(z \right) \, dv\right)\left(\int \mathcal{K} \left( z\right) \xi \left(z\right) \, dv\right)' \left(\int \xi(a)\xi(a)'\mathcal{K}(a) \, da\right)^{-1} \, e_1 f_{X_t'\beta_0|V}(\ubx_t'\beta_0 + b_N u|V) \, du\right] + o((Nb_N)^{-1})\\
       &= \frac{1}{N b_N} \E\left[ \Var(Y_t|X_t'\beta_0 = \ubx_t'\beta_0,V)  \frac{f_V(V)}{f_{Z_t(\beta_0)}(\ubx_t'\beta_0,V)} \1((\ubx_t'\beta_0,V) \in \mathcal{Z}_t) \right]  \\
       & \cdot e_1'  \left(\int \xi(a)\xi(a)'\mathcal{K}(a) \, da\right)^{-1}  \int \left(\int \mathcal{K} \left(z\right) \xi \left(z \right) \, dv\right) \left(\int \mathcal{K} \left( z\right) \xi \left(z\right) \, dv\right)'   du \left(\int \xi(a)\xi(a)'\mathcal{K}(a) \, da\right)^{-1} e_1  + o((Nb_N)^{-1})\\
       &= \frac{1}{Nb_N} \sigma^2_{\text{ASF}_t}(\ubx_t'\beta_0)  + o((Nb_N)^{-1}).
\end{align*}\normalsize
The third equality follows from the change of variables $\tilde{u} = \ubx_t'\beta_0 + b_N u$. The above equations re-derive and fix a minor typo in equation (A.42) in KLX. By the proof of Corollary 2 in KLX, we have that
\begin{align*}
       \sqrt{Nb_N} \avgij \phi_{1,jN}(\ubx_t'\beta_0,V_i)\pi_{it} &\dconv \mathcal{N}(0,\sigma^2_{\text{ASF}_t}(\ubx_t'\beta_0)).
\end{align*}

Also, the term $\sqrt{Nb_N} \avg R_{1,N}(\ubx_t'\beta_0,V_i)\pi_{it}$ is $o_p(1)$ because
\begin{align*}
       \left|\sqrt{Nb_N} \avg R_{1,N}(\ubx_t'\beta_0,V_i)\pi_{it}\right| &\leq \sqrt{Nb_N} \cdot\sup_{z\in\mathcal{Z}_t} \left|R_{1,N}(z)\right|\\
       &= \sqrt{Nb_N} \cdot O_p\left(\frac{\log(N)}{N b_N^{1+d_V}}\right)\\
       &= O_p\left(\frac{\log(N)}{\sqrt{N b_N^{1 + 2d_V}}}\right)\\
       &= o_p(1)
\end{align*}
by B\ref{assn:ASFrates}.

Third, term \eqref{eq:lemma8:2} above is of order $O_p(\sqrt{b_N}) = o_p(1)$ by an application of the central limit theorem.

Finally, we obtain that 
\begin{align*}
       \sqrt{Nb_N} \left(\avg \widehat{h}_1(\ubx_t'\beta_0, V_i; \beta_0)\pi_{it} - \E[h_1(\ubx_t'\beta_0, V; \beta_0) \pi_t]\right) &= \sqrt{Nb_N} \avgij \phi_{jN}(\ubx_t'\beta_0, V_i) \pi_{it} + o_p(1)\\
       &\dconv \mathcal{N}(0,\sigma^2_{\text{ASF}_t}(\ubx_t'\beta_0)).
\end{align*}    
\end{proof}

We use the following technical lemma in the proof of Theorem \ref{thm:ASF}. 
\begin{lemma}\label{lemma:matrix_ineq}
Let $A$ and $B$ be positive-definite, symmetric matrices. Let $\lambda_{\min}(A)$ denote the minimum eigenvalue of $A$. Then,
\begin{align*}
       |\lambda_{\min}(A) - \lambda_{\min}(B)| &\leq \|A - B\|.
\end{align*}
\end{lemma}

\begin{proof} [Proof of Lemma \ref{lemma:matrix_ineq}]
Since $A$ and $B$ are positive-definite and symmetric, they are invertible and $\lambda_{\min}(A) = \|A^{-1}\|^{-1} > 0$ and $\lambda_{\min}(B) = \|B^{-1}\|^{-1} > 0$. We then have
\begin{align*}
       |\lambda_{\min}(A) - \lambda_{\min}(B)| &= |\|A^{-1}\|^{-1} - \|B^{-1}\|^{-1}|\\
       &= |\|A^{-1}\| - \|B^{-1}\|| \cdot \frac{1}{\|A^{-1}\| \|B^{-1}\|}\\
       &\leq \|A^{-1} - B^{-1}\|\cdot \frac{1}{\|A^{-1}\| \|B^{-1}\|}\\
       &= \|B^{-1}(B-A)A^{-1}\|\cdot \frac{1}{\|A^{-1}\| \|B^{-1}\|}\\
       &\leq \|B^{-1}\| \|A-B\| \|A^{-1}\| \cdot \frac{1}{\|A^{-1}\| \|B^{-1}\|}\\
       &= \|A - B\|.
\end{align*}    
The first inequality follows from the triangle inequality, and the second inequality is from $\|C D\| \leq \|C\| \|D\|$ for the spectral norm and square matrices $C$ and $D$.
\end{proof}

\begin{proof}[Proof of Theorem \ref{thm:ASF}]
We have the following decomposition:
\begin{align}
       \sqrt{Nb_N} \left(\widehat{\text{ASF}}_t(\ubx_t) - \text{ASF}^\pi_t(\ubx_t) \right) &= \sqrt{Nb_N} \left(\avg \left(\widehat{h}_1(\ubx_t'\widehat{\beta}, V_i; \widehat{\beta}) - \widehat{h}_1(\ubx_t'\widehat{\beta}, V_i; \beta_0)\right)\hat\pi_{it}\right) \label{eq:thm1:1}\\
       &+ \sqrt{Nb_N} \left(\avg \left(\widehat{h}_1(\ubx_t'\widehat{\beta}, V_i; \beta_0) - \widehat{h}_1(\ubx_t'\beta_0, V_i; \beta_0)\right) \hat\pi_{it} \right) \label{eq:thm1:2}\\
       &+ \sqrt{Nb_N} \left(\avg \widehat{h}_1(\ubx_t'\beta_0, V_i; \beta_0)(\hat\pi_{it} - \pi_{it})\right) \label{eq:thm1:3}\\
       &+ \sqrt{Nb_N} \left(\avg \widehat{h}_1(\ubx_t'\beta_0, V_i; \beta_0)\pi_{it} - \E[h_1(\ubx_t'\beta_0, V; \beta_0) \pi_t]\right). \label{eq:thm1:4}
\end{align}

We break down the proof into four parts. In the first three parts, we show that terms \eqref{eq:thm1:1}--\eqref{eq:thm1:3} are $o_p(1)$. In the fourth and last part, we show that term \eqref{eq:thm1:4} converges in distribution.

\noindent\textbf{Part 1: Convergence of Term \eqref{eq:thm1:1}}

We have that
\begin{align*}
       &\sqrt{Nb_N} \cdot  \left| \avg \left(\widehat{h}_1(\ubx_t'\widehat{\beta}, V_i;\widehat{\beta}) - \widehat{h}_1(\ubx_t'\widehat{\beta}, V_i; \beta_0)\right)\hat\pi_{it} \right|\\
       &= \sqrt{Nb_N} \cdot \left|\avgj e_1'\left(S_N(\ubx_t'\widehat{\beta}, V_i;\widehat{\beta})^{-1} T_N(\ubx_t'\widehat{\beta}, V_i; \widehat{\beta}) - S_N(\ubx_t'\widehat{\beta}, V_i; \beta_0)^{-1} T_N(\ubx_t'\widehat{\beta}, V_i; \beta_0)\right) \hat\pi_{it}  \right|\\
       &= \sqrt{Nb_N} \cdot \left|\avgj e_1'\left(S_N(\ubx_t'\widehat{\beta}, V_i;\widehat{\beta})^{-1} (T_N(\ubx_t'\widehat{\beta}, V_i; \widehat{\beta}) - T_N(\ubx_t'\widehat{\beta}, V_i; \beta_0))\right.\right.\\
       &\left.\left. + S_N(\ubx_t'\widehat{\beta}, V_i; \widehat{\beta})^{-1}\left(S_N(\ubx_t'\widehat{\beta}, V_i; \beta_0) - S_N(\ubx_t'\widehat{\beta}, V_i; \widehat{\beta})\right)S_N(\ubx_t'\widehat{\beta}, V_i; \beta_0)^{-1} T_N(\ubx_t'\widehat{\beta}, V_i; \beta_0)\right) \1((\ubx_t'\widehat{\beta}, V_i) \in \mathcal{Z}_t) \right|\\
       &\leq \sqrt{Nb_N} \cdot \|e_1\| \sup_{z\in\mathcal{Z}_t} \left\|S_N(z;\widehat{\beta})^{-1} \right\| \sup_{z\in\mathcal{Z}_t} \left\|T_N(z; \widehat{\beta}) - T_N(z;\beta_0)\right\|\\
       & + \sqrt{Nb_N} \cdot \|e_1\| \sup_{z\in\mathcal{Z}_t}\left\|S_N(z;\widehat{\beta})^{-1} \right\| \sup_{z\in\mathcal{Z}_t} \left\| S_N(z; \widehat{\beta}) - S_N(z; \beta_0) \right\| \sup_{z\in\mathcal{Z}_t} \left\| S_N(z; \beta_0)^{-1} \right\| \sup_{z \in \mathcal{Z}_t} \left\| T_N(z;\beta_0) \right\|.
\end{align*}

The terms in the previous expressions are of these asymptotic orders:
\begin{itemize}
       \item $\|e_1\| = 1$.
       \item $\left\| S_N(z;\widehat{\beta})^{-1} \right\| = \lambda_{\min}\left(S_N(z;\widehat{\beta})\right)^{-1}$, where $\lambda_{\min}(\cdot)$ denotes the minimum eigenvalue of a symmetric matrix. We have that
       \begin{align*}
               \sup_{z\in\mathcal{Z}_t} &\left| \lambda_{\min} \left(S_N(z; \widehat{\beta})\right) - \lambda_{\min} \left(S(z;\beta_0)\right)\right| \leq \sup_{z \in \mathcal{Z}_t}\left\|S_N(z; \widehat{\beta}) - S(z; \beta_0)\right\|\\
               &\leq \sup_{z \in \mathcal{Z}_t} \left\|S_N(z; \widehat{\beta}) - S_N(z; \beta_0)\right\| + \sup_{z \in \mathcal{Z}_t} \left\|S_N(z; \beta_0) - S(z; \beta_0)\right\|\\
               &= o_p\left(\frac{1}{\sqrt{Nb_N}}\right) + O_p\left(\left(\frac{\log(N)}{Nb_N^{1+d_V}}\right)^{1/2}\right) + O(b_N)\\
               &= o_p(1).
       \end{align*}
      The first line follows from Lemma \ref{lemma:matrix_ineq}. The second line follows from the triangle inequality. The third line follows from Lemmas \ref{lemma:lem1} and \ref{lemma:lem1b}. The last line follows from B\ref{assn:ASFrates}. Also note that
       \begin{align*}
               \inf_{z \in \mathcal{Z}_t} \lambda_{\min} \left(S(z; \beta_0)\right) &= \inf_{z \in \mathcal{Z}_t} f_{Z_t}(z) \cdot \lambda_{\min} \left(\int \xi(a)\xi(a)'\mathcal{K}(a) \, da\right) > 0.
               \end{align*}
This follows from the definition of the set $\mathcal{Z}_t$, which is such that $\inf_{z \in \mathcal{Z}_t} f_{Z_t}(z)>0$: see B\ref{assn:smoothness}.(ii). $\int \xi(a)\xi(a)'\mathcal{K}(a) \, da$ is positive definite since, for $c \in \R^{\bar{N}}$ such that $c \neq \textbf{0}$,
               \begin{align*}
               c' \left(\int \xi(a)\xi(a)'\mathcal{K}(a) \, da \right) c &= \int (c'\xi(a))^2 \mathcal{K}(a) \, da = 0
               \end{align*}
implies that $c' \xi(a) = 0$ for all $a$ in the support of $\mathcal{K}(a)$. Since $\xi(a)$ is comprised of products of powers of components of $a$, $c'\xi(a) = 0$ over this entire support implies $c = \textbf{0}$, a contradiction. Therefore $\lambda_{\min} \left(\int \xi(a)\xi(a)'\mathcal{K}(a) \, da\right) > 0$ and $\inf_{z \in \mathcal{Z}_t} \lambda_{\min} \left(S(z; \beta_0)\right) > 0$.

This implies that,
       \begin{align*}
               \sup_{z \in \mathcal{Z}_t} \left\| S_N(z; \widehat{\beta})^{-1} \right\| &= \frac{1}{\inf_{z \in \mathcal{Z}_t} \lambda_{\min} \left(S_N(z; \widehat{\beta}) \right)}\\
               &\leq \frac{1}{\inf_{z \in \mathcal{Z}_t} \lambda_{\min} \left(S(z; \beta_0)\right) - \sup_{z \in \mathcal{Z}_t}\left|\lambda_{\min} \left(S_N(z; \widehat{\beta})\right) - \lambda_{\min} \left(S(z; \beta_0)\right)\right|}\\
               &= \frac{1}{\inf_{z \in \mathcal{Z}_t} \lambda_{\min}\left(S(z; \beta_0)\right) -o_p(1)}\\
               &= O_p(1).
       \end{align*}
       \item By Lemma \ref{lemma:lem1}, we have that $\sup_{z \in \mathcal{Z}_t} \left\|S_N(z; \widehat{\beta}) - S_N(z; \beta_0) \right\| = o_p\left(\frac{1}{\sqrt{Nb_N}}\right)$.
       \item By Lemma \ref{lemma:lem2}, we have that $\sup_{z \in \mathcal{Z}_t}\left\| T_N(z; \widehat{\beta}) - T_N(z; \beta_0) \right\| = o_p\left(\frac{1}{\sqrt{Nb_N}}\right)$
       \item As above, we have that $\sup_{z\in\mathcal{Z}_t} \left\| S_N(z;\beta_0)^{-1} \right\| = O_p(1)$.
       \item We have that
       \begin{align*}
               \sup_{z \in \mathcal{Z}_t} \| T_N(z; \beta_0) \| &\leq \sup_{z \in \mathcal{Z}_t} \| T_N(z; \beta_0) -  T(z; \beta_0)\| + \sup_{z \in \mathcal{Z}_t} \| T(z; \beta_0) \|
       \end{align*}
       where
       \begin{align*}
               \sup_{z \in \mathcal{Z}_t} \| T_N(z; \beta_0) - T(z; \beta_0) \| &= O_p\left(\left(\frac{\log(N)}{Nb_N^{1+d_V}}\right)^{1/2}\right) + O(b_N)
       \end{align*}
       by Lemma \ref{lemma:lem2b}. We also have that
       \begin{align*}
               \sup_{z \in \mathcal{Z}_t} \| T(z; \beta_0) \| &= \sup_{z \in \mathcal{Z}_t} \left|\E[Y_t| Z_t = z] f_{Z_t}(z) \right| \cdot \left\| \int \xi(a) \mathcal{K}(a)\, da \right\|\\
               &\leq  \sup_{z \in \mathcal{Z}_t} \left|\E[Y_t| Z_t = z]\right| \cdot \sup_{z \in \mathcal{Z}_t}  f_{Z_t}(z) \cdot O(1)\\
               &= O(1)
       \end{align*}
       by $\sup_{z \in \mathcal{Z}_t} \left|\E[Y_t| Z_t = z]\right| <\infty$ (Assumption B\ref{assn:smoothness}.(iii)), $\sup_{z\in\mathcal{Z}_t}  f_{Z_t}(z) < \infty$ (Assumption B\ref{assn:smoothness}.(iii)), and $\left\| \int \xi(a) \mathcal{K}(a)\, da \right\| <\infty$ (Assumption B\ref{assn:kernel}). Therefore, 
       \begin{align*}
               \sup_{z \in \mathcal{Z}_t} \|T_N(z;\beta_0)\| &= O_p\left(\left(\frac{\log(N)}{Nb_N^{1+d_V}}\right)^{1/2}\right) + O(b_N) + O(1)\\
               &= O_p(1),
       \end{align*}
       by B\ref{assn:ASFrates}.
\end{itemize}

Combining the asymptotic orders of the above six terms, we have
\begin{align*}
       \sqrt{Nb_N} \cdot  \left| \avg \left(\widehat{h}_1(\ubx_t'\widehat\beta, V_i;\widehat{\beta}) - \widehat{h}_1(\ubx_t'\widehat{\beta}, V_i; \beta_0)\right) \hat\pi_{it} \right|
       &\leq \sqrt{Nb_N} \cdot O_p(1) \cdot o_p\left(\frac{1}{\sqrt{Nb_N}}\right)\\
       & + \sqrt{Nb_N} \cdot O_p(1)\cdot o_p\left(\frac{1}{\sqrt{Nb_N}}\right) \cdot O_p(1) \cdot O_p(1)\\
       &= o_p(1).
\end{align*} 

\noindent\textbf{Part 2: Convergence of Term \eqref{eq:thm1:2}}

We have that 
\begin{align}
       &\left|\avg \left(\widehat{h}_1(\ubx_t'\widehat\beta, V_i; \beta_0) - \widehat{h}_1(\ubx_t'\beta_0, V_i; \beta_0)\right)\hat\pi_{it}\right| \notag \\
       &= \left|\avgj e_1'\left(S_N(\ubx_t'\widehat\beta, V_i; \beta_0)^{-1} T_N(\ubx_t'\widehat\beta, V_i; \beta_0) - S_N(\ubx_t'\beta_0, V_i; \beta_0)^{-1} T_N(\ubx_t'\beta_0, V_i; \beta_0)\right) \hat\pi_{it}\right| \notag \\
       &= \left|\avgj e_1'\left(S_N(\ubx_t'\widehat\beta, V_i; \beta_0)^{-1}(T_N(\ubx_t'\widehat\beta, V_i; \beta_0) - T_N(\ubx_t'\beta_0, V_i; \beta_0))\right.\right. \notag \\
       &\left.\left. + S_N(\ubx_t'\widehat\beta,V_i; \beta_0)^{-1} \left(S_N(\ubx_t'\beta_0,V_i; \beta_0) - S_N(\ubx_t'\widehat\beta, V_i; \beta_0) \right) S_N(\ubx_t'\beta_0, V_i; \beta_0)^{-1}T_N(\ubx_t'\beta_0,V_i; \beta_0)\right) \hat\pi_{it}\right| \notag \\
       &\leq \|e_1\| \sup_{z \in \mathcal{Z}_t} \left\| S_N(z; \beta_0)^{-1} \right\| \sup_{z \in \mathcal{Z}_t} \left\| \frac{\partial}{\partial u} T_N(z;\beta_0)\right\| \left\|\ubx_t'\widehat\beta- \ubx_t'\beta_0\right\| \notag \\
       & + \|e_1\| \sup_{z \in \mathcal{Z}_t} \left\|S_N(z; \beta_0)^{-1}\right\| \sup_{z \in \mathcal{Z}_t} \left\|\frac{\partial}{\partial u} S_N(z; \beta_0)\right\| \left\|\ubx_t'\widehat\beta - \ubx_t'\beta_0\right\| \sup_{z \in \mathcal{Z}_t} \left\| S_N(z; \beta_0)^{-1} \right\| \sup_{z \in \mathcal{Z}_t} \left\| T_N(z; \beta_0)\right\|. \label{eq:thm1:5}
\end{align}
The inequality follows from applications of the mean-value theorem and the Cauchy-Schwarz inequality. By Lemmas \ref{lemma:lem2} and \ref{lemma:lem3}, 
\begin{align*}
       \sup_{z \in \mathcal{Z}_t} \left\|\frac{\partial}{\partial u} S_N(z; \beta_0) \right\| &= o_p\left(\frac{a_N}{\sqrt{Nb_N}}\right)\\
       \sup_{z \in \mathcal{Z}_t} \left\|\frac{\partial}{\partial u} T_N(z; \beta_0) \right\| &= o_p\left(\frac{a_N}{\sqrt{Nb_N}}\right).
\end{align*}

By B\ref{assn:betahat}, $\|\ubx_t'\widehat\beta - \ubx_t'\beta_0\| \leq \|\ubx_t\| \|\widehat{\beta} - \beta_0\| = O_p(a_N^{-1})$. The asymptotic orders of all other terms in equation \eqref{eq:thm1:5} were characterized in the analysis of the convergence of term \eqref{eq:thm1:1}. Therefore
\begin{align*}
       &\sqrt{Nb_N} \cdot \left|\avg \left(\widehat{h}_1(\ubx_t'\widehat\beta, V_i; \beta_0) - \widehat{h}_1(\ubx_t'\beta_0, V_i; \beta_0)\right) \hat\pi_{it}\right|\\
       &=\sqrt{Nb_N} \cdot O_p(1) \cdot o_p\left(\frac{a_N}{\sqrt{Nb_N}}\right) \cdot O_p(a_N^{-1}) + \sqrt{Nb_N}\cdot O_p(1) \cdot o_p\left(\frac{a_N}{\sqrt{Nb_N}}\right) \cdot O_p(a_N^{-1}) \cdot O_p(1) \cdot O_p(1)\\
       &= o_p(1).
\end{align*}

\noindent\textbf{Part 3: Convergence of Term \eqref{eq:thm1:3}}

First note that
\begin{align*}
       \left|\avg \widehat{h}_1(\ubx_t'\beta_0, V_i; \beta_0)(\hat\pi_{it} - \pi_{it})\right| &\leq \avg \left|\widehat{h}_1(\ubx_t'\beta_0, V_i; \beta_0)\right| \cdot \sup_{i = 1,\ldots,N}\left|\hat\pi_{it} - \pi_{it}\right|.
\end{align*}
Therefore 
\begin{align*}
       \prob&\left(\sqrt{Nb_N} \left|\avg \widehat{h}_1(\ubx_t'\beta_0, V_i; \beta_0)(\hat\pi_{it} - \pi_{it}) \right| = 0\right)\\
       &\geq \prob\left( \avg \left|\widehat{h}_1(\ubx_t'\beta_0, V_i; \beta_0)\right| \cdot \sup_{i = 1,\ldots,N}\left| \hat\pi_{it} - \pi_{it} \right| = 0\right)\\
       &\geq \prob\left( \sup_{i = 1,\ldots,N} \left|\hat\pi_{it} - \pi_{it}\right| = 0\right)\\
       &\to 1
\end{align*}
as $N\to \infty$ by Lemma \ref{lemma:lem5}. Therefore
\begin{align*}
       \sqrt{Nb_N} \left|\avg \widehat{h}_1(\ubx_t'\beta_0, V_i; \beta_0)(\hat\pi_{it} - \pi_{it})\right| &= o_p(1)
\end{align*}

\noindent\textbf{Part 4: Convergence of Term \eqref{eq:thm1:4}}

By Lemma \ref{lemma:convASF}, this term converges in distribution:
\begin{align*}
       \sqrt{Nb_N}\left(\avg \widehat{h}_1(\ubx_t'\beta_0, V_i; \beta_0)\pi_{it} - \E[h_1(\ubx_t'\beta_0, V; \beta_0)\pi_t]\right) &\dconv \mathcal{N}(0,\sigma^2_{\text{ASF}_t}(\ubx_t'\beta_0)).
\end{align*}

The conclusion follows from an application of Slutsky's Theorem.
\end{proof}

\begin{lemma}[APE convergence in distribution]\label{lemma:convAPE}
Suppose B\ref{assn:iid}--B\ref{assn:ASFrates} hold. Then,
\begin{align*}
       \sqrt{Nb_N^3} \left(\avg \widehat{h}_2(\ubx_t'\beta_0, V_i; \beta_0)\pi_{it} - \E[h_2(\ubx_t'\beta_0, V; \beta_0) \pi_t]\right) &\dconv \mathcal{N}(0,\sigma^2_{\text{APE}_t}(\ubx_t'\beta_0)).
\end{align*}
\end{lemma}

\begin{proof}[Proof of Lemma \ref{lemma:convAPE}]
    This proof builds on that of Corollary 2 in KLX and our Lemma \ref{lemma:convASF}. Recall that Assumptions A1--A7 of KLX hold under ours. We can then use equation (13) in KLX and their Corollary 1 to write
\begin{align*}
       b_N \widehat{h}_2(z; \beta_0) &= b_N h_2(z; \beta_0) + B_{2,N}(z) + \avgj \phi_{2,jN}(z) + R_{2,N}(z)\\
       &= e_{2+d_V}'h(z; \beta_0)+ B_{2,N}(z) + \avgj \phi_{2,jN}(z) + R_{2,N}(z),
\end{align*}
where $B_{2,N}(z)$ is a bias term satisfying $\sup_{z\in\mathcal{Z}_t} |B_{2,N}(z)| = O(b_N^{\ell+1})$ if $\ell$ is odd or $O(b_N^{\ell+2})$ if $\ell$ is even, where $\phi_{2,jN}(z)$ are mean-zero random variables, and where $R_{2,N}(z)$ is a higher-order term satisfying $\sup_{z \in \mathcal{Z}_t} |R_{2,N}(z)| = O_p\left(\frac{\log(N)}{Nb_N^{1+d_V}}\right)$. 
       
Second, note that
\begin{align}
       \sqrt{Nb_N^3} &\left(\avg \widehat{h}_2(\ubx_t'\beta_0, V_i; \beta_0)\pi_{it} - \E[h_2(\ubx_t'\beta_0, V; \beta_0) \pi_t]\right)\notag\\
       &= \sqrt{Nb_N^3} \avg \left(\widehat{h}_2 (\ubx_t'\beta_0 ,V_i; \beta_0) - h_2(\ubx_t'\beta_0, V_i; \beta_0) \right) \pi_{it} \label{eq:lemma9:1}\\
       &+ \sqrt{b_N^3} \cdot \frac{1}{\sqrt{N}} \sum_{i=1}^N \left(h_2(\ubx_t'\beta_0, V_i; \beta_0)\pi_{it} - \E[h_2(\ubx_t'\beta_0, V; \beta_0)\pi_t]\right)\label{eq:lemma9:2}
\end{align}

To analyze term \eqref{eq:lemma9:1}, we use the fact that
\begin{align*}
       &\sqrt{Nb_N^3}\avg \left(\widehat{h}_2(\ubx_t'\beta_0, V_i; \beta_0) - h_2 (\ubx_t'\beta_0, V_i; \beta_0) \right) \pi_{it} \\
       &= \sqrt{Nb_N} \avg e_{2+d_V}'\left(\widehat{h}(\ubx_t'\beta_0, V_i; \beta_0) - h(\ubx_t'\beta_0, V_i; \beta_0) \right) \pi_{it} \\
       &= \sqrt{Nb_N} \avg B_{2,N}(\ubx_t'\beta_0, V_i)\pi_{it} + \sqrt{Nb_N} \avgij \phi_{2,jN}(\ubx_t'\beta_0,V_i) \pi_{it} + \sqrt{Nb_N} \avg R_{2,N}(\ubx_t'\beta_0, V_i) \pi_{it}.
\end{align*}

The terms $\sqrt{Nb_N} \avg B_{2,N}(\ubx_t'\beta_0, V_i) \pi_{it}$ and $\sqrt{Nb_N} \avg R_{2,N}(\ubx_t'\beta_0, V_i)\pi_{it}$ are $o_p(1)$ from the same arguments used in the proof of Lemma \ref{lemma:convASF}.

The term $\sqrt{Nb_N} \avgij \phi_{2,jN}(\ubx_t'\beta_0, V_i)\pi_{it}$ converges in distribution to
\begin{align*}
       \sqrt{Nb_N} \avgij \phi_{2,jN}(\ubx_t'\beta_0, V_i)\pi_{it} &\dconv \mathcal{N}(0,\sigma^2_{\text{APE}_t}(\ubx_t'\beta_0))
\end{align*}
by standard arguments from \cite{Masry1996} referred to in the proof of Corollary 2 in KLX. 

Term \eqref{eq:lemma9:2} above is of order $O_p(b_N^{3/2}) = o_p(1)$ by an application of the central limit theorem. Therefore, 
\begin{align*}
       \sqrt{Nb_N^3} \left(\avg \widehat{h}_2(\ubx_t'\beta_0, V_i; \beta_0)\pi_{it} - \E[h_2(\ubx_t'\beta_0, V; \beta_0)\pi_t] \right) &= \sqrt{Nb_N} \avgij \phi_{2,jN}(\ubx_t'\beta_0 ,V_i) \pi_{it} + o_p(1)\\
       &\dconv \mathcal{N}(0,\sigma^2_{\text{APE}_t}(\ubx_t'\beta_0)).
\end{align*}    
\end{proof}

\begin{proof}[Proof of Theorem \ref{thm:APE}]
First, we write
\begin{align}
       \sqrt{Nb_N^3} &\left(\widehat{\text{APE}}_{k,t}(\ubx_t) - \text{APE}^\pi_{k,t}(\ubx_t) \right)\notag\\
       &=
       \widehat{\beta}^{(k)} \cdot \sqrt{Nb_N^3} \left(\avg \left(\widehat{h}_2 (\ubx_t'\widehat{\beta}, V_i; \widehat{\beta}) - \widehat{h}_2(\ubx_t'\widehat\beta, V_i; \beta_0)\right)\hat\pi_{it}\right)\label{eq:thm2:1}\\
       &+ \widehat{\beta}^{(k)} \cdot \sqrt{Nb_N^3} \left(\avg \left(\widehat{h}_2(\ubx_t'\widehat\beta, V_i; \beta_0) - \widehat{h}_2(\ubx_t'\beta_0, V_i; \beta_0)\right)\hat\pi_{it} \right) \label{eq:thm2:2}\\
       &+ \widehat{\beta}^{(k)} \cdot \sqrt{Nb_N^3} \left(\avg \widehat{h}_2(\ubx_t'\beta_0, V_i; \beta_0)(\hat\pi_{it} - \pi_{it})\right) \label{eq:thm2:3}\\
       &+ \widehat{\beta}^{(k)} \cdot \sqrt{Nb_N^3} \left(\avg \widehat{h}_2(\ubx_t'\beta_0, V_i; \beta_0)\pi_{it} - \E[h_2(\ubx_t'\beta_0, V; \beta_0)\pi_t]\right) \label{eq:thm2:4}\\
       &+ \sqrt{Nb_N^3}(\widehat\beta^{(k)} - \beta_0^{(k)}) \cdot \E[h_2(\ubx_t'\beta_0, V; \beta_0)\pi_t]. \label{eq:thm2:5}
\end{align}

We will show that terms \eqref{eq:thm2:1}--\eqref{eq:thm2:3} and \eqref{eq:thm2:5} are $o_p(1)$, and that term \eqref{eq:thm2:4} converges in distribution.

\noindent\textbf{Convergence of Term \eqref{eq:thm2:1}}

Note that
\begin{align*}
       &\sqrt{Nb_N^3} \cdot  \left| \avg \left(\widehat{h}_2(\ubx_t'\widehat\beta, V_i; \widehat\beta) - \widehat{h}_2(\ubx_t'\widehat\beta, V_i; \beta_0)\right) \hat\pi_{it} \right|\\
       &= \sqrt{Nb_N} \cdot \left|\avg e_{2+d_V}' \left(S_N(\ubx_t'\widehat\beta, V_i;\widehat\beta)^{-1} T_N(\ubx_t'\widehat\beta, V_i; \widehat\beta) - S_N(\ubx_t'\widehat\beta, V_i; \beta_0)^{-1} T_N(\ubx_t'\widehat\beta, V_i; \beta_0)\right) \hat\pi_{it} \right|
\end{align*}
by the definition of $\widehat{h}_2$. Also note that $\widehat\beta^{(k)} = O_p(1)$.  Therefore, we can follow the same steps used in the proof of Theorem \ref{thm:ASF} to show term \eqref{eq:thm1:1} is $o_p(1)$.

\noindent\textbf{Convergence of Term \eqref{eq:thm2:2}}

We have that
\begin{align*}
       \sqrt{Nb_N^3} &\cdot \left|\widehat\beta^{(k)} \avg \left(\widehat{h}_2(\ubx_t'\widehat\beta, V_i; \beta_0) - \widehat{h}_2(\ubx_t'\beta_0, V_i; \beta_0)\right) \hat\pi_{it} \right|\\
       &= \left|\widehat\beta^{(k)}\right| \cdot \sqrt{Nb_N}\\
       &\cdot \left|\avg e_{2+d_V}'\left(S_N(\ubx_t'\widehat\beta, V_i;\beta_0)^{-1} T_N(\ubx_t'\widehat\beta, V_i; \beta_0) - S_N(\ubx_t'\beta_0, V_i; \beta_0)^{-1}T_N(\ubx_t'\beta_0, V_i; \beta_0)\right) \hat\pi_{it}\right|.
\end{align*}

Again, we can follow the same steps used in the proof of Theorem \ref{thm:ASF} to show term \eqref{eq:thm1:2} is $o_p(1)$.

\noindent\textbf{Convergence of Term \eqref{eq:thm2:3}}

The convergence of this term is shown identically to that of term \eqref{eq:thm1:3}. 

\noindent\textbf{Convergence of Term \eqref{eq:thm2:4}}

By Lemma \ref{lemma:convAPE}, we have that $\sqrt{Nb_N^3} \left(\avg \widehat{h}_2(\ubx_t'\beta_0, V_i; \beta_0)\pi_{it} - \E[h_2(\ubx_t'\beta_0, V; \beta_0)\pi_t]\right) \dconv \mathcal{N}(0,\sigma^2_{\text{APE}_t}(\ubx_t'\beta_0))$. By B\ref{assn:betahat}, $\widehat\beta^{(k)} \pconv \beta_0^{(k)}$. Therefore, by Slutsky's Theorem, term \eqref{eq:thm2:4} converges in distribution to a mean-zero Gaussian distribution with variance $(\beta_0^{(k)})^2 \cdot \sigma^2_{\text{APE}_t}(\ubx_t'\beta_0)$.

\noindent\textbf{Convergence of Term \eqref{eq:thm2:5}}

Note that $\E[h_2(\ubx_t'\beta_0, V; \beta_0)\pi_t] = O(1)$. Term \eqref{eq:thm2:5} is of order $\sqrt{N b_N^3}(\widehat\beta^{(k)} - \beta_0^{(k)})\cdot O(1) = O_p\left(\sqrt{N b_N^3} a_N^{-1}\right)$. By B\ref{assn:betahat}, the order of this term is
\begin{align*}
       O_p\left(N^{\frac{1}{2}(1 - 3\delta - 2\epsilon)} \right) &= o_p(1).
\end{align*}
This equality follows from $\delta > \frac{1 - 2\epsilon}{3}$, which can be seen from $\delta > 1 - 2\epsilon$ and $\delta > 0$: see B\ref{assn:ASFrates}.

Combining the convergence of terms \eqref{eq:thm2:1}--\eqref{eq:thm2:5} with Slutsky's Theorem, we obtain our result.
\end{proof}

\subsection{Proofs for Appendix \ref{sec:extensions}}\label{app_sec:ext-proofs}
\begin{proof}[Proof of Theorem \ref{thm:ID_weakexog}]
This proof is similar to the proofs of Theorems \ref{thm:ASF-APE_id} and \ref{thm:LAR-AME_id} in the main paper. For the ASF, note that
\begin{align*}
        \text{ASF}_t(\ubx_t) &= \E[g_t(\ubx_t'\beta_0,C,U_t)]\\
        &= \int_{\supp(V^t)} \E[g_t(\ubx_t'\beta_0,C,U_t)|V^t = v^t] \, dF_{V^t}(v^t)\\
        &= \int_{\supp(V^t|X_t'\beta_0 = \ubx_t'\beta_0)} \E[g_t(\ubx_t'\beta_0,C,U_t)|V^t = v^t] \, dF_{V^t}(v^t)\\
        &= \int_{\supp(V^t|X_t'\beta_0 = \ubx_t'\beta_0)} \E[g_t(\ubx_t'\beta_0,C,U_t)|X_t'\beta_0 = \ubx_t'\beta_0,V^t = v^t] \, dF_{V^t}(v^t)\\
        &= \int_{\supp(V^t|X_t'\beta_0 = \ubx_t'\beta_0)} \E[Y_t|X_t'\beta_0 = \ubx_t'\beta_0,V^t =v^t] \, dF_{V^t}(v^t).
\end{align*}
The second equality follows from the law of iterated expectations. The third follows from the support condition. The fourth follows from $(C,U_t) \indep X_t'\beta_0|V^t$, which can be shown similarly to step 1 in the proof of Theorem \ref{thm:ASF-APE_id}. Note that $U_t \indep X_t'\beta_0|C,V^t$ is implied by $U_t \indep (\X_{\text{exog}},\X^t_\text{pre})|C$ and by $X_t'\beta_0$ being a function of $(\X_{\text{exog}},\X^t_\text{pre})$. Also note that $C \indep X_t'\beta_0|V^t$, which follows from $C \indep (\X_{\text{exog}},\X^t_\text{pre})|V^t$ and from $X_t'\beta_0$ being a function of $(\X_{\text{exog}},\X^t_\text{pre})$. The last line follows directly. Finally, the last expression is identified from the distribution of $(Y,\X)$ using similar arguments as before.

Proofs for the identification of the APE, LAR, and AME proceed similarly.
\end{proof}

\clearpage

\singlespacing
\bibliographystyle{econometrica}
\bibliography{bibfile111}